\pdfoutput=1  
\documentclass[11pt]{article}

\PassOptionsToPackage{numbers, sort&compress}{natbib}

\usepackage[utf8]{inputenc}
\usepackage[T1]{fontenc}
\usepackage{lmodern}
\usepackage[margin=1in]{geometry}
\usepackage{amsmath,amssymb,amsfonts,amsthm}
\usepackage{booktabs}
\usepackage{float}
\usepackage{placeins}
\usepackage{tikz}
\usetikzlibrary{decorations.pathreplacing}
\usepackage{longtable}
\usepackage{array}
\usepackage{multirow}
\usepackage{microtype}
\usepackage{graphicx}
\usepackage{xcolor}
\usepackage{enumitem}
\usepackage{url}
\newcommand{\certpath}[1]{\textcolor{blue}{\path{#1}}}
\usepackage{mathtools}
\usepackage{caption}
\usepackage{authblk}
\usepackage{natbib}
\usepackage{setspace}
\usepackage{tocloft}

\usepackage{hyperref}
\hypersetup{
  colorlinks=true,
  linkcolor=blue,
  citecolor=blue,
  urlcolor=blue,
  pdftitle={New Bounds on the Competitive Ratio of Longest Queue Drop},
  pdfauthor={Alex Davydow and Sergey Nikolenko},
  pdfkeywords={Longest Queue Drop, shared memory switches, online algorithms, competitive analysis}
}

\newtheorem{theorem}{Theorem}[section]
\newtheorem{proposition}[theorem]{Proposition}
\newtheorem{lemma}[theorem]{Lemma}
\newtheorem{corollary}[theorem]{Corollary}

\theoremstyle{definition}
\newtheorem{definition}[theorem]{Definition}

\newtheorem{remark}[theorem]{Remark}

\newcommand{\OPT}{\ensuremath{\mathrm{OPT}}}
\newcommand{\LQD}{\ensuremath{\mathrm{LQD}}}

\newcommand{\LateQD}{\ensuremath{\mathrm{LateQD}}}
\newcommand{\CR}{\ensuremath{\mathrm{CR}}}
\newcommand{\rhostar}{\ensuremath{\varrho^{*}}}

\newcommand{\OPTX}{\ensuremath{\mathrm{OPT}_{\mathrm{EXTRA}}}}
\newcommand{\LQDX}{\ensuremath{\mathrm{LQD}_{\mathrm{EXTRA}}}}

\newcommand{\front}{\mathrm{fl}}  
\newcommand{\eps}{\varepsilon}

\title{New Bounds on the Competitive Ratio of Longest Queue Drop: $1.46929591 \le \CR(\mathrm{LQD}) \le 1.683652$}

\author[1]{Alex Davydow}
\author[2,3]{Sergey Nikolenko}
\affil[1]{Independent Scholar, Panorama, Greece}
\affil[2]{St.\ Petersburg Department of the Steklov Institute of Mathematics, St.\ Petersburg, Russia, \texttt{sergey@logic.pdmi.ras.ru}}
\affil[3]{St.\ Petersburg State University}

\date{}

\begin{document}
\maketitle

\begin{abstract}
We study the competitive ratio of \emph{Longest Queue Drop} (LQD), the canonical buffer management policy for shared memory switches, for which the previously published bounds were $\CR(\LQD)\in[1.44546086,\,1.6918]$ \citep{BDGN19,AEMV24}. We improve both ends. For the lower bound, we provide an exact-integer certificate on a new instance family, the \emph{front-loaded family} (Section~\ref{sec:lower}), which gives
\[
\CR(\LQD)\;\ge\;
\frac{184{,}815{,}365{,}566{,}285}{125{,}784{,}985{,}866{,}185}
\;=\;1.46929591\ldots,
\]
the exact ratio on a specific finite instance evaluated against the exactly optimal offline policy. The instance is evaluated under one fixed tie rule, but the bound does not depend on the tie rule: an adaptive coupling transfers the certified value to every non-clairvoyant \emph{deterministic} tie rule, and to every randomized tie rule in the adaptive adversary sense (the instance may depend on the coins already flipped). 

For the upper bound, we prove that $\CR(\LQD)\le1.683652$ by replacing the per-packet endpoint relaxation in the endgame of \citet{AEMV24} with the continuum envelope relaxation of the same payment expression, which we solve exactly in closed form; unlike the lower bound, this upper bound holds for \emph{every} tie rule. In the process, we find a gap in the derivation of the published proof's aggregation step (Lemma~18 from \cite{AEMV24}).
We repair the aggregation with one amortized lemma and an exact finite-head analysis; the repair restores the published $1.6918$, which in turn restores the weaker conference guarantee $1.707$ of \citet{AEMV21}, and it also supports our further improvement. 
\end{abstract}

\section{Introduction}\label{sec:intro}

A \emph{shared memory switch} is a network device with $N$ output ports
and one pool of $B$ packet slots shared among them.  Each output port
serves its own queue of stored packets, transmitting one packet per time
step, but the queues have no private memory, so a burst of traffic
aimed at one port competes for the very same $B$ slots as every other
port.  When arrivals would push the total past $B$, packets must be
discarded (an \emph{overflow}), and a discarded packet is lost for good;
the switch may also evict packets it has already stored
(\emph{push-out}). Deciding \emph{which} packets to discard is the
switch's only freedom, and the entire problem: a \emph{buffer management
policy} is a rule for that choice, and a good policy maximizes the total
number of packets transmitted.

What makes the choice hard is that it must be made \emph{online}: the
switch sees only the arrivals so far, yet it is judged against the
\emph{offline optimum} $\OPT$, a clairvoyant algorithm that sees the
whole arrival sequence in advance and discards optimally. The
\emph{competitive ratio} of a policy is the smallest asymptotic factor
by which $\OPT$'s throughput can exceed the policy's, allowing the
standard $O(B)$ additive term independent of the input length (see
Definition~\ref{def:switch}); it is always $\ge1$, and $1$ would mean
asymptotically matching clairvoyance.

\emph{Longest Queue Drop} (LQD) is a natural heuristic rule: on an
overflow, take the space back from whichever queue currently holds the
most, that is, discard from a currently longest queue. Proposed in 1991 \citep{WCH91}, it has
been the standard policy for this model ever since \citep{HKM01,AKM08},
and it remains the only policy known to be competitive for this problem
\citep{AEMV24}. Yet its exact competitive ratio $\CR(\LQD)$ is
unknown, and pinning it down is one of the oldest open problems in
buffer management \citep{Gol10,NK16}.

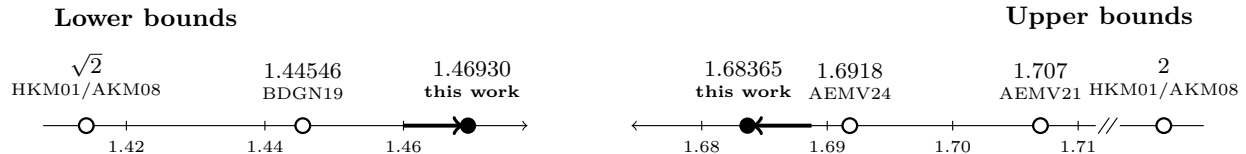
\begin{figure}[!t]
\centering
\resizebox{\linewidth}{!}{%
\begin{tikzpicture}[x=1cm,y=1cm,lbl/.style={font=\scriptsize},ttl/.style={font=\footnotesize\bfseries},oursmk/.style={line width=1.4pt}]
\begin{scope}
\node[ttl,anchor=west] at (0,1.35) {Lower bounds};
\draw[->] (0,0) -- (6.05,0);
\draw (1.039,0.06) -- (1.039,-0.06) node[below,font=\tiny]{1.42};
\draw (2.770,0.06) -- (2.770,-0.06) node[below,font=\tiny]{1.44};
\draw (4.501,0.06) -- (4.501,-0.06) node[below,font=\tiny]{1.46};
\filldraw[fill=white,thick] (0.538,0) circle(2.6pt);
\node[lbl,anchor=south,align=center] at (0.538,0.20){$\sqrt2$\\[-2pt]\tiny HKM01/AKM08};
\filldraw[fill=white,thick] (3.243,0) circle(2.6pt);
\node[lbl,anchor=south,align=center] at (3.243,0.20){$1.44546$\\[-2pt]\tiny BDGN19};
\draw[oursmk,->] (4.507,0) -- (5.247,0);
\fill (5.307,0) circle(2.9pt);
\node[lbl,anchor=south,align=center,font=\scriptsize\bfseries] at (5.357,0.22){$1.46930$\\[-2pt]\tiny\bfseries this work};
\end{scope}
\begin{scope}[xshift=7.6cm]
\node[ttl,anchor=east] at (6.90,1.35) {Upper bounds};
\draw[<-] (-0.25,0) -- (5.55,0);
\draw (5.85,0) -- (6.90,0);
\draw (5.58,-0.10) -- (5.72,0.10);
\draw (5.68,-0.10) -- (5.82,0.10);
\draw (0.627,0.06) -- (0.627,-0.06) node[below,font=\tiny]{1.68};
\draw (2.195,0.06) -- (2.195,-0.06) node[below,font=\tiny]{1.69};
\draw (3.762,0.06) -- (3.762,-0.06) node[below,font=\tiny]{1.70};
\draw (5.330,0.06) -- (5.330,-0.06) node[below,font=\tiny]{1.71};
\draw[oursmk,->] (1.999,0) -- (1.259,0);
\fill (1.199,0) circle(2.9pt);
\node[lbl,anchor=south,align=center,font=\scriptsize\bfseries] at (1.149,0.22){$1.68365$\\[-2pt]\tiny\bfseries this work};
\filldraw[fill=white,thick] (2.477,0) circle(2.6pt);
\node[lbl,anchor=south,align=center] at (2.477,0.20){$1.6918$\\[-2pt]\tiny AEMV24};
\filldraw[fill=white,thick] (4.859,0) circle(2.6pt);
\node[lbl,anchor=south,align=center] at (4.859,0.20){$1.707$\\[-2pt]\tiny AEMV21};
\filldraw[fill=white,thick] (6.40,0) circle(2.6pt);
\node[lbl,anchor=south,align=center] at (6.40,0.20){$2$\\[-2pt]\tiny HKM01/AKM08};
\end{scope}
\end{tikzpicture}%
}
\caption{History of the state of the art for the competitive ratio $\CR(\LQD)$.}
\label{fig:sota}
\end{figure}

Progress in this problem has come by squeezing
it between \emph{lower bounds} (inputs on which LQD provably falls short
of the clairvoyant optimum) and \emph{upper bounds} (guarantees that it
does not). Figure~\ref{fig:sota} shows a brief history of the state of the art:
the first formal analysis \citep{AKM08,HKM01} placed
$\CR(\LQD)$ in $[\sqrt2,\,2]$, and for almost two decades further
progress was confined to switches with two or three ports.
In earlier work with Bochkov and Gaevoy, we raised the lower
bound to $1.44546086$ with an explicit
family of hard inputs \citep{BDGN19}; \citet{AEMV21} broke the long-standing barrier
of $2$ with a per-queue payment scheme, obtaining $1.707$, sharpened to
$1.6918$ in the journal version \citep{AEMV24}.  The state of the art
before this work was thus $\CR(\LQD)\in[1.44546086,\ 1.6918]$;
Section~\ref{sec:related} reviews this line of work, and the wider
buffer management landscape, in detail.

\textbf{Contributions.} We make three contributions: we tighten both
ends of the interval, and we repair the proof foundation that one of
these improvements---and the originally published bound---depends on.
\begin{enumerate}
\item \textbf{Lower bound $1.46929591$ (exact-integer certificate,
Section~\ref{sec:lower}).}  We introduce a family of input instances that we call the
\emph{front-loaded family} (Definition~\ref{def:frontfamily}).
In our instances, each queue receives its arrivals only during one time
window; we say it is \emph{born} when the window opens and \emph{dies}
when it closes (Definition~\ref{def:vocab}).  Our instances start the
queues in synchronized batches (\emph{cohorts}); within each cohort a
fixed fraction of the queues die almost immediately after birth, and the
rest die along a prescribed schedule, which we call the \emph{death profile}.  We
found the particular profile by numerically optimizing a fluid
(continuum) relaxation of the instance design problem, but its
validity as a lower bound does not depend on that search: the
certificate is a single finite instance on which both LQD and the exact
offline optimum are evaluated in exact integer arithmetic.
The evaluation fixes a specific tie-breaking rule, but the bound does not: an adaptive
coupling (Theorem~\ref{thm:tieinv}) transfers the certified value to
every \emph{non-clairvoyant} tie rule (one whose choice at each slot
may depend on the run so far, but not on arrivals that have not yet
happened).
\item \textbf{Erratum and repair in the published upper-bound proof
(Sections~\ref{sec:erratum} and~\ref{sec:upper}).}  We found that one step of the proof of
\citet{AEMV24} (their Lemma~18, the aggregation step) is invalid as
published: a symbolic countermodel exhibits a shortfall arbitrarily
close to $1/2$ per queue, which does not shrink as their parameter
$\alpha$ grows (an \emph{$\alpha$-free} deficit).  Without a repair,
the published bounds $1.6918$ (journal) and $1.707$ (conference) are
unproven as written.  We provide an amortized repair of the
aggregation in place, valid for $\alpha\ge1/2$.
An exact finite-head envelope bound and two
elementary short-window cases absorb the resulting slack and restore
their best upper bound of $1.6918$.
\item \textbf{Upper bound $1.683652$ (Section~\ref{sec:upper}).}  The proof
of \citet{AEMV24} ends by bounding their payment scheme through a
per-packet relaxation that no single running-max envelope
can attain at once (see Section~\ref{sec:aemvdefs}).  Replacing that endgame with
the continuum envelope relaxation of the same payment expression,
solved exactly in
closed form, tightens the constant from $1.6918$ to $1.683652$. Our
improvement relies upon the repaired aggregation of contribution~2.
\end{enumerate}
In short, we certify $1.46929591\le\CR(\LQD)\le1.683652$, and
we repair the published proof that both the old and the new upper
bounds stand on.

\paragraph{How the proofs go.} Since all three arguments are somewhat
long, we sketch them here informally; each of the corresponding
sections also opens with a more detailed plan of its proofs.

The \emph{lower bound} rests on one intuition: LQD wastes buffer space
on queues that are about to die.  It cannot help it: an online policy
does not know which queues these are, and LQD in particular evicts from
a longest queue on every overflow, which on our instances keeps all
live queues at a single common height (the \emph{waterline}).  A
clairvoyant policy behaves very differently: it keeps each live queue
at height one, just enough to never miss a transmission, and invests
the entire remaining buffer into \emph{hoards}, packets stockpiled in a
queue right before it dies and transmitted from it, one per slot, long
after.  The instances that maximize this advantage turn out to be
\emph{front-loaded}: they kill a large fraction of every cohort almost
immediately, letting the offline player hoard into those queues while
the buffer is still mostly free, and the hoards then pay out over the
whole cycle.  We found the precise death schedule by numerically
optimizing a continuum (fluid) relaxation of the instance design
problem, but no property of that search enters the proof: the certified
bound is the exact ratio $\OPT(I)/\LQD(I)$ of a single finite instance,
with both sides evaluated in exact integer arithmetic.  What does
require a proof is that the value does not depend on how LQD breaks
ties among equally long queues: the only freedom a tie rule ever has is
which of the tied queues keeps a marginal packet, and we show by an
adaptive coupling that an adversary can always relabel future arrivals
so that this freedom never helps.

The \emph{erratum} concerns the payment scheme of \citet{AEMV24}, which
funds each queue's share of the $\OPT-\LQD$ gap from LQD's own
transmissions and aggregates the per-phase payment bounds along a
running-maximum envelope (their Lemma~18).  We show that this
aggregation spends one payment term twice.  When a queue's reference
level drops between phases, the published pairing covers the drop
penalty with the first payment term of the next phase; when that same
step also carries a correction charge, the correction must be covered
by the very same term, and the accounting falls short.  The shortfall
per queue can approach $\tfrac12$, and it does not shrink for any
choice of the scheme's parameters; we exhibit both an abstract family
approaching $\tfrac12$ and a fully realized nine-queue run losing
exactly $\tfrac14$.  The repair is an amortization: a potential
function on phase transitions shows that the total double-spend over
any single window telescopes down to at most $\tfrac12$, and the final
optimization absorbs this $\tfrac12$ from slack that every window's
payment provably carries.

For the \emph{upper bound}, the repaired aggregation leaves us with a
sum over a non-decreasing envelope sequence, and the published proof
bounds that sum packet by packet: each packet is priced at its own
worst-case envelope level.  These per-packet prices are not
simultaneously achievable: no single envelope is worst for every packet
at once, so the extracted constant is not tight.  We instead minimize
the whole sum over the envelope directly.  Passing to a continuum
relaxation turns the sum into an integral functional of the envelope
path; its minimizer is an envelope that jumps once at the start and
stays constant, and the resulting minimum has a closed form, proved
optimal by a direct supporting-line (convexity) argument.  The rest of the section is bookkeeping with a
purpose: a finite-head refinement of the same optimization shows that
every queue's payment clears the required rate with a reserve strictly
larger than the $\tfrac12$ lost in the repair, except in two smallest
cases that we check directly from the unaggregated inequalities.

The rest of the paper is organized as follows.
Section~\ref{sec:related} reviews related work.  Section~\ref{sec:model} fixes
the problem setting: the model, notation, and the prior results we
build on (drawn from \cite{AKM08,BDGN19,HKM01}).  The main results
occupy the next three sections, in the order of the contributions
above: Section~\ref{sec:lower} proves the lower bound; Section~\ref{sec:erratum}
presents the erratum in the published upper-bound proof, and its
repair; Section~\ref{sec:upper} builds the improved upper bound on the
repaired step.  In the latter two, we splice into the payment scheme
proof of \citet{AEMV24}, and every statement we import from that paper
is restated with its original number (symbols relabeled to our
notation, proofs theirs), so that each import can be checked against
its source in place.  Section~\ref{sec:conclusion} concludes the paper, and
Appendix~\ref{sec:appendix-certs} maps every computational claim to
its machine-checked certificate. We have made the code for all 
certificates referred to in this paper openly available 
at \url{https://doi.org/10.5281/zenodo.21260154}.

\section{Related work}\label{sec:related}

We first trace the LQD line of results that this paper extends, then
place shared memory within the wider buffer management landscape, and
close with the methodology of certified bounds that we build on.

\paragraph{The LQD line.}
LQD was proposed by \citet{WCH91} for shared memory packet switches,
and longest-queue preemption has remained attractive in practice since
it needs no per-flow configuration and shares bandwidth fairly under
overload \citep{CL07book,Suter98}.  Its first formal analysis is due
to \citet{HKM01} (journal version in \cite{AKM08}): LQD is
$2$-competitive and at least $\sqrt2$-competitive, and \emph{no}
deterministic online policy is better than $4/3$-competitive, which is a
\emph{universal} lower bound. The journal version also analyzes
non-preemptive policies, which we do not study here but which
sharpen the role of push-out.  Push-out
is what makes constant competitiveness possible: in the
non-preemptive setting, where admitted packets may not be evicted,
\citet{KM04} give a Harmonic policy that is $O(\log N)$-competitive
and prove a nearly matching lower bound for every non-preemptive
policy.

For a bounded number of ports, more is known.  \citet{KMO07,KMO08}
refined the LQD upper bound to
$2-\min_{k\le N}\bigl(\lfloor B/k\rfloor+k-1\bigr)/B$; for $N>\sqrt B$
this is only $2-O(1/\sqrt B)$, so it does not give $2-\eps$ for a
constant $\eps>0$. The same papers obtained the exact ratio
$\frac{4B-4}{3B-2}$ for $N=2$ ports (as \citet{AEMV24} note, the
published argument requires an even buffer size). \citet{Mat15}
proved that LQD is $1.5$-competitive for $N=3$ ports.

An interesting cautionary tale happened when a 2012 preprint
claiming $1.5$-competitiveness for arbitrary $N$ \citep{Mat12} was
later withdrawn, the recorded reason being an error in the definition
of the optimal offline algorithm. Offline optimality and tie-breaking are
indeed subtle in this model, and that is why our lower bound is stated
as an exact-integer certificate against a provably optimal offline policy
(Section~\ref{sec:lower}), and why the erratum of Section~\ref{sec:erratum}
is also verified by exact symbolic arithmetic.

In earlier work with Bochkov and Gaevoy \citep{BDGN19}, two of us
raised the LQD lower bound from $\sqrt2\approx1.41421$
to $1.44546086$, with an explicit uniform-death instance family,
$\Phi_k$ (see Definition~\ref{def:phi} below), found by direct
simulation and, independently, by a linear programming search over
structured instances.  The same paper raised the universal
deterministic lower bound from $4/3$ to $\sqrt2$ (a weaker number
than the LQD-specific bound, but holding for \emph{every}
deterministic online policy) via a staggered-birth construction
whose adversary always kills the online player's shortest live queue.
It also exhibited an explicit, exactly optimal offline policy
called \LateQD{}, which we reuse as Theorem~\ref{thm:lateqd} and evaluate
inside the certificates.

On the upper side, the first general progress since 2001 is due to \citet{AEMV21}: a per-queue
\emph{payment scheme} (restated in Section~\ref{sec:aemvdefs}) showing
$\CR(\LQD)\le1.707$, the first $2-\eps$ bound with constant
$\eps>0$, later improved to $\CR(\LQD)\le1.6918$ in the journal version
\citep{AEMV24}.  LQD remains the only policy known to be
constant-competitive for shared memory; no randomized policy has been
analyzed \citep{AEMV24}; and closing the LQD gap is a long-standing
open problem highlighted in several surveys \citep{Gol10,NK16}.  

Relative to this line, our lower bound (Section~\ref{sec:lower})
improves on the $1.44546086$ above using a new instance family
evaluated in exact integer arithmetic; our erratum and repair
(Sections~\ref{sec:erratum} and~\ref{sec:upper}) address a gap we find
in the AEMV24 aggregation step; and our upper bound
(Section~\ref{sec:upper}) solves the continuum relaxation of that
scheme's
endgame exactly to improve on $1.6918$.

\paragraph{The wider buffer management landscape.}
Shared memory is one member of a family of buffer management models
studied under competitive analysis; \citet{Gol10} surveys the area,
and \citet{NK16} surveys later developments.  In the single-queue QoS
model, packets carry values and the buffer is one FIFO queue; tight or
near-tight bounds are known in several regimes \citep{EW09,Zhu04},
and the neighboring bounded-delay (deadline) model was settled at the
golden ratio for deterministic algorithms by \citet{VCJS22}.  In
multi-queue switches with \emph{separate} per-port buffers, the
asymptotically optimal ratio is $e/(e-1)\approx1.582$: a matching
lower bound \citep{AS05} and algorithm \citep{AL06} are known; see also
\citep{AR05}.  Buffered crossbar and combined input--output queued
(CIOQ) architectures place buffers at the crosspoints or at both
inputs and outputs \citep{AEW18,KR06,KKS12}.  Shared-memory buffers
for packets with heterogeneous processing requirements were studied in
\citep{EKNS14}, where the natural LQD counterparts lose super-constant
factors.  None of these settings subsumes the model studied here: the
difficulty of LQD is the coupling of all $N$ queues through
one memory pool, with unit-value packets leaving no room for
value-based tricks.

\paragraph{Certified bounds.}
Methodologically we continue the computational approach that our
earlier work \citep{BDGN19} pioneered for this problem by searching
for hard instances with a linear program.
We push the search one step further, into a fluid (continuum)
relaxation of the instance design problem (Section~\ref{sec:lower}), but
note that the search only \emph{proposes} an instance.  The claimed bound is the
exact rational value of one finite instance, evaluated in
exact integer arithmetic against the exactly optimal offline
policy; the analytic constants of the upper bound are enclosed by
outward interval arithmetic (arithmetic that tracks a
guaranteed-containing range instead of a single floating-point number,
so a computed bound is proven rather than estimated); and for every
computational claim in the paper we provide a runnable certificate
(Appendix~\ref{sec:appendix-certs}).

\section{Problem setting: model, notation, and imported results}\label{sec:model}

In this section, we introduce the switch model and the \LQD{}  policy
(Definitions~\ref{def:switch}--\ref{def:lqd}), the instance vocabulary
used by the lower-bound construction of Section~\ref{sec:lower}
(Definition~\ref{def:vocab}), and the exactly optimal offline policy
against which every bound is measured; Table~\ref{tbl:notation}
collects the notation for reference throughout.

\begin{table}[!t]\small\centering
\caption{Notation used throughout the paper.}\label{tbl:notation}
\begin{tabular}{@{}l>{\raggedright\arraybackslash}p{0.74\linewidth}@{}}
\toprule
Symbol & Meaning (first used) \\
\midrule
$N$, $B$ & number of queues; buffer size (Def.~\ref{def:switch}) \\
\LQD, \OPT, \LateQD{} & the online policy; the optimal offline value; the exactly-optimal offline policy achieving it (Section~\ref{sec:model}) \\
$\gamma_i,\delta_i,\ell_i$ & a queue's birth slot, death slot, lifespan (Def.~\ref{def:vocab}) \\
$\omega,\kappa$; $z_L,z_C$ & the online/offline coupling bijections and slot-transmission
indicators of the tie-rule analysis (Thm.~\ref{thm:tieinv}, Lemma~\ref{lem:exchange}) \\
$\Phi_k$ & the uniform-death instance family of \citet{BDGN19} (Def.~\ref{def:phi}) \\
$s,\lambda,\theta$ & rescaled time $s=t/k$, lifespan
$\lambda=\ell/k$, and buffer size $\theta=B/k^2$ within a cohort
(Section~\ref{sec:lower} only) \\
$F^{\front}$, $\Phi^{\front}_k$ & our front-loaded death profile and instance family (Def.~\ref{def:frontfamily}) \\
$\alpha,\beta$ & AEMV's payment scheme parameters (Section~\ref{sec:aemvdefs}) \\
$\Phi_q$ & AEMV's per-queue profit counter (relabeled $P_q$ here; Section~\ref{sec:aemvdefs}) --- unrelated to the instance families $\Phi_k,\Phi^{\front}_k$ \\
$\hat e_q,\hat e_{q,i},\hat e_{\bullet,i}$ & queue $q$'s outstanding OPT-extra credit --- overall, per phase, and ($\bullet$) summed over queues at phase $i$ (Section~\ref{sec:aemvdefs}) \\
$\OPTX,\LQDX$ & the OPT-extra and LQD-extra packet counts underlying $\hat e_q$ (Section~\ref{sec:aemvdefs}) \\
$x=\hat e/b_0$ & the rescaled credit driving the envelope functions (Section~\ref{sec:upper}) \\
$\tau_i$, phase $i$ & AEMV's overflow-indexed time boundaries (Section~\ref{sec:aemvdefs}) \\
$\sigma_{q,i}$ & AEMV's average live-queue level (Section~\ref{sec:aemvdefs}) \\
$b_j$ & the running-max envelope, re-indexed by step (Section~\ref{sec:aemvdefs}) \\
$\varrho,\hat\varrho,\rhostar$ & a generic payment rate; ours; AEMV's (Section~\ref{sec:aemvdefs}--Section~\ref{sec:upper}) \\
$L$ & length of a queue's constant-debt block in its pending window (Lemma~\ref{lem:windowstructure}) \\
$E^{\inf}_{b,e}$, $K_b$, $J_b^\#$, $J_{\rm cont}$ & the discrete envelope infimum; its finite-head continuum lower bound; that bound including the base S-payment; the continuum ($b\to\infty$) limit (Thm.~\ref{thm:improved}) \\
$\Xi_b(x)$ & the finite-head reserve absorbing the repair's $1/2$ loss (Lemma~\ref{lem:reserve}) \\
\bottomrule
\end{tabular}
\end{table}

\begin{definition}[Shared memory switch, policies, and the competitive ratio]\label{def:switch}
A \emph{shared memory switch} consists of $N$ output queues sharing a
single buffer of integer size $B$ (at most $B$ packets stored in total,
across all queues).  Time is slotted; in each slot $t=1,2,\dots$ two
phases occur:
\begin{enumerate}[label=(\roman*)]
	\item \emph{arrival phase}: a batch $r_t\in\mathbb{Z}_{\ge0}^N$ of new packets arrives (entry $q$ is the
	number arriving for queue $q$); the switch may admit or evict packets
	arbitrarily, including packets already stored (\emph{push-out}), so long
	as the buffer holds at most $B$ packets when the phase ends; an
	\emph{overflow} is an arrival phase after which admitting all
	arrivals would leave more than $B$ packets stored, i.e., a phase where some eviction is forced;
	\item \emph{transmission phase}: every nonempty queue transmits one packet.
\end{enumerate}
All packets have unit value, and the goal is to transmit as
many as possible.

A \emph{policy} (or algorithm) $A$ is a rule for the switch's
admission/eviction decisions.  It is \emph{online} if its decisions in
slot $t$ depend only on the arrivals $r_1,\dots,r_t$ seen so far, and
\emph{offline} (equivalently, clairvoyant) if they may depend on the
entire arrival sequence.  An \emph{instance} $I$ is a finite arrival
sequence $r_1,r_2,\dots$; the run continues until the buffer empties
(the queues \emph{drain}), and $A(I)$ denotes the total number of
packets $A$ transmits on $I$.  We write $\OPT(I)$ for the maximum of
$A(I)$ over all offline policies; this is the benchmark against which
the competitive ratio is measured, and it is achieved by an
explicit policy in Theorem~\ref{thm:lateqd}.  We use the standard
additive/asymptotic definition: the \emph{competitive ratio} of $A$ is
\[
\CR(A)\;=\;\inf\bigl\{c\ge1:\ \exists K\ge0\ \ \forall B\ \forall I,\quad
\OPT(I)\le c\,A(I)+K B\bigr\}.
\]
Here $K$ is independent of the instance length, the number of queues,
and $B$.  For a fixed finite instance we separately call
$\OPT(I)/A(I)$ its \emph{exact ratio}; repeating drained copies turns
such a witness into an asymptotic lower bound (Remark~\ref{rem:cradditive}).
\end{definition}

This model, including push-out, is exactly the model of
\cite{AKM08,HKM01}; the same works also
treat non-preemptive policies, which are not used here.  Our vocabulary for the
instances of Section~\ref{sec:lower}, and the offline machinery of
Theorem~\ref{thm:lateqd}, follows \citet{BDGN19}.

\begin{definition}[Longest Queue Drop, LQD]\label{def:lqd}
\emph{LQD} is the online policy that resolves each overflow
(Definition~\ref{def:switch}) by repeatedly evicting one packet from a
currently longest queue until at most $B$ remain.  When
several queues are tied for longest, the choice is settled by a fixed
deterministic \emph{tie rule} $T$; each choice of $T$ completes the
description above to a deterministic policy $\LQD_T$, and the
competitive ratio is defined as the worst case over tie rules,
\[
\CR(\LQD)\;=\;\sup_T\,\CR(\LQD_T).
\]
For a \emph{randomized} tie rule $T$ (one that may flip coins to break
ties), $\CR(\LQD_T)$ is read in the \emph{adaptive adversary} sense: the
instance may depend on the coins $T$ has already flipped, and the ratio
is required on every realization of $T$'s choices. This is the definition
in which Theorem~\ref{thm:tieinv} states its randomized conclusion;
\hyperlink{rem:tiescope-target}{Remark~\ref*{rem:tiescope}} contrasts it with the oblivious convention.
\end{definition}

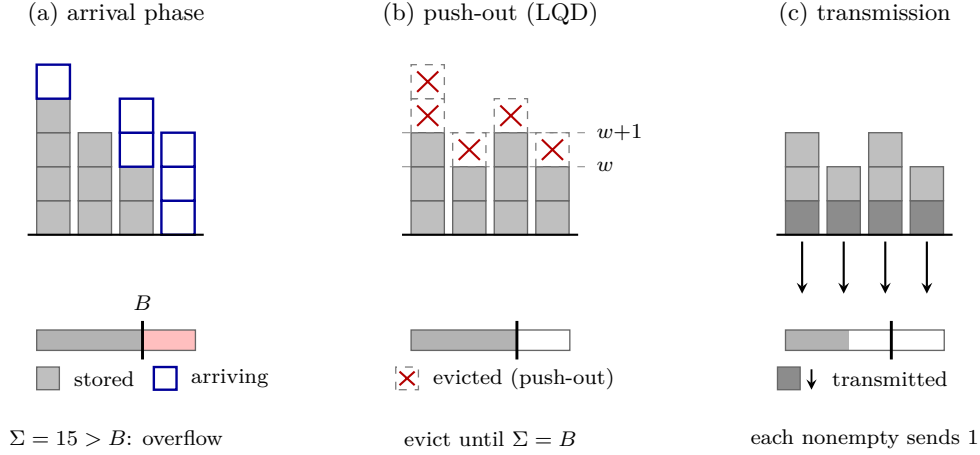
\begin{figure}[!t]
\centering
\begin{tikzpicture}[x=1cm,y=1.5cm,
  cellS/.style={fill=black!25,draw=black!60,line width=0.5pt},   
  cellA/.style={fill=white,draw=blue!60!black,line width=0.9pt}, 
  cellE/.style={draw=black!50,dashed,line width=0.5pt},          
  cellT/.style={fill=black!45,draw=black!60,line width=0.5pt},   
  ttl/.style={font=\footnotesize},
  lbl/.style={font=\scriptsize}]


\begin{scope}
\node[ttl] at (1.05,1.95) {(a) arrival phase};
\foreach \q/\h in {0/4, 1/3, 2/2}
  \foreach \l in {1,...,\h}
    \draw[cellS] ({\q*0.55}, {(\l-1)*0.30}) rectangle ({\q*0.55+0.44}, {\l*0.30});
\draw[cellA] (0.00,1.20) rectangle (0.44,1.50);
\foreach \l in {3,4}
  \draw[cellA] (1.10,{(\l-1)*0.30}) rectangle (1.54,{\l*0.30});
\foreach \l in {1,2,3}
  \draw[cellA] (1.65,{(\l-1)*0.30}) rectangle (2.09,{\l*0.30});
\draw[line width=0.8pt] (-0.12,0) -- (2.21,0);
\foreach \s in {1,...,10}
  \fill[black!30] ({(\s-1)*0.14},-1.02) rectangle ({\s*0.14},-0.84);
\foreach \s in {11,...,15}
  \fill[red!25] ({(\s-1)*0.14},-1.02) rectangle ({\s*0.14},-0.84);
\draw[black!60,line width=0.5pt] (0,-1.02) rectangle (2.10,-0.84);
\draw[line width=1.1pt] (1.40,-1.10) -- (1.40,-0.76) node[above,lbl] {$B$};
\node[lbl] at (1.05,-1.80) {$\Sigma=15>B$: overflow};
\end{scope}

\begin{scope}[xshift=.3\linewidth]
\node[ttl] at (1.05,1.95) {(b) push-out (LQD)};
\draw[black!45,dashed] (-0.12,0.60) -- (2.30,0.60);
\draw[black!45,dashed] (-0.12,0.90) -- (2.30,0.90);
\node[lbl,anchor=west] at (2.32,0.90) {$w{+}1$};
\node[lbl,anchor=west] at (2.32,0.60) {$w$};
\foreach \q/\h in {0/3, 1/2, 2/3, 3/2}
  \foreach \l in {1,...,\h}
    \draw[cellS] ({\q*0.55}, {(\l-1)*0.30}) rectangle ({\q*0.55+0.44}, {\l*0.30});
\foreach \q/\l in {0/4, 0/5, 1/3, 2/4, 3/3}{
  \draw[cellE] ({\q*0.55}, {(\l-1)*0.30}) rectangle ({\q*0.55+0.44}, {\l*0.30});
  \draw[red!70!black,line width=0.8pt]
    ({\q*0.55+0.08},{(\l-1)*0.30+0.06}) -- ({\q*0.55+0.36},{\l*0.30-0.06})
    ({\q*0.55+0.08},{\l*0.30-0.06})    -- ({\q*0.55+0.36},{(\l-1)*0.30+0.06});
}
\draw[line width=0.8pt] (-0.12,0) -- (2.21,0);
\foreach \s in {1,...,10}
  \fill[black!30] ({(\s-1)*0.14},-1.02) rectangle ({\s*0.14},-0.84);
\draw[black!60,line width=0.5pt] (0,-1.02) rectangle (2.10,-0.84);
\draw[line width=1.1pt] (1.40,-1.10) -- (1.40,-0.76);
\node[lbl] at (1.05,-1.80) {evict until $\Sigma=B$};
\end{scope}

\begin{scope}[xshift=.6\linewidth]
\node[ttl] at (1.05,1.95) {(c) transmission};
\foreach \q/\h in {0/3, 1/2, 2/3, 3/2}{
  \draw[cellT] ({\q*0.55},0) rectangle ({\q*0.55+0.44},0.30);
  \foreach \l in {2,...,\h}
    \draw[cellS] ({\q*0.55}, {(\l-1)*0.30}) rectangle ({\q*0.55+0.44}, {\l*0.30});
  \draw[-stealth,line width=0.8pt] ({\q*0.55+0.22},-0.06) -- ({\q*0.55+0.22},-0.52);
}
\draw[line width=0.8pt] (-0.12,0) -- (2.21,0);
\foreach \s in {1,...,6}
  \fill[black!30] ({(\s-1)*0.14},-1.02) rectangle ({\s*0.14},-0.84);
\draw[black!60,line width=0.5pt] (0,-1.02) rectangle (2.10,-0.84);
\draw[line width=1.1pt] (1.40,-1.10) -- (1.40,-0.76);
\node[lbl] at (1.05,-1.80) {each nonempty sends $1$};
\end{scope}

\begin{scope}[yshift=-2.05cm]
\draw[cellS] (0.00,0) rectangle (0.30,0.20); \node[lbl,anchor=west] at (0.36,0.10) {stored};
\draw[cellA] (1.55,0) rectangle (1.85,0.20); \node[lbl,anchor=west] at (1.91,0.10) {arriving};
\draw[cellE] (4.75,0) rectangle (5.05,0.20);
\draw[red!70!black,line width=0.8pt] (4.80,0.04) -- (5.0,0.16) (4.80,0.16) -- (5.0,0.04);
\node[lbl,anchor=west] at (5.11,0.10) {evicted (push-out)};
\draw[cellT] (9.80,0) rectangle (10.10,0.20);
\draw[-stealth,line width=0.7pt] (10.25,0.18) -- (10.25,0.00);
\node[lbl,anchor=west] at (10.39,0.10) {transmitted};
\end{scope}
\end{tikzpicture}
\caption{One slot of the shared memory switch (Definition~\ref{def:switch}) under LQD (Definition~\ref{def:lqd}), with
$N=4$ queues and buffer $B=10$; the ruler below each panel is the shared memory.}
\label{fig:modelstep}
\end{figure}

Figure~\ref{fig:modelstep} illustrates the LQD policy and overflows in the shared memory switch.
\FloatBarrier

Both kinds of bounds in this paper are
statements about $\CR(\LQD)$: following \citet{AEMV24}, the upper
bounds hold for \emph{every} tie rule, so they bound the supremum;
a lower bound exhibits a bad finite instance for a single rule and
repeats drained copies to defeat every additive term.  The rule we fix for the certificate of
Section~\ref{sec:lower}, $T_0$, is specified in Remark~\ref{rem:tierule};
Theorem~\ref{thm:tieinv} shows the certified value does not
depend on that choice, and the lower bound holds for every
non-clairvoyant tie rule.

The lower bound construction of Section~\ref{sec:lower} uses instances of a
special shape, for which we now fix some vocabulary.  All of the following
definitions name features of an \emph{instance} (the input), not of the policy.

\begin{definition}[Interval instances and the queue life cycle]\label{def:vocab}
An instance is \emph{interval-structured} if each queue $i$ receives its
arrivals during a single contiguous block of slots $[\gamma_i,\delta_i]$ and none
outside it.  A queue's \emph{level} at any point is the number of
packets it currently holds.  For such a queue:
\begin{itemize}\itemsep2pt
\item it is \emph{born} at $\gamma_i$, \emph{live} at any slot $t\in[\gamma_i,\delta_i]$,
  \emph{dying} at its last live slot $\delta_i$ (its \emph{death slot}), and
  \emph{dead} afterwards --- a dead queue receives no more arrivals but
  keeps transmitting one packet per slot until it empties;
\item its \emph{lifespan} is $\ell_i=\delta_i-\gamma_i+1$.
\end{itemize}
A \emph{cohort} is a set of queues sharing the same birth slot.
When all arrivals have ceased, the buffer empties over the remaining
slots (the drain, Definition~\ref{def:switch}); a policy's value on the instance counts every
transmission, drain included.  A dead queue that still holds packets is
a \emph{corpse}; once the policy has evicted some of a corpse's packets
it is \emph{trimmed}, and the corpse's buffer level at its death slot is
its \emph{pre-trim level}.

Our lower-bound constructions additionally assume \emph{saturation}:
every live queue receives $B$ packets on every slot of its interval, so
a single live queue could by itself fill the buffer.
\end{definition}

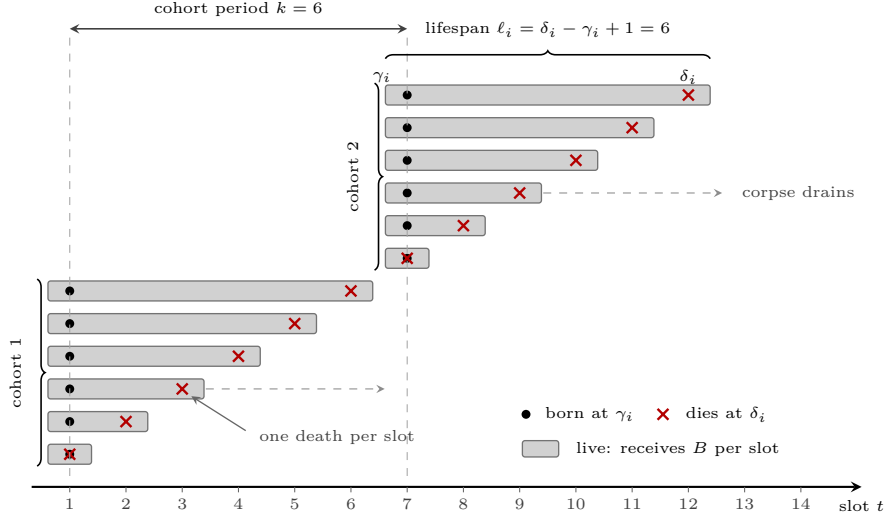
\begin{figure}[!t]
\centering
\begin{tikzpicture}[x=1cm,y=1cm,scale=1.2,
  bar/.style={fill=black!18,draw=black!55,line width=0.5pt,rounded corners=1pt},
  lbl/.style={font=\tiny},
  note/.style={font=\tiny,black!75}]

\draw[-stealth,line width=0.8pt] (0.20,0) -- (9.35,0) node[below,lbl,yshift=-1pt] {slot $t$};
\foreach \s in {1,...,14}{
  \draw[black!60] ({0.62*\s},0) -- ({0.62*\s},-0.06);
  \node[lbl,black!75] at ({0.62*\s},-0.19) {\s};
}

\foreach \j in {1,...,6}{
  \draw[bar] ({0.62*1-0.24},{0.36*\j-0.11}) rectangle ({0.62*\j+0.24},{0.36*\j+0.11});
  \fill ({0.62*1},{0.36*\j}) circle (1.4pt);
  \draw[red!70!black,line width=0.9pt]
    ({0.62*\j-0.06},{0.36*\j-0.06}) -- ({0.62*\j+0.06},{0.36*\j+0.06})
    ({0.62*\j-0.06},{0.36*\j+0.06}) -- ({0.62*\j+0.06},{0.36*\j-0.06});
}
\foreach \j in {1,...,6}{
  \draw[bar] ({0.62*7-0.24},{0.36*(6+\j)-0.11}) rectangle ({0.62*(6+\j)+0.24},{0.36*(6+\j)+0.11});
  \fill ({0.62*7},{0.36*(6+\j)}) circle (1.4pt);
  \draw[red!70!black,line width=0.9pt]
    ({0.62*(6+\j)-0.06},{0.36*(6+\j)-0.06}) -- ({0.62*(6+\j)+0.06},{0.36*(6+\j)+0.06})
    ({0.62*(6+\j)-0.06},{0.36*(6+\j)+0.06}) -- ({0.62*(6+\j)+0.06},{0.36*(6+\j)-0.06});
}

\draw[black!45,dashed,-stealth] ({0.62*3+0.26},{0.36*3}) -- ({0.62*6.6},{0.36*3});
\draw[black!45,dashed,-stealth] ({0.62*9+0.26},{0.36*9}) -- ({0.62*12.6},{0.36*9});
\node[note,anchor=west] at (7.92,{0.36*9}) {corpse drains};

\draw[decorate,decoration={brace,amplitude=3pt},line width=0.6pt]
  (0.26,{0.36*6+0.13}) -- (0.26,{0.36*1-0.13});
\node[lbl,rotate=90] at (0.02,{0.36*3.5}) {cohort 1};
\draw[decorate,decoration={brace,amplitude=3pt},line width=0.6pt]
  (3.98,{0.36*12+0.13}) -- (3.98,{0.36*7-0.13});
\node[lbl,rotate=90] at (3.74,{0.36*9.5}) {cohort 2};

\draw[black!35,dashed] ({0.62*1},0.12) -- ({0.62*1},5.02);
\draw[black!35,dashed] ({0.62*7},0.12) -- ({0.62*7},5.02);
\draw[stealth-stealth,black!75,line width=0.6pt] ({0.62*1},5.05) -- ({0.62*7},5.05);
\node[lbl,anchor=south] at ({0.62*4},5.08) {cohort period $k=6$};

\node[lbl,anchor=east] at ({0.62*7-0.06},4.52) {$\gamma_i$};
\node[lbl] at ({0.62*12},4.52) {$\delta_i$};
\draw[decorate,decoration={brace,amplitude=3pt},line width=0.6pt]
  ({0.62*7-0.24},4.72) -- ({0.62*12+0.24},4.72);
\node[lbl,anchor=south] at ({0.62*9.5},4.85) {lifespan $\ell_i=\delta_i-\gamma_i+1=6$};

\node[note,anchor=west] at (2.60,0.55) {one death per slot};
\draw[-stealth,black!60,line width=0.5pt] (2.56,0.62) -- ({0.62*3+0.10},{0.36*3-0.08});

\fill (5.65,0.80) circle (1.4pt);
\node[lbl,anchor=west] at (5.75,0.80) {born at $\gamma_i$};
\draw[red!70!black,line width=0.9pt] (7.10,0.74) -- (7.22,0.86) (7.10,0.86) -- (7.22,0.74);
\node[lbl,anchor=west] at (7.30,0.80) {dies at $\delta_i$};
\draw[bar] (5.61,0.33) rectangle (6.01,0.51);
\node[lbl,anchor=west] at (6.09,0.42) {live: receives $B$ per slot};
\end{tikzpicture}
\caption{Interval instances and the queue life cycle
(Definition~\ref{def:vocab}), drawn for the uniform family $\Phi_6(2)$
of Definition~\ref{def:phi}.}
\label{fig:vocab}
\end{figure}

Figure~\ref{fig:vocab} illustrates the concepts of
Definition~\ref{def:vocab} on the uniform family $\Phi_6(2)$: each bar
is one queue's arrival interval $[\gamma_i,\delta_i]$, from birth to
death slot, and the dashed tails show dead queues draining their
corpses, one packet per slot.  Queues born together form a cohort; in
$\Phi_k$ a fresh cohort of $k$ queues is born every $k$ slots and
exactly one of its queues dies per slot.
\FloatBarrier

To evaluate $\OPT(I)$ on a concrete instance we need an explicit optimal
offline policy.  \citet{BDGN19} provide one, which we reuse as a black box below.

\begin{theorem}[optimal offline policy, imported from Thm.~1 and Prop.~1 of \cite{BDGN19}]\label{thm:lateqd}
There is a clairvoyant policy \LateQD{} that is exactly optimal offline
(so $\LateQD(I)=\OPT(I)$ for every instance $I$).  It is defined as follows:
\begin{itemize}
	\item run the instance once with an unbounded buffer under last-in-first-out
	(LIFO) service; since nothing is ever dropped there, every packet
	acquires a definite slot in which it is eventually sent, its
	\emph{provisional transmission time}; 
	\item then, back in the real $B$-bounded run, on each overflow evict the stored packet whose
	provisional transmission time is \emph{latest} (hence the name).
\end{itemize}
On saturated interval-structured instances (Definition~\ref{def:vocab}),
\citet{BDGN19}'s Prop.~1 additionally gives a value-equivalent explicit
form that Section~\ref{sec:lower} evaluates directly.  If the number of
nonempty queues after arrivals exceeds $B$, it keeps one packet in any
$B$ of them (and the ensuing transmission empties the buffer).  Otherwise
it keeps one packet per live queue that is not yet at its death slot and
runs LQD's own eviction rule on the dying and dead queues in the remaining
buffer space (\emph{LQD-equalization}), draining them as evenly as that
space allows.
\end{theorem}

\begin{definition}[the uniform family $\Phi_k$, from \cite{BDGN19}]\label{def:phi}
$\Phi_k(C)$ is an interval-structured instance defined as follows 
(using the vocabulary of Definition~\ref{def:vocab}).
Slots run in \emph{cycles} of length $k$; at
each cycle boundary a fresh cohort of $k$ queues is born, and within a
cycle exactly one of the currently live queues dies per slot, so the
death slots are spread uniformly.  Formally, a queue $j\ge1$ is live
during slots $[\,k\lfloor(j-1)/k\rfloor+1,\ j\,]$ (its birth slot is
the start of its cycle and its death slot is $j$), and every live
queue receives $B$ packets each slot.  The parameter $C$ truncates the
construction to $C$ cycles (after which the buffer drains), making
$\Phi_k(C)$ a finite instance.  This is the \emph{uniform} death profile (one
death per slot); the family of Section~\ref{sec:lower} generalizes it by
letting the deaths be scheduled non-uniformly.
\end{definition}

\begin{remark}[number of queues]\label{rem:nqueues}
Every instance here is finite: $\Phi_k(C)$ and the front-loaded family
$\Phi^{\front}_k(C)$ of Section~\ref{sec:lower} use $N=kC$ distinct queues (a
fresh set of $k$ per cohort), the competitive ratio is defined over
instances of arbitrary finite $N$, and the upper bounds of
\citet{HKM01,AEMV24} and of Theorem~\ref{thm:improved} hold for every
$N$; no limit in $N$ is taken.  Large $N$ is also not essential: at most
$O(k)$ queues are ever simultaneously nonempty, so one may recycle queue
indices to obtain an equivalent instance with $N\approx2k$.  We do not
certify the recycled variants, because recycling breaks the
one-interval-per-queue structure that makes the explicit form of
\LateQD{} (Theorem~\ref{thm:lateqd}) fast to evaluate; the general form
of \LateQD{} remains optimal but is costlier, and no result here needs
small $N$.  (The waterline picture behind ``$O(k)$
simultaneously nonempty'' is explained in Section~\ref{sec:lower}.)
\end{remark}

\begin{remark}[competitive ratio with additive terms; convention]\label{rem:cradditive}
Definition~\ref{def:switch} permits an additive $K B$ independent of
the instance length.  This is the convention used by the imported upper
bounds: on runs with $A(I)\to\infty$, the additive term vanishes after
division by $A(I)$.  Our lower-bound certificate is still stated on
\emph{one} finite instance and reports its \emph{exact} ratio
$\OPT(I)/A(I)$; concatenation is only the vehicle that lifts that
ratio past the additive $K B$.  Concatenate $R$ copies on disjoint
queue names, each starting only after the previous copy has drained to
an empty switch, so the copies do not interact.  The reference run
then replays on every copy, multiplying its LQD value by exactly $R$
(the reference rule is replayed per copy, so no hidden reset
assumption), while the offline value is multiplied by at least $R$: a
lower bound is all the numerator needs.  Hence no fixed $K B$ can
reduce the limiting ratio.  For an arbitrary history-dependent tie
rule, which need not replay across copies, we apply
Theorem~\ref{thm:tieinv} (every non-clairvoyant tie rule attains at
least the ratio certified on a given instance) directly to the
$R$-copy reference instance, and the same $R$-linear growth of both
values follows.
\end{remark}

\begin{figure}[!t]
\centering
\includegraphics[width=0.85\linewidth]{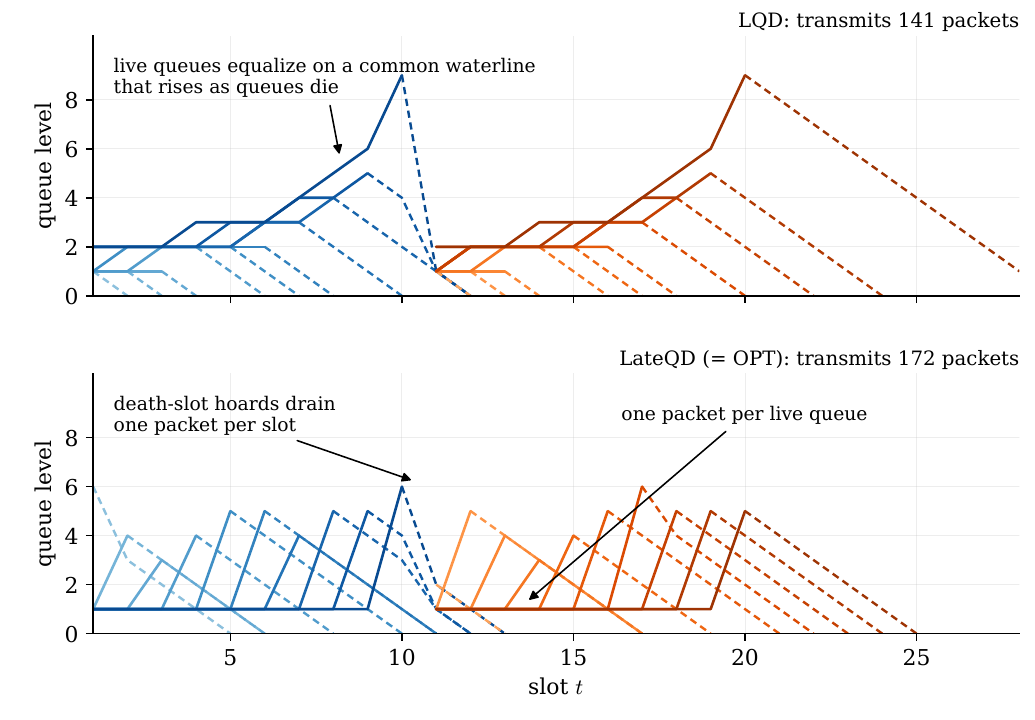}
\caption{LQD (top) and the optimal offline policy \LateQD{} (bottom;
Theorem~\ref{thm:lateqd}) on the instance $\Phi_{10}(2)$ with $B=15$:
queue levels per slot, solid while live, dashed after death, coloured
by cohort.  LQD runs under the tie rule $T_0$ of
Remark~\ref{rem:tierule}, and here $\OPT/\LQD=172/141\approx1.2199$.}
\label{fig:run}
\end{figure}

Figure~\ref{fig:run} illustrates how LQD loses to \LateQD{} on an
instance from the $\Phi$ family.  LQD cannot tell long-lived queues
from short-lived ones, so all live queues share one \emph{waterline}
(a common fill level), which rises as the cohort thins; a dying queue
leaves behind only whatever the waterline happened to give it.
\LateQD{} spends the same buffer very differently: it keeps exactly
one packet in every live queue and invests everything else into
\emph{death-slot hoards}, packets stockpiled for a queue's own death
slot, which then drain one per slot long after the queue is dead.
Already on this small instance the clairvoyant advantage is about
$22\%$.  The exact LQD total depends on the tie rule: the figure
fixes $T_0$ of Remark~\ref{rem:tierule}, the rule used throughout the
paper, and the opposite rule would transmit $152$ packets on the same
instance.  The waterline-versus-hoards contrast does not depend on
that choice.
\FloatBarrier


Table~\ref{tbl:notation} summarizes the notation used in the paper; every symbol is also introduced at its point of first use. With the model, tie-rule convention, and optimal offline policy \LateQD{} in place, we now construct an explicit instance witnessing a lower bound on $\CR(\LQD)$.

\section{The lower bound: an exact-integer certificate}\label{sec:lower}

In this section, we prove the paper's lower bound in two stages. First, we
exhibit a single finite instance and a single tie rule under which
\LQD{} falls short of \OPT{} by the claimed ratio
(Theorem~\ref{thm:certified}); this alone suffices to lower-bound
$\CR(\LQD)$. Second, we show the same ratio holds for \emph{every}
non-clairvoyant tie rule (Theorem~\ref{thm:tieinv}), by an adaptive
coupling argument whose full proof is deferred to
Appendix~\ref{sec:appendix-tie}.  Before either stage, we explain
where the instances come from: a fluid relaxation of the
instance design problem, whose only role is to propose the death
profile.

We found the instance family below by numerically maximizing a
continuum (\emph{fluid}) relaxation of the instance design problem; the
certificate's validity rests only on the exact evaluation of one finite
instance, but stating the relaxation explains where the family comes
from and the role of ``front-loading.''  Figure~\ref{fig:families}
shows one cohort of each family, one queue's live interval per
horizontal bar: in $\Phi_k$ exactly one queue dies per slot, while in
$\Phi^{\front}_k$ about $21\%$ of the cohort dies almost immediately
(the \emph{early-death mass}) and the rest follow the continuous part
of the schedule; panel~(c) compares the two lifespan CDFs, with
$F^{\front}$ plotted from the certified knot data.

\paragraph{The fluid design problem (search objective only; no
convergence claim is made or used).}
A lower bound on $\CR(\LQD)$ comes from an instance on which the offline
optimum out-transmits LQD by as large a factor as possible.  A natural
question therefore drives the search: which death schedule hurts LQD
the most?  We answer it numerically, by searching for the death
profile that maximizes
the offline-to-LQD throughput ratio $T_O/T_L$, both throughputs taken
in a continuum idealization of one cohort; its two sides $T_O$
(offline) and $T_L$ (LQD) are defined next.

\emph{Rescaling.} Rescale time within a cohort by the cohort length
$k$, so a cohort occupies rescaled time $[0,1)$; rescale queue heights
by $k$; and write $\theta=B/k^2$ for the rescaled buffer.  All $k$
queues of a cohort are born together at rescaled time $0$, so a queue's
rescaled lifespan $\lambda$ is also its rescaled death time.  A
\emph{death profile} is a probability distribution $F$ on $(0,1]$ of
these lifespans (the uniform profile recovers $\Phi_k$).  We work in the
cyclic steady state, in which the profile repeats every cycle; write
$\mathrm{age}_\lambda(s)=(s-\lambda)\bmod1$ for the time elapsed since a
lifespan-$\lambda$ queue most recently died.  To include packets that
survive a cycle boundary, define the finite sums
\[
R_z(a)=\sum_{r\ge0}(z-a-r)_+,
\qquad
D_z(a)=\sum_{r\ge0}\mathbf 1\{z-a-r>0\}.
\]
Here $r$ indexes earlier cohorts; only finitely many summands are nonzero.
At phase $s$ the current cohort's live queues are those with
$\lambda>s$, of total mass $1-F(s)$, while $R_z$ and $D_z$ account for
all still-draining earlier deaths.

\emph{The offline side.} On interval instances the optimal offline
policy keeps one packet in each live queue and spends the entire
remaining buffer on \emph{hoards}: just before a queue dies it
stockpiles $H(\lambda)$ packets, then transmits them one per step after
death.  Its per-cohort throughput is
\begin{align*}
T_O(F,\theta)&=\int\lambda\,dF(\lambda)+V(F,\theta),\qquad\text{where}\\
V(F,\theta)&=\sup_{H\ge0}\int H\,dF
\quad \text{s.t.}\quad \int R_{H(\lambda)}\!\bigl(\mathrm{age}_\lambda(s)\bigr)dF(\lambda)\le\theta
\ \ \forall s.
\end{align*}
The first term is live throughput (a queue live for time $\lambda$
transmits $\lambda$).  In the \emph{hoard LP} $V$, the integrand
$R_{H(\lambda)}(\mathrm{age}_\lambda(s))$ is the total part of that
lifespan class's periodic hoards still in the buffer at phase $s$, and
the constraint is that all draining hoards together never exceed
$\theta$; $V$ picks the hoards to maximize the extra post-death service.

\emph{The LQD side.} LQD cannot see lifespans, so it cannot target its
buffer at the queues about to die.  On every overflow it evicts from a
longest queue; in the fluid limit this keeps all live queues at a
single common height, the \emph{waterline} $w(s)$, while a queue
that has already died keeps the height it had at its own death and
drains one packet per step thereafter (having no further arrivals, it
is never refilled).  Two deliberate idealizations enter here.  First,
the model tracks each dead queue's residual as shrinking by
transmission only; in the exact dynamics a later overflow can also
evict from a dead queue whose residual still tops the waterline, so
real residuals shrink by trimming as well.  Second, the model fixes
the waterline by keeping the buffer full,
\[
\underbrace{(1-F(s))\,w(s)}_{\text{live queues at the waterline}}
\;+\;\underbrace{\int R_{w(\lambda)}\!\bigl(\mathrm{age}_\lambda(s)\bigr)dF(\lambda)}_{\text{draining residuals of dead queues from all cohorts}}
\;=\;\theta ,
\]
so a queue dying at $\lambda$ leaves residual $w(\lambda)$ behind;
this equality describes the arrival-saturated cyclic steady state and
fails in the final drain phase, after the last cohort's arrivals stop
and the buffer empties.
Every nonempty queue transmits one packet per step, so the per-cohort
throughput is the time-integral of the transmitting mass,
\[
T_L(F,\theta)=\int_0^1\Bigl[\,\underbrace{\bigl(1-F(s)\bigr)}_{\text{live}}
\;+\;\underbrace{\int D_{w(\lambda)}\!\bigl(\mathrm{age}_\lambda(s)\bigr)\,dF(\lambda)}_{\text{dead, still draining}}\,\Bigr]\,ds ,
\]
the live throughput $\int\lambda\,dF$ (as in $T_O$) plus the incidental
post-death drain: LQD's blind analogue of the offline hoard $V$.
The waterline $w(s)$ couples all queues through the buffer equation and
has no closed form, so we solve it by a fixed-point iteration and
evaluate $T_L$ numerically (in the cyclic steady state a residual not
drained within its cycle carries into the next).  The design problem is
$\sup_{F,\theta}T_O/T_L$.

\emph{Why front-loading wins.} Ascending this ratio from the uniform profile
produces the front-loaded $F^{\front}$ of
Definition~\ref{def:frontfamily}, and the two sides above show why it
wins.  Concentrating death mass near $\lambda=0$ lets the offline
player hoard heavily into those queues while the buffer is nearly empty
of competing residuals, and those hoards then drain over almost the
whole cycle, making $V$ large.  LQD gains nothing from the same profile:
its waterline is blind to lifespans and stocks every queue alike,
so its incidental drain stays small.  Every quantity in the theorem
below is finite-$k$ and exact; the fluid problem is a heuristic used
only to \emph{propose} the profile, and neither idealization above
enters any certified statement.

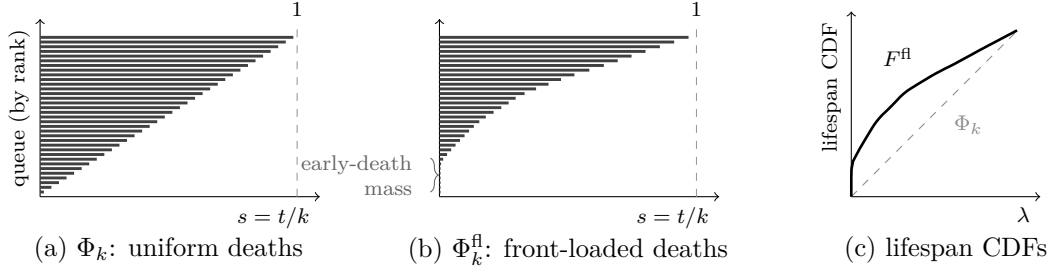
\begin{figure}[!t]
\centering
\begin{tikzpicture}[barcol/.style={black!75}]
\begin{scope}
\draw[->] (0,0) -- (3.7,0) node[below left,font=\scriptsize] {$s=t/k$};
\draw[->] (0,0) -- (0,2.35) node[rotate=90,anchor=south east,font=\scriptsize] {queue (by rank)};
\draw[dashed,black!40] (3.4,0) -- (3.4,2.25) node[above,font=\scriptsize,black] {$1$};
\draw[line width=1.1pt,barcol] (0,0.062) -- (0.050,0.062);
\draw[line width=1.1pt,barcol] (0,0.124) -- (0.150,0.124);
\draw[line width=1.1pt,barcol] (0,0.186) -- (0.250,0.186);
\draw[line width=1.1pt,barcol] (0,0.248) -- (0.350,0.248);
\draw[line width=1.1pt,barcol] (0,0.310) -- (0.450,0.310);
\draw[line width=1.1pt,barcol] (0,0.372) -- (0.550,0.372);
\draw[line width=1.1pt,barcol] (0,0.434) -- (0.650,0.434);
\draw[line width=1.1pt,barcol] (0,0.496) -- (0.750,0.496);
\draw[line width=1.1pt,barcol] (0,0.558) -- (0.850,0.558);
\draw[line width=1.1pt,barcol] (0,0.620) -- (0.950,0.620);
\draw[line width=1.1pt,barcol] (0,0.682) -- (1.050,0.682);
\draw[line width=1.1pt,barcol] (0,0.744) -- (1.150,0.744);
\draw[line width=1.1pt,barcol] (0,0.806) -- (1.250,0.806);
\draw[line width=1.1pt,barcol] (0,0.868) -- (1.350,0.868);
\draw[line width=1.1pt,barcol] (0,0.930) -- (1.450,0.930);
\draw[line width=1.1pt,barcol] (0,0.992) -- (1.550,0.992);
\draw[line width=1.1pt,barcol] (0,1.054) -- (1.650,1.054);
\draw[line width=1.1pt,barcol] (0,1.116) -- (1.750,1.116);
\draw[line width=1.1pt,barcol] (0,1.178) -- (1.850,1.178);
\draw[line width=1.1pt,barcol] (0,1.240) -- (1.950,1.240);
\draw[line width=1.1pt,barcol] (0,1.302) -- (2.050,1.302);
\draw[line width=1.1pt,barcol] (0,1.364) -- (2.150,1.364);
\draw[line width=1.1pt,barcol] (0,1.426) -- (2.250,1.426);
\draw[line width=1.1pt,barcol] (0,1.488) -- (2.350,1.488);
\draw[line width=1.1pt,barcol] (0,1.550) -- (2.450,1.550);
\draw[line width=1.1pt,barcol] (0,1.612) -- (2.550,1.612);
\draw[line width=1.1pt,barcol] (0,1.674) -- (2.650,1.674);
\draw[line width=1.1pt,barcol] (0,1.736) -- (2.750,1.736);
\draw[line width=1.1pt,barcol] (0,1.798) -- (2.850,1.798);
\draw[line width=1.1pt,barcol] (0,1.860) -- (2.950,1.860);
\draw[line width=1.1pt,barcol] (0,1.922) -- (3.050,1.922);
\draw[line width=1.1pt,barcol] (0,1.984) -- (3.150,1.984);
\draw[line width=1.1pt,barcol] (0,2.046) -- (3.250,2.046);
\draw[line width=1.1pt,barcol] (0,2.108) -- (3.350,2.108);
\node[font=\small] at (1.7,-0.7) {(a) $\Phi_k$: uniform deaths};
\end{scope}
\begin{scope}[xshift=.32\linewidth]
\draw[->] (0,0) -- (3.7,0) node[below left,font=\scriptsize] {$s=t/k$};
\draw[->] (0,0) -- (0,2.35);
\draw[dashed,black!40] (3.4,0) -- (3.4,2.25) node[above,font=\scriptsize,black] {$1$};
\draw[line width=1.1pt,barcol] (0,0.062) -- (0.014,0.062);
\draw[line width=1.1pt,barcol] (0,0.124) -- (0.014,0.124);
\draw[line width=1.1pt,barcol] (0,0.186) -- (0.014,0.186);
\draw[line width=1.1pt,barcol] (0,0.248) -- (0.014,0.248);
\draw[line width=1.1pt,barcol] (0,0.310) -- (0.014,0.310);
\draw[line width=1.1pt,barcol] (0,0.372) -- (0.014,0.372);
\draw[line width=1.1pt,barcol] (0,0.434) -- (0.021,0.434);
\draw[line width=1.1pt,barcol] (0,0.496) -- (0.049,0.496);
\draw[line width=1.1pt,barcol] (0,0.558) -- (0.101,0.558);
\draw[line width=1.1pt,barcol] (0,0.620) -- (0.155,0.620);
\draw[line width=1.1pt,barcol] (0,0.682) -- (0.213,0.682);
\draw[line width=1.1pt,barcol] (0,0.744) -- (0.272,0.744);
\draw[line width=1.1pt,barcol] (0,0.806) -- (0.332,0.806);
\draw[line width=1.1pt,barcol] (0,0.868) -- (0.392,0.868);
\draw[line width=1.1pt,barcol] (0,0.930) -- (0.455,0.930);
\draw[line width=1.1pt,barcol] (0,0.992) -- (0.527,0.992);
\draw[line width=1.1pt,barcol] (0,1.054) -- (0.613,1.054);
\draw[line width=1.1pt,barcol] (0,1.116) -- (0.716,1.116);
\draw[line width=1.1pt,barcol] (0,1.178) -- (0.818,1.178);
\draw[line width=1.1pt,barcol] (0,1.240) -- (0.915,1.240);
\draw[line width=1.1pt,barcol] (0,1.302) -- (1.008,1.302);
\draw[line width=1.1pt,barcol] (0,1.364) -- (1.126,1.364);
\draw[line width=1.1pt,barcol] (0,1.426) -- (1.278,1.426);
\draw[line width=1.1pt,barcol] (0,1.488) -- (1.442,1.488);
\draw[line width=1.1pt,barcol] (0,1.550) -- (1.608,1.550);
\draw[line width=1.1pt,barcol] (0,1.612) -- (1.785,1.612);
\draw[line width=1.1pt,barcol] (0,1.674) -- (1.983,1.674);
\draw[line width=1.1pt,barcol] (0,1.736) -- (2.175,1.736);
\draw[line width=1.1pt,barcol] (0,1.798) -- (2.358,1.798);
\draw[line width=1.1pt,barcol] (0,1.860) -- (2.541,1.860);
\draw[line width=1.1pt,barcol] (0,1.922) -- (2.726,1.922);
\draw[line width=1.1pt,barcol] (0,1.984) -- (2.916,1.984);
\draw[line width=1.1pt,barcol] (0,2.046) -- (3.105,2.046);
\draw[line width=1.1pt,barcol] (0,2.108) -- (3.295,2.108);
\draw[decorate,decoration={brace,mirror,amplitude=3pt},black!60]
  (-0.12,0.062) -- (-0.12,0.51)
  node[midway,left=3pt,font=\scriptsize,align=right,black!60] {early-death\\mass};
\node[font=\small] at (1.7,-0.7) {(b) $\Phi^{\front}_k$: front-loaded deaths};
\end{scope}
\begin{scope}[xshift=.65\linewidth]
\draw[->] (0,0) -- (2.5,0) node[below left,font=\scriptsize] {$\lambda$};
\draw[->] (0,0) -- (0,2.45) node[rotate=90,anchor=south east,font=\scriptsize] {lifespan CDF};
\draw[black!45,dashed] (0.000,0.000) -- (2.200,2.200);
\draw[line width=1pt] (0.000,0.000) -- (0.000,0.035) -- (0.000,0.070) -- (0.000,0.105) -- (0.000,0.140) -- (0.000,0.182) -- (0.000,0.217) -- (0.000,0.252) -- (0.003,0.364) -- (0.010,0.405) -- (0.017,0.438) -- (0.024,0.471) -- (0.043,0.507) -- (0.066,0.551) -- (0.086,0.587) -- (0.105,0.624) -- (0.126,0.660) -- (0.151,0.702) -- (0.172,0.737) -- (0.192,0.771) -- (0.214,0.807) -- (0.235,0.842) -- (0.261,0.885) -- (0.282,0.919) -- (0.305,0.955) -- (0.334,0.994) -- (0.369,1.037) -- (0.398,1.069) -- (0.427,1.098) -- (0.457,1.126) -- (0.500,1.168) -- (0.536,1.204) -- (0.572,1.241) -- (0.608,1.279) -- (0.651,1.325) -- (0.695,1.364) -- (0.739,1.399) -- (0.783,1.429) -- (0.836,1.461) -- (0.880,1.489) -- (0.925,1.516) -- (0.985,1.552) -- (1.044,1.588) -- (1.116,1.629) -- (1.175,1.660) -- (1.234,1.691) -- (1.294,1.720) -- (1.370,1.759) -- (1.434,1.794) -- (1.498,1.829) -- (1.563,1.864) -- (1.639,1.906) -- (1.703,1.941) -- (1.769,1.976) -- (1.835,2.011) -- (1.915,2.053) -- (1.981,2.088) -- (2.047,2.123) -- (2.114,2.158) -- (2.193,2.200);
\node[font=\scriptsize,black!45] at (1.55,0.95) {$\Phi_k$};
\node[font=\scriptsize] at (0.62,1.85) {$F^{\front}$};
\node[font=\small] at (1.25,-0.7) {(c) lifespan CDFs};
\end{scope}
\end{tikzpicture}
\caption{One cohort of the uniform family $\Phi_k$ (a) and of the
front-loaded family $\Phi^{\front}_k$ (b), in rescaled time $s=t/k$,
with the two lifespan CDFs (c).}
\label{fig:families}
\end{figure}
\FloatBarrier

\begin{definition}[the front-loaded instance family]\label{def:frontfamily}
The construction has two integer parameters --- the \emph{cohort size}
$k\ge1$ (the number of queues in a batch, which also sets the
resolution) and the \emph{cohort count} $C\ge1$ --- together with one
fixed input, a front-loaded death profile shared by the whole family.
For each choice of $k$ and $C$ it produces a single instance
$\Phi^{\front}_k(C)$, and the \emph{front-loaded family} is the set of
all of them.

\emph{The profile.} $F^{\front}$ is a fixed probability distribution of
rescaled lifespans $\lambda=\ell/k\in(0,1]$ (a lifespan of $\ell$ slots
rescaled by the cohort size), the same for every instance, given as a
CDF by piecewise-linear interpolation of $310$ knots shipped with the
certificate package (Appendix~\ref{sec:appendix-certs}).  It is the
high-value profile returned by our numerical search for $T_O/T_L$, and it is
\emph{front-loaded}: a concentrated early death mass of $\approx0.212$
lies in the first $1\%$ of the range, with the rest a spread continuum.

\emph{One cohort.} A cohort is $k$ queues born together in one slot
(the born/live/dying/dead vocabulary of Definition~\ref{def:vocab}), with
lifespans chosen to follow $F^{\front}$ by \emph{quantile sampling}:
queue $j$ is given the $\tfrac{j-1/2}{k}$-quantile of $F^{\front}$ (the
$k$ evenly spaced sample points), scaled up to an integer number of
slots in $[1,k]$,
\[
\ell_j=\mathrm{clip}\Bigl(\mathrm{round}\bigl(k\,(F^{\front})^{-1}
\bigl(\tfrac{j-1/2}{k}\bigr)\bigr),\,1,\,k\Bigr),\qquad j=1,\dots,k,
\]
so the cohort's empirical lifespan distribution is the specified
quantile discretization of $F^{\front}$;
queue $j$ is then live for its $\ell_j$ slots after birth and dies.

\emph{The instance.} $\Phi^{\front}_k(C)$ runs $C$ such cohorts in
succession and then lets the buffer drain.  While live, every queue
receives $B$ packets per slot, so it is always saturated, and the
shared buffer has size $B=\lfloor k^2/p\rfloor$ with
$p=3588332140194315\cdot2^{-49}\approx6.3741583392739098$.  Each queue's arrivals therefore occupy the
single interval from its cohort's birth to its death slot, so, like
$\Phi_k(C)$, the instance is \emph{interval-structured}.

\emph{Exact conventions} (those of the certificate package's
generator): round is round-half-to-even, $(F^{\front})^{-1}$ is linear
interpolation in the knot list, and the floor over the dyadic $p$
above is taken in unbounded-integer arithmetic --- pinning $p$
this way avoids the one-slot error a floating-point floor could
introduce (in fluid terms $\theta\approx1/p$; at $k=3\cdot10^5$ this
gives $B=14{,}119{,}511{,}190$).
\end{definition}

\begin{remark}[Tie rule $T_0$]\label{rem:tierule}
At an overflow LQD evicts down to a common level $w$; $R\ge0$ queues
are left one packet above it, and these $R$ marginal packets are the
\emph{residuals} (Proposition~\ref{prop:tiefreedom}).  Tie rule $T_0$
awards them in \emph{decreasing} numerical order (equivalently, it
evicts from the lowest-index tied queues first).  Label queue $j$ of
cohort $c=0,1,\ldots,C-1$ globally by $ck+j$.  The quantile lifespans
are non-decreasing in $j$, and consecutive cohorts are born one cycle
apart, so decreasing global index happens to be latest-death first
(with index order breaking equal-death ties).  That alignment is a
property of the hard instance; the rule itself reads only queue
identities and is not clairvoyant.
\end{remark}

\begin{theorem}[certified instances]\label{thm:certified}
For LQD with the deterministic tie rule $T_0$, the explicit front-loaded death-profile instance in Table~\ref{tab:frontladder}
(its $k=3\cdot10^5$ row) gives
\[
\CR(\LQD)\ \ge\
\frac{184{,}815{,}365{,}566{,}285}
{125{,}784{,}985{,}866{,}185}
\;=\;\frac{36963073113257}{25156997173237}
\;=\;1.46929591\ldots .
\]
\end{theorem}

Table~\ref{tab:frontladder} lists the six certified rows; the ladder
is monotone in $k$ from $k=3\cdot10^3$ on.  The instance is evaluated
under one fixed tie rule, as any concrete evaluation must be;
Theorem~\ref{thm:tieinv} below transfers the value to every
non-clairvoyant tie rule.  The upper bounds of \citet{AEMV24} and of
Theorem~\ref{thm:improved} hold for every tie rule as well.

\begin{table}[!t]
\centering\small
\begin{tabular}{rrrrrl}
\toprule
$k$ & $B$ & $C$ & $\OPT(I)$ & $\LQD(I)$ & ratio $\ge$ \\
\midrule
$400$ & 25{,}101 & 20 & 2{,}617{,}078 & 1{,}802{,}645 & 1.4517988 \\
$3{\cdot}10^3$ & 1{,}411{,}951 & 3000 & 22{,}153{,}437{,}329 & 15{,}083{,}767{,}097 & 1.4686939 \\
$10^4$ & 15{,}688{,}345 & 3000 & 246{,}340{,}227{,}396 & 167{,}676{,}782{,}543 & 1.4691373 \\
$3{\cdot}10^4$ & 141{,}195{,}111 & 3000 & 2{,}217{,}549{,}563{,}033 & 1{,}509{,}296{,}362{,}859 & 1.4692605 \\
$10^5$ & 1{,}568{,}834{,}576 & 2500 & 20{,}534{,}590{,}474{,}926 & 13{,}975{,}921{,}458{,}495 & 1.4692834 \\
$3{\cdot}10^5$ & 14{,}119{,}511{,}190 & 2500 &
184{,}815{,}365{,}566{,}285 &
125{,}784{,}985{,}866{,}185 & \textbf{1.46929591} \\
\bottomrule
\end{tabular}
\caption{The certified ladder of the front-loaded family
(Definition~\ref{def:frontfamily}) under the tie rule $T_0$; the
$k=3\cdot10^5$ row is the paper's lower bound.}
\label{tab:frontladder}
\end{table}

\begin{proof}[Proof of validity]
The rows are legal finite instances (Definition~\ref{def:frontfamily}).  The
\LQD{} value is the definitionally exact profit of the deterministic
policy with tie rule fixed.  The \OPT{} value is the profit of
\LateQD{}, which is exactly optimal by Theorem~\ref{thm:lateqd}.  Its
explicit interval-structured form applies because $\Phi^{\front}_k(C)$
is interval-structured (Definition~\ref{def:frontfamily}): specifically,
that construction assigns each queue its live-slot transmissions from
its own arrival interval $[\gamma_i,\delta_i]$, together with the shared
LQD-equalization level (Theorem~\ref{thm:lateqd}) that \LateQD{}'s
explicit form imposes on the dying and dead queues in the
buffer's leftover space.  Each $\Phi^{\front}_k(C)$ is itself a finite
interval-structured instance, so Theorem~\ref{thm:lateqd} applies to it
directly, terminal drain included.
Both sides are exact integer arithmetic.  Repeating drained copies as in
Remark~\ref{rem:cradditive} turns their exact ratio into a lower bound on
$\CR(\LQD)$ under Definition~\ref{def:switch}.
\end{proof}

\paragraph{The bound does not depend on the tie rule.}
Theorem~\ref{thm:certified} fixes $T_0$, and $\CR(\LQD)$ is the supremum
over tie rules (Definition~\ref{def:lqd}), so nothing more is needed for
the headline claim.  The next theorem shows the certified value holds for every
non-clairvoyant tie rule, not only $T_0$, with the full argument deferred to
Appendix~\ref{sec:appendix-tie}.

\begin{theorem}[tie-rule invariance of the lower bound]\label{thm:tieinv}
Let $I$ be a saturated interval instance, $T_{\mathrm{ref}}$ any tie
rule (not necessarily that of Remark~\ref{rem:tierule}), and $T$
\emph{any} non-clairvoyant tie rule (its choice at each slot may depend
on the run so far, and on randomness, but not on future arrivals).
Then there is an instance $I_T$ --- constructed adaptively around $T$'s
choices, and in general not interval-structured --- with
\[
\LQD_T(I_T)\;=\;\LQD_{T_{\mathrm{ref}}}(I)
\qquad\text{and}\qquad
\OPT(I_T)\;\ge\;\OPT(I)
\]
(for a randomized $T$, on every realization of its choices).
For each $R$, apply this statement to $R$ drained copies of $I$ on
disjoint queue names and to the reference rule that repeats the given
reference run on every copy.  The two reference values scale by $R$, so
Definition~\ref{def:switch} implies that every non-clairvoyant tie rule satisfies
$\CR(\LQD_T)\ge\OPT(I)/\LQD_{T_{\mathrm{ref}}}(I)$.  With the instance
and reference run $T_0$ of Theorem~\ref{thm:certified},
$\CR(\LQD_T)\ge1.46929591\ldots$ for every $T$.
\end{theorem}

\emph{Idea of the proof.}
A tie rule's only freedom is which of the queues tied for longest keeps
a leftover \emph{residual} packet after an overflow
(Proposition~\ref{prop:tiefreedom}); tied queues of the same type are
interchangeable, so this freedom is a mere relabeling.  Rather than run
$T$ on the fixed instance $I$ --- where a lucky $T$ could raise $\LQD$'s
value and weaken the bound --- we build a tailored instance $I_T$ slot
by slot, reacting to each choice $T$ makes; this is well-founded because
$T$ is non-clairvoyant, so its slot-$t$ choice is committed before we
fix the arrivals after $t$.  Whenever $T$'s choice diverges from the
reference run we restore the match, either by swapping the two queues'
future arrival schedules or, when one queue is already dead, by the
zero-cost \emph{exchange move} (Lemma~\ref{lem:exchange},
illustrated in Figure~\ref{fig:exchange}) that hands a
corpse's hoard to a freshly killed queue; the offline normal form of
Lemma~\ref{lem:normalform} keeps this accounting exact.  The result is
$\LQD_T(I_T)=\LQD_{T_{\mathrm{ref}}}(I)$ and $\OPT(I_T)\ge\OPT(I)$, so the
exact ratio transfers; applying the construction to $R$ drained copies
removes the additive term.  The computational certificate itself stays stated for
$T_0$: the theorem transfers the certified \emph{value}, not the
certificate, and for a randomized $T$ it gives the adaptive-adversary
guarantee (on a \emph{fixed} instance a different rule can do better:
the opposite rule, which evicts from the highest-index tied queues and
therefore leaves residuals on lower indices, reaches
$\LQD=1{,}806{,}170$ against $T_0$'s $1{,}802{,}645$ at $k=400$, both
computed by the cross-check engine).

Altogether, Theorem~\ref{thm:certified} certifies
$\CR(\LQD)\ge1.46929591\ldots$ under an explicit tie rule, and
Theorem~\ref{thm:tieinv} extends it to every non-clairvoyant tie rule,
deterministic or randomized.  We turn next to the upper bound: we found
a gap in its published proof, which we repair below.

\section{The erratum: a gap in the published aggregation, and its repair}\label{sec:erratum}

In this section, we present an erratum and repair for a single step
in the proof of \citet{AEMV24}: their payment scheme argument
bounding $\CR(\LQD)$.  The plan is as follows.  In
Section~\ref{sec:aemvdefs}, we restate the definitions we need.  In
Section~\ref{sec:thegap}, we show that the proof of their aggregation
lemma (Lemma~18) does not go through: we first explain in words which
two accounting obligations collide, then construct a minimal abstract
pattern on which the claimed inequality fails by almost $\tfrac12$,
and finally realize the same pattern on an actual switch run.  In
Section~\ref{sec:therepair}, we repair the step by an amortization
over phase transitions, restoring the proof at the cost of an
explicit $\tfrac12$-per-queue charge that we absorb later, in
Section~\ref{sec:upper}.

The deficit is $\alpha$-free and grows with the instance
(Section~\ref{sec:thegap}), so it cannot be absorbed as a bounded
additive term.  Charging it directly to the
payment rate $\varrho$ would cost a full $\tfrac12$ of the rate, moving
AEMV's published bound from $1.6918$ to $\approx2.058$: above the
trivial $2$-competitiveness of \citet{HKM01}.  We therefore handle the
$\tfrac12$ within the aggregation step rather than absorbing it afterwards.
The repaired step is also the basis for the Section~\ref{sec:upper}
improvement.

\subsection{The AEMV scheme, restated from the source}\label{sec:aemvdefs}

\paragraph{The payment scheme, in brief.} AEMV bounded $\CR(\LQD)$ with a
\emph{payment scheme}, an accounting argument summarized in
Figure~\ref{fig:payment}.  Assign each queue $q$ a
\emph{debt} $\hat e_q$, the
share of the $\OPT-\LQD$ gap charged to $q$, so that
$\OPT-\LQD\le\sum_q\hat e_q$.  Treat LQD's own profit as a budget
distributed to the queues as \emph{payments}.  Suppose one can show
that every queue's debt is covered at a fixed \emph{rate} $\varrho>0$:
every queue is paid at least $\varrho\,\hat e_q$,
where $\varrho$ is the quantity to be made as large as possible.  The
scheme's potential makes the exact budget slightly larger than LQD's
profit: telescoping the per-phase inequalities below gives total payments
at most $\LQD+\Psi_{\rm initial}\le\LQD+\alpha B$.  Thus
\[
\LQD+\alpha B\;\ge\;\varrho\sum_q\hat e_q
\;\ge\;\varrho\,(\OPT-\LQD),
\quad\text{which rearranges to}\quad
\OPT\;\le\;\Bigl(1+\tfrac1\varrho\Bigr)\LQD+\tfrac\alpha\varrho\,B:
\]
a larger $\varrho$ gives a
better bound.  Both of our upper-bound edits act on this scheme:
we repair the step that aggregates a queue's per-phase payments into its
total (the erratum, their Lemma~18), and we replace their relaxed final
per-packet estimate of the worst-case payment by the exact envelope
infimum of its continuum relaxation at fixed $\alpha$ (the
improvement), which raises $\varrho$.

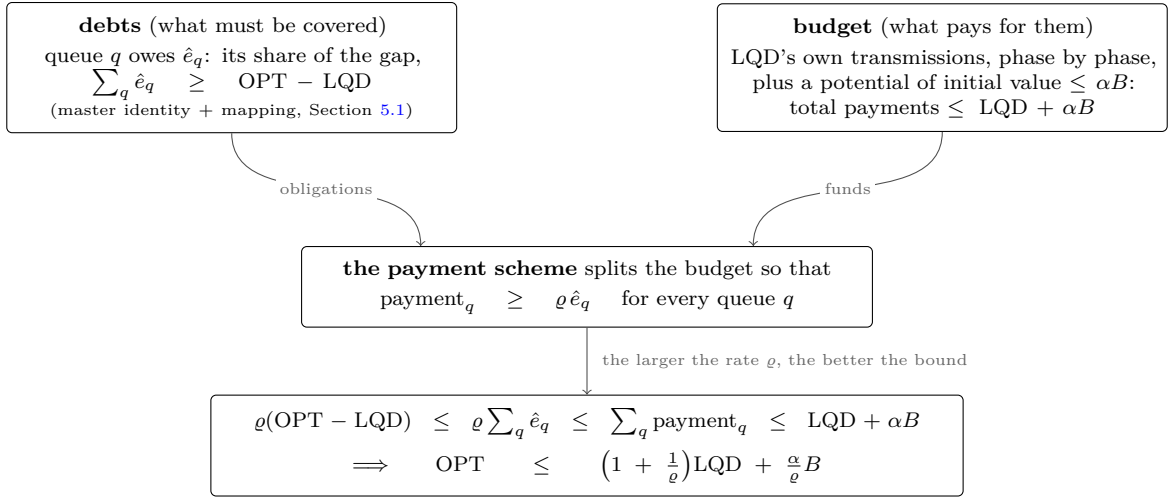
\begin{figure}[!t]
\centering
\begin{tikzpicture}[
  box/.style={draw,rounded corners=2pt,align=center,font=\scriptsize,
              inner sep=5pt,text width=5.6cm},
  wire/.style={->,black!70},
  lab/.style={font=\tiny,black!60,align=center,fill=white,inner sep=1.5pt}]
\node[box] (debt) at (0,0)
  {\textbf{debts} (what must be covered)\\[2pt]
   queue $q$ owes $\hat e_q$: its share of the gap,\\
   $\sum_q\hat e_q\;\ge\;\OPT-\LQD$\\[1pt]
   {\tiny (master identity $+$ mapping, Section~\ref{sec:aemvdefs})}};
\node[box] (budget) at (9.4,0)
  {\textbf{budget} (what pays for them)\\[2pt]
   LQD's own transmissions, phase by phase,\\
   plus a potential of initial value $\le\alpha B$:\\
   total payments $\le\LQD+\alpha B$};
\node[box,text width=7.2cm] (rate) at (4.7,-2.9)
  {\textbf{the payment scheme} splits the budget so that\\[2pt]
   $\operatorname{payment}_q\;\ge\;\varrho\,\hat e_q$ \quad for every queue $q$};
\node[box,text width=9.6cm] (concl) at (4.7,-5.0)
  {$\varrho(\OPT-\LQD)\;\le\;\varrho\sum_q\hat e_q\;\le\;\sum_q\operatorname{payment}_q\;\le\;\LQD+\alpha B$\\[3pt]
   $\Longrightarrow\quad \OPT\;\le\;\bigl(1+\tfrac1\varrho\bigr)\LQD+\tfrac\alpha\varrho B$};
\draw[wire] (debt.south) to[out=-90,in=90] node[lab,pos=0.5]{obligations}
  ([xshift=-2.2cm]rate.north);
\draw[wire] (budget.south) to[out=-90,in=90] node[lab,pos=0.5]{funds}
  ([xshift=2.2cm]rate.north);
\draw[wire] (rate.south) -- node[lab,right=4pt]{the larger the rate $\varrho$, the better the bound}
  (concl.north);
\end{tikzpicture}
\caption{The payment scheme of \citet{AEMV24} as an accounting
argument: if every queue is paid at least $\varrho\,\hat e_q$ out of
LQD's profit plus an $O(B)$ potential, then
$\CR(\LQD)\le1+1/\varrho$.}
\label{fig:payment}
\end{figure}
\FloatBarrier

We restate below, from the journal version of \citet{AEMV24}'s paper,
everything we use, fix, or claim broken.  The displayed imported claims
are unchanged up to notation; surrounding definitions are expanded for
self-containedness, and each item carries its exact source number.  We quote the openly available accepted
manuscript (arXiv~v2, corresponding to ACM~Trans.\ Algorithms
20(4), 2024) and use its numbering throughout, since the ACM version of
record renumbers items by section.  \emph{Proofs} of the imported
statements are by \citet{AEMV24}.

Notational changes: their buffer size $M$ is our
$B$, and their per-queue profit counter $\Phi_q$ is written as $P_q$ here
(so that the letter $\Phi$ keeps a single meaning, namely the instance
families $\Phi_k,\Phi^{\front}_k$; see the notation table in
Section~\ref{sec:model}).  When a queue $q$ is fixed by context we drop the subscript $q$.

Table~\ref{tab:aemvdeps} summarizes what we import, what each
import is used for, and where our two upper-bound contributions
(the erratum and the improvement) intervene.

\begin{table}[!t]
\centering\small
\begin{tabular}{lll}
\toprule
Imported from \citet{AEMV24} & used for & status here \\
\midrule
Section~2 master identity, Section~5 mapping & per-queue reduction & unchanged \\
Lemmas 15--16 & per-phase L-increase bounds & unchanged \\
Lemma 18 aggregation & running-max envelope bound &
\textbf{repaired} (Cor.~\ref{cor:repaired}) \\
Claim 19, Lemma 21 & per-packet endgame relaxation &
\textbf{replaced} (Thm.~\ref{thm:improved}) \\
\bottomrule
\end{tabular}
\caption{What we import from \citet{AEMV24} and where our two edits
intervene.}
\label{tab:aemvdeps}
\end{table}

In the rest of this subsection, we restate the scheme in seven blocks:
setup; one standing assumption; overflows, key steps, and phases; the
splitting scheme; live and dying queues and the average level; the two
per-phase bounds; and the aggregation with its envelope.  Each block is
tagged with its exact source location.  A reader who takes the imported AEMV
payment scheme lemmas as given can skip ahead to the erratum (``The
gap'', below) or straight to Theorem~\ref{thm:improved} in
Section~\ref{sec:upper}, returning here only when a term tagged
{\rm[AEMV, \ldots]} needs tracing back to its definition.

\paragraph{Setup and the master identity {\rm[AEMV, Section~2]}.}
Fix an instance and run LQD and OPT side by side; $s^t_{\mathrm{LQD}}(q)$
and $s^t_{\mathrm{OPT}}(q)$ denote queue $q$'s size in the respective
buffer at step $t$ (after arrivals and evictions, before transmission),
and $s^t_{\max}=\max_q s^t_{\mathrm{LQD}}(q)$.  OPT transmits an
\emph{OPT-extra} packet from $q$ at step $t$ if it transmits from $q$
while LQD's copy of $q$ is empty; \emph{LQD-extra} packets are defined
symmetrically.  With $\OPTX$ and $\LQDX$ the total counts, the \emph{master identity}
\[
\OPT-\OPTX\;=\;\LQD-\LQDX
\quad\text{holds, hence}\quad
\frac{\OPT}{\LQD}\;=\;1+\frac{\OPTX-\LQDX}{\LQD}.
\]
Writing $e_q$ for the number of
OPT-extra packets transmitted from $q$ (so $\OPTX=\sum_q e_q$), their
Section~5 mapping procedure assigns to each queue $q$ some number
$m_q\le e_q$ of transmitted LQD-extra packets, with $\sum_q m_q\le\LQDX$.
Setting $\hat e_q=e_q-m_q\ge0$, it suffices to assign each queue a payment of
at least $\varrho\,\hat e_q$ out of the LQD profit to conclude
$\OPT\le(1+1/\varrho)\LQD+(\alpha/\varrho)B$, once the potential budget
described below is included.

\paragraph{The standing assumption (A1) {\rm[AEMV, Section~3]}.}
If $s^t_{\mathrm{LQD}}(q)\le1$ and $q$ has received a packet at or
before $t$, then nothing arrives to $q$ after $t$ (w.l.o.g.: later
packets are redirected to fresh queues, which does not change LQD's
profit and cannot decrease OPT's).  Consequently a queue that drains for
LQD never refills.

\paragraph{Overflows, key steps, and phases
{\rm[AEMV, Def.~2 and Section~3]}.}
Queue $q$ \emph{overflows} at step $t$ if a packet destined to $q$ is
evicted by LQD at $t$, or if the LQD buffer is full and
\[
s^t_{\mathrm{LQD}}(q)\;\ge\;\max\{s^t_{\max}-1,\;1\}.
\]
The one-packet slack below the maximum is deliberate: two queues that
overflow at the same step may end one packet apart.  The eviction
clause matters only when a packet of $q$ (arriving or stored) is
evicted at $t$ and $q$ ends the phase \emph{empty}: the size clause
requires $s^t_{\mathrm{LQD}}(q)\ge1$ and so would miss this
case.
Write $t_q$ for $q$'s \emph{key step}: the last step at which $q$
overflows, or $t_q=-1$ if $q$ never overflows.  No packet of $q$,
stored or arriving, is evicted after $t_q$.  The \emph{phase boundaries}
$\tau_1<\tau_2<\dots$ are the distinct last-overflow steps (pooled over
all queues, so the phase grid is global, shared by every queue), and phase
$i$ is $[\tau_i,\tau_{i+1})$: this is AEMV's overflow-indexed sense
of ``phase,'' unrelated to the arrival/transmission phases of a single
slot in Definition~\ref{def:switch}.  The last phase runs to the end of
the instance: for the final index $\ell$, AEMV set
$\tau_{\ell+1}=T+1$, one past the last active step $T$.  This choice
makes the phase grid cover every step from $\tau_1$ on; in particular,
every OPT-extra packet is transmitted before the final boundary, so
the per-queue debts defined below have all drained to $0$ by then.

The analysis below uses several quantities keyed to this phase grid;
here we depart from AEMV's notation in one respect.  They superscript
\emph{both} the time step and the phase index, so that $x^i$ means the
value at step $\tau_i$, not at step $i$.  We instead reserve the superscript for a step and
mark phase $i$ with a subscript.  For a \emph{state} quantity (one that
also has a step form) the subscript is the boundary value,
$x_i:=x^{\tau_i}$ (e.g.\ $\hat e_{q,i}$, $\sigma_{q,i}$, $\Psi_i$); for a
\emph{flow} quantity accumulated over the phase (no step form) it is the
total over $[\tau_i,\tau_{i+1})$ (e.g.\ $o_{q,i}$, $o_i$, $n_i$,
$\LQD_i$, $u_i$, $v_i$).  Every displayed formula below is AEMV's under this
relabeling.  Sizes keep the step superscript:
$s^t_{\mathrm{LQD}}(q)$, or $s^{\tau_i}_{\mathrm{LQD}}(q')$ when read at a
boundary.  For a state quantity $X$ we write $\Delta_i X:=X_{i+1}-X_i$ for
its change across phase $i$; the payments $\Delta_iP_q$ and the potential
change $\Delta_i\Psi$ below are instances of this notation.

Let $e^t_q$ count the OPT-extra packets transmitted from $q$ at step
$t$ \emph{or later} (so $e_q=e^1_q$ is the total of the Setup
paragraph).  Then $\hat e^t_q:=\max\{e^t_q-m_q,0\}$ is queue $q$'s
debt still outstanding as of step $t$; it holds at the net value
$\hat e_q$ while LQD still stores $q$ (that is, while LQD's copy of
$q$ is nonempty), then
falls by one per step through the \emph{drained tail} (an OPT-extra packet needs
LQD's copy of $q$ empty, so all of them fall in that tail).  Two
structural facts follow from assumption (A1): once LQD's copy drains it never
refills, and OPT's copy receives nothing thereafter while transmitting
automatically whenever nonempty; so the OPT-extra transmissions
occupy consecutive steps, and the debt falls by exactly one per
step through the tail.  Call a step $t$ with $\hat e^t_q>0$ but
$s^t_{\mathrm{LQD}}(q)=0$ a \emph{correction step} (the term of
Lemma~16 below); the correction steps are exactly these decay steps.

\paragraph{The splitting scheme {\rm[AEMV, Section~3, eq.~(2); Section~6, eq.~(13)]}.}
The scheme carries two parameters $\alpha,\beta\in(0,1)$ with
$\alpha+\beta<1$; they enter through the potential and the payment
display below.  Let $o_{q,i}$ be
the number of packets LQD transmits from $q$ during phase $i$, let
$o_i$ be the total of these over queues already past their last
overflow from which OPT transmits at least one extra packet at some
point in the run ($t_{q'}\le\tau_i$ and $e_{q'}>0$; the prime marks a
running queue index), and let
$n_i$ be the analogous total over all other queues (so $\LQD_i=n_i+o_i$).  The
potential is $\Psi_i=\alpha\,|\mathcal A_i|$, where $\mathcal A_i$ is the
set of queues \emph{active} at $\tau_i$ (nonempty in LQD's or OPT's
buffer, i.e.\ $s^{\tau_i}_{\mathrm{LQD}}(q)\ge1$ or
$s^{\tau_i}_{\mathrm{OPT}}(q)\ge1$) whose \emph{last} overflow still lies
ahead ($t_q>\tau_i$, so they overflow again after $\tau_i$).
Let $u_i$ be the number of queues newly entering $\mathcal A$ during phase
$i$, and $v_i$ the number of queues active at $\tau_i$ that last-overflow at
$\tau_{i+1}$, hence leave $\mathcal A$ for good (membership requires the
last overflow still ahead; and by assumption (A1) a queue cannot exit
$\mathcal A$ by emptying and later return, so entries and exits are
the only two flows); then $\Delta_i\Psi=\alpha\,(u_i-v_i)$.  The
scheme assigns each $q$ with $t_q\le\tau_i$ and $\hat e_{q,i}>0$ the payment
$\Delta_iP_q=\frac{\hat e_{q,i}}{\hat e_{\bullet,i}}\,(n_i+\alpha\,o_i
-\Delta_i\Psi)+(1-\alpha)\,o_{q,i}$ with
$\hat e_{\bullet,i}=\sum_{q':\,t_{q'}\le\tau_i}\hat e_{q',i}$ (the $\bullet$
marks summation over the queue index) [eq.~(2)].  In
their Section~6 this is re-bracketed with the second parameter $\beta$
into eq.~(13):
\[
\Delta_iP_q\;=\;
\underbrace{\frac{\hat e_{q,i}}{\hat e_{\bullet,i}}\,
\bigl(n_i+\alpha\,o_i-\Delta_i\Psi\bigr)+\beta\,o_{q,i}}_{\text{L-increase}}
\;+\;\underbrace{(1-\alpha-\beta)\,o_{q,i}}_{\text{S-increase}}.
\]
These are AEMV's labels: the \emph{L-increase}, the component that
Lemmas 15--16 below lower-bound, is the object used throughout the rest
of this section and the erratum; the \emph{S-increase} is the
leftover $(1-\alpha-\beta)\,o_{q,i}$ from LQD's own transmissions in phase
$i$.  The per-phase totals are feasible:
$\sum_q\Delta_iP_q\le\LQD_i-\Delta_i\Psi$.
Moreover $|\mathcal A_i|\le B$: an active queue whose last overflow is
still ahead must be nonempty in LQD at $\tau_i$ (otherwise (A1) rules out
the future arrival needed to overflow).  Hence $0\le\Psi_i\le\alpha B$,
and telescoping the last display over all phases gives the budget
$\sum_{i,q}\Delta_iP_q\le\LQD+\alpha B$ used above.

\paragraph{Live and dying queues; the average level
{\rm[AEMV, Def.~4, Obs.~5, eq.~(3)]}.}
This is AEMV's own \emph{relative} notion of live/dying: a relation
between two queues at a step, indexed by a reference queue $q$.  It is
unrelated to the intrinsic live/dead/dying vocabulary of
Definition~\ref{def:vocab} (which describes a single queue's own
arrival interval); the two vocabularies are never mixed, and only the
relative notion is used in this section.
Fix $q$ with $\hat e_q>0$, and let $j_q:=\min\{j:\hat e_{q,j}=0\}$ be
the index of the earliest phase boundary at which $q$'s debt has fully
drained; $j_q$ is well-defined since the debt is $0$ at the final boundary
(possibly $j_q=\ell+1$).  Because $j_q$ indexes phase
\emph{boundaries}, which are coarser than steps, $\tau_{j_q}$ is the
first phase boundary at or after the right endpoint of the pending
window defined below; it need not equal that endpoint.  Here $j_q$ counts
the \emph{debt} $\hat e_q$ down to $0$, not the raw total $e_q$: the
two part ways by $m_q$, since $\hat e^t_q=\max\{e^t_q-m_q,0\}$ reaches
$0$ that many steps earlier in the drained tail.  For a phase $i$ with
$\tau_i\ge t_q$ and $i\le j_q$, a queue $q'$ storing LQD packets at
$\tau_i$ is
\emph{live w.r.t.\ $q$} if (i) $\tau_i<t_{q'}$, or (ii) $e_{q'}=0$ (that
is, $q'$ transmits no OPT-extra packet at all, in the atemporal sense of
the Setup paragraph above) and $s^{\tau_{j_q}}_{\mathrm{LQD}}(q')\ge1$;
otherwise
\emph{dying}.  Condition~(ii)'s size is read at the fixed boundary
$\tau_{j_q}$: the same for every phase $i$, not at $\tau_i$.  A
never-OPT-extra queue $q'$ counts as live for $q$'s analysis exactly
when it is still nonempty at the boundary where $q$'s debt has drained;
by (A1), a queue nonempty there was nonempty at every earlier boundary
too, since a queue that drains never refills.  Dying queues
never become live; a live queue becomes dying only by
last-overflowing [Obs.~5].  $\mathcal L_{q,i}$ denotes the live set at
$\tau_i$, and $\sigma_{q,i}$ is the average of $s^{\tau_i}_{\mathrm{LQD}}(q')$
over $q'\in\mathcal L_{q,i}$ (set to $1$ if the set is empty); if
$\hat e_{q,i}>0$ then the live set is nonempty and $\sigma_{q,i}\ge2$
[Lem.~13 and Obs.~8].

\paragraph{The two per-phase lower bounds on the L-increase
{\rm[AEMV, Lemmas 15 and 16]}.}
Both bounds use the average level $\sigma_{q,i}$ defined above.  For any queue $q$ with $\hat e_q>0$ and any phase $i$ with
$\tau_i\ge t_q$ and $\hat e_{q,i}>0$:
\begin{itemize}
\item {\rm[Lemma 15]} (rise, $\sigma_{q,i+1}\ge\sigma_{q,i}$)
the L-increase in phase $i$ is at least
\[
\hat e_{q,i}\cdot\frac{\alpha\cdot\max\bigl\{\sigma_{q,i+1}-\sigma_{q,i},\;
1-\sigma_{q,i}/\sigma_{q,i+1}\bigr\}+(\tau_{i+1}-\tau_i)}
{\sigma_{q,i}-1}\,;
\]
inside the proof of Lemma~18, AEMV weaken this to their display~(29),
the exact rise-phase provision used below:
\begin{equation*}\tag{AEMV-29}\label{eq:aemv29}
\sum_{t=\tau_i}^{\tau_{i+1}-1}
 \frac{\hat e^t_q}{\sigma_{q,i}-1}
+\alpha\hat e^{\tau_{i+1}-1}_q
 \frac{\sigma_{q,i+1}-\sigma_{q,i}}{\sigma_{q,i}-1}.
\end{equation*}
\item {\rm[Lemma 16; display (30) in the proof of Lemma 18]}
(drop, $\sigma_{q,i+1}<\sigma_{q,i}$)
for $\beta\ge1-\sqrt{1-\alpha}-\alpha/2$, the L-increase is at least
\begin{equation*}\tag{AEMV-30}\label{eq:aemv30}
\sum_{t=\tau_i}^{\tau_{i+1}-1}\frac{\hat e^t_q}{\sigma_{q,i}-1}
\;-\;\frac{g_{q,i}}{2(\sigma_{q,i}-1)}
\;-\;\hat e_{q,i+1}\Bigl(\frac{1}{\sigma_{q,i+1}-1}
-\frac{1}{\sigma_{q,i}-1}\Bigr),
\end{equation*}
where $g_{q,i}$ is the number of steps $t\in[\tau_i,\tau_{i+1})$ with
$\hat e^t_q>0$ and $s^t_{\mathrm{LQD}}(q)=0$ (the \emph{correction
steps}), and the last term is void if $\hat e_{q,i+1}=0$.
\end{itemize}
Our theorem's choice $\alpha=613/1000$, $\beta=1-\sqrt{1-\alpha}-\alpha/2$
satisfies Lemma~16's condition with equality (it is the smallest
permitted $\beta$).  Call $[\,t_q,\ \min\{t\ge t_q:\hat e^t_q=0\}\,)$ the \emph{pending
window} for $q$.  Its start is pinned to $t_q$, not to the first step
with positive debt: $\hat e^t_q$ counts transmissions at step $t$ or
later, so it already sits at the full $\hat e_q$ long before $q$
overflows.  The phases and steps below
are always understood within this window: outside it, either no phase
of $q$'s analysis exists yet, or $\hat e^t_q=0$ carries no payment
obligation.

\paragraph{The aggregation and its envelope {\rm[AEMV, Lemma 18]}.}
Their Lemma~18 carries the standing hypothesis $\alpha\le2/3$, which
both of this paper's parameter choices satisfy
(Table~\ref{tab:constants}); every statement below that routes through
the aggregation display inherits it.  The lemma aggregates the two per-phase bounds along a non-decreasing
\emph{envelope}.  Write $s_i:=\sigma_{q,i}-1$; let $i_q$ be the phase
containing $t_q$, and set $a_{i_q}=\lceil\sigma_{q,i_q}\rceil-1$
(rounded up so that the base level $b_0:=a_{i_q}$ is a positive
integer, as the sums in~\eqref{eq:thirtytwo} require; the later levels
stay real).  For
$i>i_q$, set the running maximum $a_i=\max(s_i,a_{i-1})$.  Writing $i(t)$
for the phase containing step $t$, re-index the phase sequence by steps
via $b_j=a_{i(t_q+j)}$; this gives the non-decreasing sequence
$b_0\le b_1\le\dots$ with $b_0=\lceil\sigma_{q,i_q}\rceil-1$.  The
step at a phase boundary uses the new phase's $a$-level.  After the
final boundary, extend the sequence constantly; the debts and
correction indicators are zero there.  Thus every $b_{j+1}$ below is
defined.  The envelope is a bookkeeping device
(Figure~\ref{fig:envelope}).

For later auditability, here is the unnumbered phase-level expression in
AEMV's proof before it is re-indexed by steps.  Restricting the sum to
phases at or after $t_q$ that meet the pending window, and writing
$g_i=g_{q,i}$, set
\begin{equation}\label{eq:envphase}
\mathrm{ENV}:=
\sum_{i\ge i_q:\,\hat e_{q,i}>0}\left[
 \sum_{t=\tau_i}^{\tau_{i+1}-1}\frac{\hat e^t_q}{a_i}
 +\alpha\hat e^{\tau_{i+1}-1}_q\frac{a_{i+1}-a_i}{a_i}
 -\frac{g_i}{2a_i}\right].
\end{equation}
The sum in~\eqref{eq:envphase} keeps the rise term into the terminal
(zero-debt) phase even though that phase adds no further dwell or
correction.  In particular, if the debt first
reaches zero exactly at $\tau_{i+1}$, then
$\hat e_q^{\tau_{i+1}-1}=1$ even though $\hat e_{q,i+1}=0$; the next
envelope level is still defined by the same running-maximum rule.

Equivalently, put $g_q^t=1$ at a correction step and $0$ otherwise.
Re-indexing \eqref{eq:envphase} by steps gives AEMV's display~(31):
\begin{equation*}\tag{AEMV-31}\label{eq:aemv31}
\sum_{t\ge t_q}\left[
 \hat e_q^t\frac{\alpha(b_{t-t_q+1}-b_{t-t_q})+1}{b_{t-t_q}}
 -\frac{g_q^t}{2b_{t-t_q}}\right].
\end{equation*}

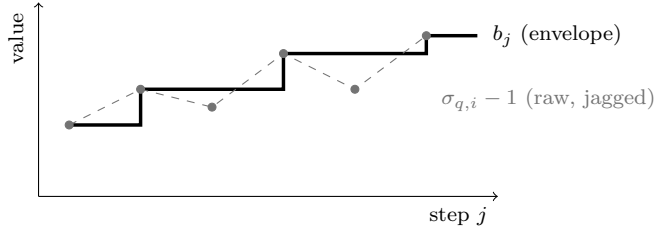
\begin{figure}[!t]
\centering
\begin{tikzpicture}[scale=1.35]
\draw[->] (-0.3,0) -- (4.2,0) node[below left,font=\scriptsize] {step $j$};
\draw[->] (-0.3,0) -- (-0.3,1.9) node[rotate=90,anchor=south east,font=\scriptsize,yshift=2pt] {value};
\draw[line width=1.3pt] (0,0.70) -- (0.7,0.70) -- (0.7,1.05) -- (2.1,1.05)
                          -- (2.1,1.40) -- (3.5,1.40) -- (3.5,1.575) -- (4.0,1.575);
\draw[dashed,black!55] (0,0.70) -- (0.7,1.05) -- (1.4,0.875) -- (2.1,1.40)
                          -- (2.8,1.05) -- (3.5,1.575);
\foreach \x/\y in {0/0.70, 0.7/1.05, 1.4/0.875, 2.1/1.40, 2.8/1.05, 3.5/1.575}
  \filldraw[black!55] (\x,\y) circle (1.1pt);
\node[font=\scriptsize,anchor=west] at (4.05,1.575) {$b_j$ (envelope)};
\node[font=\scriptsize,anchor=west,black!55] at (3.55,0.95) {$\sigma_{q,i}-1$ (raw, jagged)};
\end{tikzpicture}
\caption{The envelope $b_j$: the running maximum of the raw per-phase
levels, non-decreasing even where $\sigma_{q,i}-1$
dips.}\label{fig:envelope}
\end{figure}
\FloatBarrier
The aggregated claim [their display (32)] is that, with
$m=b_0+\hat e_q$, the total L-increase is at least
\begin{equation*}\tag{AEMV-32}\label{eq:thirtytwo}
\sum_{j=0}^{b_0-1}\frac{\alpha(b_{j+1}-b_j)+1}{b_j}\,\hat e_q
\;+\;\sum_{j=b_0}^{m-1}\Bigl(\frac{\alpha(b_{j+1}-b_j)+1}{b_j}\,(\hat e_q+b_0-j)
-\frac{1}{2b_j}\Bigr),
\end{equation*}
and the S-increases contribute at least $(1-\alpha-\beta)\,b_0$ (LQD
transmits at least $b_0$ packets from $q$ at steps $\ge t_q$, their
Section~6.2).  Their remark between displays (32) and (33) states that
this bound ``holds for any positive integers $b_0$ and $\hat e_q$ and
any non-decreasing sequence of \emph{positive numbers}
$b_0,b_1,\ldots$.''  The display is evaluated at the running-max
sequence, whose levels are real (averages), but the remark licenses
treating it as a function of a free real-valued envelope, more general
than the specific running-max sequence, and bounding that function
from below.  This is the optimization Corollary~\ref{cor:repaired}
passes to; Section~\ref{sec:upper} lower-bounds it through a
finite-head continuum relaxation whose optimum is evaluated exactly.

\begin{lemma}[pending-window structure]\label{lem:windowstructure}
If $\hat e_q>0$, then $q$ overflows, its key step $t_q$ is a phase
boundary, and its pending window begins with a constant-debt block of
$L\ge b_0$ steps on which LQD transmits from $q$.
\end{lemma}
\begin{proof}
If $q$ never overflowed, LQD would never discard a packet of $q$.
Queue by queue, its backlog would then dominate OPT's, so OPT could not
make an OPT-extra transmission from $q$; this would give $e_q=0$ and
hence $\hat e_q=0$, a contradiction.  Thus $t_q$ exists, and because
the global phase grid consists of the distinct last-overflow steps,
$t_q=\tau_{i_q}$ for some $i_q$.

At $t_q$, the overflow condition gives
$s^{t_q}_{\mathrm{LQD}}(q)\ge s^{t_q}_{\max}-1$.  This is immediate
from its size clause.  In the remaining edge case, a packet of $q$
(arriving or stored) is evicted and $q$ ends empty; because LQD evicts
from a longest queue, all surviving queues then have size at most one,
so the same inequality holds.  AEMV's
Observation~7 gives
$s^{t_q}_{\max}\ge\lceil\sigma_{q,i_q}\rceil$.  Hence LQD stores at
least $\lceil\sigma_{q,i_q}\rceil-1=b_0$ packets of $q$ at that step.
After $t_q$ it discards no further packet of $q$, so these packets are
transmitted in at least $b_0$ consecutive nonempty steps.  No
OPT-extra transmission from $q$ occurs while LQD's copy is nonempty,
so the outstanding debt stays at its full value $\hat e_q$ throughout
this block.  Consequently its length $L$ satisfies $L\ge b_0$.
\end{proof}

Figure~\ref{fig:window} summarizes the resulting anatomy of a pending
window: a constant-debt block of $L\ge b_0$ steps on which LQD still
stores and transmits from $q$, then a drained tail on which the debt
falls by one per step.  Because $\hat e^t_q$ counts transmissions at
step $t$ \emph{or later}, the debt sits at its full value through the
whole block and the first tail step; corrections can occur only on
the tail, where LQD's copy of $q$ is empty; and each block step also
earns the flat S-payment $h=1-\alpha-\beta$, which is how the
analysis of Section~\ref{sec:upper} profits from windows with
$L>b_0$.  The triple $(b_0,\hat e_q,L)$ the figure depicts is exactly
the parametrization that Section~\ref{sec:upper} runs on.

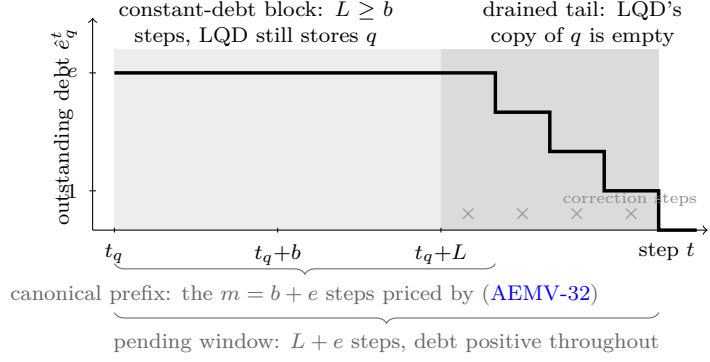
\begin{figure}[!t]
\centering
\begin{tikzpicture}[x=0.72cm,y=0.52cm]
\fill[black!7] (0,0) rectangle (6,4.6);
\fill[black!14] (6,0) rectangle (10,4.6);
\draw[->] (-0.4,0) -- (10.9,0) node[below left=1pt,font=\scriptsize]{step $t$};
\draw[->] (-0.4,0) -- (-0.4,5.4) node[rotate=90,anchor=south east,font=\scriptsize,yshift=2pt]{outstanding debt $\hat e^t_q$};
\draw (-0.46,4) -- (-0.34,4) node[left=3pt,font=\scriptsize]{$e$};
\draw (-0.46,1) -- (-0.34,1) node[left=3pt,font=\scriptsize]{$1$};
\draw (0,-0.06) -- (0,0.06); \node[below,font=\scriptsize] at (0,-0.1) {$t_q$};
\draw (3,-0.06) -- (3,0.06); \node[below,font=\scriptsize] at (3,-0.1) {$t_q{+}b$};
\draw (6,-0.06) -- (6,0.06); \node[below,font=\scriptsize] at (6,-0.1) {$t_q{+}L$};
\draw[line width=1.3pt]
  (0,4) -- (7,4) -- (7,3) -- (8,3) -- (8,2) -- (9,2) -- (9,1) -- (10,1) -- (10,0) -- (10.7,0);
\foreach \x in {6.5,7.5,8.5,9.5}{
  \node[font=\scriptsize,black!45] at (\x,0.42) {$\times$};}
\node[font=\tiny,black!45,anchor=west] at (8.05,0.78) {correction steps};
\node[font=\scriptsize,align=center,text width=3.9cm] at (2.6,5.25)
  {constant-debt block: $L\ge b$ steps, LQD still stores $q$};
\node[font=\scriptsize,align=center,text width=3.1cm] at (8.6,5.25)
  {drained tail: LQD's copy of $q$ is empty};
\draw[decorate,decoration={brace,mirror,amplitude=4pt},black!60]
  (0,-0.85) -- (7,-0.85) node[midway,below=4pt,font=\scriptsize,black!60]
  {canonical prefix: the $m=b+e$ steps priced by \eqref{eq:thirtytwo}};
\draw[decorate,decoration={brace,mirror,amplitude=4pt},black!60]
  (0,-2.1) -- (10,-2.1) node[midway,below=4pt,font=\scriptsize,black!60]
  {pending window: $L+e$ steps, debt positive throughout};
\end{tikzpicture}
\caption{One queue's pending window, drawn for $(b,e,L)=(3,4,6)$: the
debt holds at $e$ through the block and the first tail step, then
falls by one per step; corrections ($\times$) occur only on the tail.}
\label{fig:window}
\end{figure}
\FloatBarrier

\subsection{The gap}\label{sec:thegap}

In this subsection we present the gap itself.  We first locate the
step of AEMV's Lemma~18 where the aggregation fails, then construct a
minimal pending-window pattern on which its deficit approaches
$\tfrac12$ (Figure~\ref{fig:countermodel}), and finally show that the
resulting gap cannot be absorbed as a harmless $O(B)$ additive
constant.

\paragraph{Terms.} Everything below uses the vocabulary of the
restatement above: one queue's \emph{pending window}, cut into
\emph{phases}, with levels $s_i=\sigma_{q,i}-1$ and their running maxima
$a_i$, per-step debts $\hat e^t$, and correction steps.  Two nicknames are new: the
per-step summand $\hat e^t/s_i$ in (30) is called a phase's \emph{dwell}
term, and a phase at whose end the level falls is a \emph{drop} phase (the
``drop'' case of Lemma~16, $\sigma_{q,i+1}<\sigma_{q,i}$).  The
aggregation of Lemma~18 replaces each phase's denominator $s_i$ by the
running maximum $a_i$: larger, hence seemingly a pure weakening,
and simpler.  But for the \emph{subtracted} correction terms the larger
denominator strengthens the claim, and the gap is in that replacement.

\paragraph{In words.} In the proof of Lemma~18, the aggregation step
upgrades each phase's bookkeeping from its own level to the running
maximum, and every upgraded claim must be funded by some surplus that
the per-phase facts provide.  Two of these funding obligations can
land on the same term.
\emph{Corrections:} under the upgrade a correction term shrinks (it
divides by the larger running maximum), and the difference is paid out
of the same phase's \emph{dwell surplus}: the slack that the dwell terms
gain from the same upgrade.  \emph{Drops:} when the level falls,
Lemma~16 charges a \emph{drop penalty} (its last, subtracted term),
and the proof pays it with the first dwell term of the following phase
(the \emph{drop-penalty pairing}).  So a phase that follows a drop
starts with its first dwell term already spent; if that very step is
also a correction step, the correction has nothing left to draw on,
and the accounting falls short by up to $\tfrac12$ at that step.
Figure~\ref{fig:countermodel} shows the double spend on the family we
construct next.

\paragraph{The symbolic countermodel.} The shortfall does not depend
on $\alpha$ or $\beta$.  The family below is a countermodel to the
algebraic implication claimed from (29)/(30) to (31)/(32); by itself it
is not asserted to be a switch run.  The realized $n=2$ instance in the
next paragraph supplies a domain-valid counterexample.  The family
isolates the mechanism and shows that the algebraic loss approaches
$\tfrac12$.  Here $\mathrm{FACTS}$ is the sum of the per-phase lower bounds
(Lemma~15 or 16, whichever applies) over the phases of the pending
window.  Fix an integer $n\ge2$.
Consider a window of $n+1$ steps with debts $\hat e^t\equiv1$: the
outstanding debt sits at its full value $\hat e_q=1$ throughout the
window.  Set phase 1 to
span $n$ steps at level $s_1=n$ (average level $\sigma_{q,1}=n+1$,
giving $b_0=\lceil\sigma_{q,1}\rceil-1=n$), and let phase 2 be a
single step at level $s_2=1$ ($\sigma_{q,2}=2$, consistent with
Observation~8).  At step $n+1$, $q$'s \emph{own} level is
$s^t_{\mathrm{LQD}}(q)=0$ while the live-set average is
$\sigma_{q,2}=2$, so that
single step is the window's only correction step ($g_{q,1}=0$,
$g_{q,2}=1$).  Set the algebraic endpoint to
$\sigma_{q,3}=1<2$, so Lemma~16's drop test applies to both phases.
Since $\hat e_{q,3}=0$, the terminal drop penalty is void; any closing
value below $2$ gives the same $\mathrm{FACTS}$.  For $n=2$ the
operative debts, the values $\sigma_{q,1}=3$ and $\sigma_{q,2}=2$, and
the correction pattern all co-occur in the realized run below; its
irrelevant closing level is $4/3$ rather than $1$.  The phase grid is
$\{1,\dots,n\},\{n+1\}$: two drop phases, where step $n+1$ is a
correction step whose predecessor phase drops, so it carries the
dwell term that the drop-penalty pairing has already spent.  The
running-max envelope never sees the drop: $b_j\equiv n$ for every
$j$, so all rise terms $b_{j+1}-b_j$ in \eqref{eq:thirtytwo} vanish.  The pattern and this computation are verified by the
certificate
\certpath{certificates/erratum/aggregation_counterexample.py}
(Table~\ref{tab:certificates}).  The aggregation step in AEMV's proof
of Lemma~18 asserts $\eqref{eq:thirtytwo}\le\mathrm{FACTS}$; here,
exactly, phase 1 contributes its unit dwell total minus the drop penalty
$1-1/n$, while phase 2 contributes its unit dwell minus the correction
$1/2$.  Therefore
\begin{align*}
\mathrm{FACTS}
  &=\underbrace{1-(1-1/n)}_{\text{phase 1}}
    +\underbrace{1-1/2}_{\text{phase 2}}
    =\frac1n+\frac12,\\
\eqref{eq:thirtytwo}&=1+\frac1{2n},\\
\eqref{eq:thirtytwo}-\mathrm{FACTS}
  &=\frac12-\frac1{2n}\longrightarrow\frac12
    \qquad(n\to\infty),
\end{align*}
with no $\alpha$, $\beta$, or rise term anywhere; the deficit is
already $\tfrac12-\tfrac1{2n}$ at each finite $n$.
Thus the unrestricted algebraic implication from (29)/(30) to
(31)/(32) has a loss approaching $\tfrac12$.
(On this family the amortized potential account of
Section~\ref{sec:therepair} is tight: its potential $\Lambda$, defined in
Lemma~\ref{lem:amortizedagg}, starts at
$\Lambda(n,n)=\tfrac12$ and ends at $\Lambda(1,n)=\tfrac1{2n}$, and its
total drop $\tfrac12-\tfrac1{2n}$ exactly matches the deficit,
incurred entirely at the final correction step.)

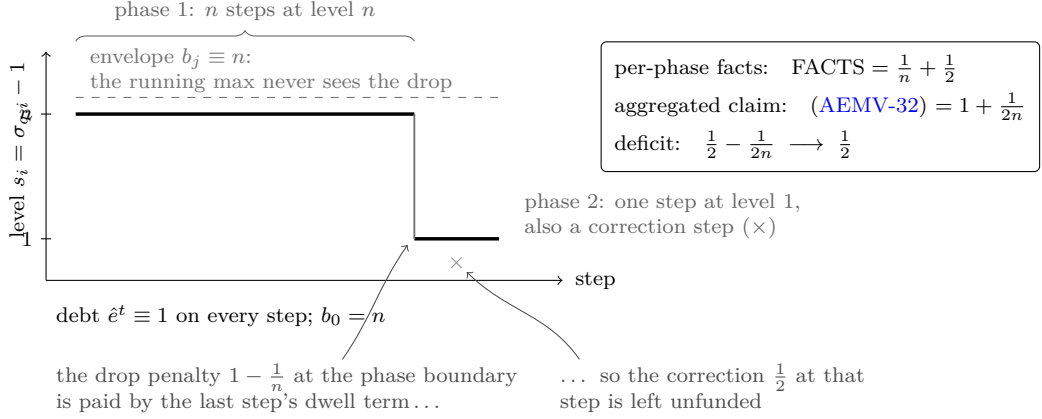
\begin{figure}[!t]
\centering
\begin{tikzpicture}[x=1.12cm,y=0.55cm]
\draw[->] (-0.35,0) -- (5.75,0) node[right=1pt,font=\scriptsize]{step};
\draw[->] (-0.35,0) -- (-0.35,5.5) node[rotate=90,anchor=south east,font=\scriptsize,yshift=2pt]{level $s_i=\sigma_{q,i}-1$};
\draw (-0.41,4) -- (-0.29,4) node[left=3pt,font=\scriptsize]{$n$};
\draw (-0.41,1) -- (-0.29,1) node[left=3pt,font=\scriptsize]{$1$};
\draw[line width=1.3pt] (0,4) -- (4,4);
\draw[line width=1.3pt] (4,1) -- (5,1);
\draw[black!55,line width=0.7pt] (4,4) -- (4,1);
\draw[dashed,black!55] (0,4.4) -- (5,4.4);
\node[font=\scriptsize,black!55,anchor=west,align=left] at (0.05,5.05)
  {envelope $b_j\equiv n$:\\ the running max never sees the drop};
\draw[decorate,decoration={brace,amplitude=4pt},black!60]
  (0,5.75) -- (4,5.75) node[midway,above=4pt,font=\scriptsize,black!60]
  {phase 1: $n$ steps at level $n$};
\node[font=\scriptsize,black!60,align=left,anchor=west] at (5.18,1.55)
  {phase 2: one step at level $1$,\\ also a correction step ($\times$)};
\node[font=\scriptsize,black!45] at (4.5,0.42) {$\times$};
\node[font=\scriptsize,anchor=west] at (-0.35,-0.85) {debt $\hat e^t\equiv1$ on every step; $b_0=n$};
\node[font=\scriptsize,align=left,anchor=north west,black!70] at (-0.35,-1.75)
  {the drop penalty $1-\tfrac1n$ at the phase boundary\\ is paid by the last step's dwell term\,\dots};
\draw[->,black!70] (3.15,-1.9) to[out=55,in=-125] (3.92,0.85);
\node[font=\scriptsize,align=left,anchor=north west,black!70] at (5.6,-1.75)
  {\dots{} so the correction $\tfrac12$ at that\\ step is left unfunded};
\draw[->,black!70] (5.95,-1.9) to[out=125,in=-45] (4.62,0.2);
\node[draw,rounded corners=2pt,font=\scriptsize,align=left,anchor=north west,inner sep=5pt] at (6.2,5.75)
  {per-phase facts:\quad $\mathrm{FACTS}=\tfrac1n+\tfrac12$\\[3pt]
   aggregated claim:\quad \eqref{eq:thirtytwo} $=1+\tfrac1{2n}$\\[3pt]
   deficit:\quad $\tfrac12-\tfrac1{2n}\;\longrightarrow\;\tfrac12$};
\end{tikzpicture}
\caption{The countermodel to the aggregation step, drawn for $n=4$:
two drop phases under a flat envelope, with the drop penalty and the
correction drawing on the same dwell term.}
\label{fig:countermodel}
\end{figure}
\FloatBarrier

\paragraph{A realized instance.} The pattern above is also realizable
as an actual run, as promised: an explicit nine-queue, $B=8$,
five-step instance realizes it at $n=2$, with LQD run under a fixed
tie rule against a brute-forced, exactly-optimal offline policy.  On
this run $\mathrm{FACTS}=1$ and $\eqref{eq:thirtytwo}=\tfrac54$,
computed in exact rational arithmetic (Figure~\ref{fig:realizedcex};
certificate
\certpath{certificates/erratum/realized_counterexample.py},
Table~\ref{tab:certificates}).  In the run, LQD transmits $19$
packets and OPT $20$; the live-set averages $\sigma_{q_1,1}=3$ and
$\sigma_{q_1,2}=2$ give $b_0=2$ and two drop phases; OPT's single
extra transmission from $q_1$ at step~$3$, while LQD's copy of $q_1$
is already empty, is the window's only correction step; and the
window closes with $\sigma_{q_1,3}=\tfrac43<2$, as anticipated
above.  For complete reproducibility, the
nonzero arrival batches are
$t=1:q_1^3,q_2^3,q_3^3$;
$t=3:q_2^2,q_4^2,q_5^3,q_6$; and
$t=4:q_4^2,q_7^3,q_8,q_9$.  When longest queues tie, LQD evicts from
the first queue in the fixed priority
$q_1,q_5,q_4,q_2,q_7,q_3,q_6,q_8,q_9$.
Thus the inequality asserted inside
AEMV's proof fails on a real run as well; one fixed tie rule suffices,
since the published aggregation is asserted for every tie rule.
Consequently the published journal derivation of $\CR\le1.6918$ is
invalid as written; the conference version \citep{AEMV21} states the
analogous aggregation lemma with the proof omitted, so its $\le1.707$
inherits the same missing justification.  What fails is the
algebraic implication from the per-phase facts (29)/(30) to the
aggregated display, not necessarily the statement of Lemma~18 itself,
since $\mathrm{FACTS}$ only lower-bounds the true L-increase.

%
\newcommand{\rcstack}[6]{%
  \ifnum#3>0
    \foreach \l in {1,...,#3}
      \draw[rcS] ({#1}, {#2+(\l-1)*0.20}) rectangle ({#1+0.36}, {#2+\l*0.20});
  \fi
  \ifnum#4>0
    \foreach \l in {1,...,#4}
      \draw[rcA] ({#1}, {#2+(#3+\l-1)*0.20}) rectangle ({#1+0.36}, {#2+(#3+\l)*0.20});
  \fi
  \ifnum#5>0
    \foreach \l in {1,...,#5}{
      \draw[rcE] ({#1}, {#2+(#3+#4+\l-1)*0.20+0.06}) rectangle ({#1+0.36}, {#2+(#3+#4+\l)*0.20+0.06});
      \draw[red!70!black,line width=0.7pt]
        ({#1+0.07},{#2+(#3+#4+\l-1)*0.20+0.10}) -- ({#1+0.29},{#2+(#3+#4+\l)*0.20+0.02})
        ({#1+0.07},{#2+(#3+#4+\l)*0.20+0.02})  -- ({#1+0.29},{#2+(#3+#4+\l-1)*0.20+0.10});
    }
  \fi
  \ifnum#6>0
    \ifnum#3>0
      \draw[rcT] ({#1}, {#2}) rectangle ({#1+0.36}, {#2+0.20});
    \else
      \draw[rcTA] ({#1}, {#2}) rectangle ({#1+0.36}, {#2+0.20});
    \fi
  \fi
}
\begin{figure}[!t]
\centering
\begin{tikzpicture}[x=0.98cm,y=1cm,
  rcS/.style={fill=black!25,draw=black!60,line width=0.4pt},
  rcA/.style={fill=white,draw=blue!60!black,line width=0.8pt},
  rcE/.style={draw=black!50,dashed,line width=0.4pt},
  rcT/.style={fill=black!45,draw=black!60,line width=0.4pt},
  rcTA/.style={fill=black!45,draw=blue!60!black,line width=0.8pt},
  ttl/.style={font=\scriptsize},
  lbl/.style={font=\scriptsize},
  liv/.style={font=\tiny,circle,draw,inner sep=0.7pt},
  tag/.style={font=\tiny},
  qn/.style={font=\tiny}]

\node[lbl,anchor=east] at (-0.35,0.35) {LQD};
\node[lbl,anchor=east] at (-0.35,-1.40) {OPT};

\foreach \x in {-0.30, 3.85, 7.15}
  \draw[black!35,dashed,line width=0.5pt] (\x,0.80) -- (\x,-2.70);
\node[tag,black!60] at (1.75,0.72) {phase 1};
\node[tag,black!60] at (5.50,0.72) {phase 2};
\node[tag,black!60] at (9.60,0.72) {phase 3 (post-window)};

\draw[line width=0.7pt] (-0.08,1.00) -- (6.75,1.00);
\draw[line width=0.7pt] (-0.08,1.00) -- (-0.08,0.90)  (6.75,1.00) -- (6.75,0.90);
\node[tag] at (3.10,1.14) {pending window of $q_1$ (debts $\hat e^t=1,1,1$)};

\node[ttl] at (0.60,1.45) {$t{=}1$ ($\tau_1{=}t_q$)};
\rcstack{0.00}{0}{0}{2}{1}{1}
\rcstack{0.42}{0}{0}{3}{0}{1}
\rcstack{0.84}{0}{0}{3}{0}{1}
\node[qn] at (0.18,-0.32) {$q_1$};
\node[liv] at (0.60,-0.32) {$q_2$};
\node[liv] at (1.02,-0.32) {$q_3$};
\node[tag] at (0.18,-0.62) {$\hat e_{q_1}{=}1$};
\rcstack{0.00}{-1.75}{0}{3}{0}{1}
\rcstack{0.42}{-1.75}{0}{2}{1}{1}
\rcstack{0.84}{-1.75}{0}{3}{0}{1}
\node[lbl,anchor=west] at (-0.10,-2.45)
  {$\sigma_{q_1,1}{=}\tfrac{3+3}{2}{=}3$,\ $s_1{=}2$,\ $b_0{=}2$};

\node[ttl] at (2.62,1.45) {$t{=}2$};
\rcstack{2.20}{0}{1}{0}{0}{1}
\rcstack{2.62}{0}{2}{0}{0}{1}
\rcstack{3.04}{0}{2}{0}{0}{1}
\node[qn] at (2.38,-0.32) {$q_1$};
\node[qn] at (2.80,-0.32) {$q_2$};
\node[qn] at (3.22,-0.32) {$q_3$};
\rcstack{2.20}{-1.75}{2}{0}{0}{1}
\rcstack{2.62}{-1.75}{1}{0}{0}{1}
\rcstack{3.04}{-1.75}{2}{0}{0}{1}

\node[ttl] at (5.30,1.45) {$t{=}3$ ($\tau_2$)};
\draw[black!40,dotted,line width=0.5pt] (4.26,0.00) rectangle (4.62,0.20);
\rcstack{4.68}{0}{1}{1}{1}{1}
\rcstack{5.10}{0}{1}{0}{0}{1}
\rcstack{5.52}{0}{0}{2}{0}{1}
\rcstack{5.94}{0}{0}{2}{1}{1}
\rcstack{6.36}{0}{0}{1}{0}{1}
\node[qn] at (4.44,-0.32) {$q_1$};
\node[liv] at (4.86,-0.32) {$q_2$};
\node[qn] at (5.28,-0.32) {$q_3$};
\node[liv] at (5.70,-0.32) {$q_4$};
\node[liv] at (6.12,-0.32) {$q_5$};
\node[qn] at (6.54,-0.32) {$q_6$};
\node[tag] at (4.44,-0.62) {corr.\ step};
\rcstack{4.26}{-1.75}{1}{0}{0}{1}
\node[tag] at (4.44,-1.97) {OPT-extra};
\rcstack{4.68}{-1.75}{0}{2}{0}{1}
\rcstack{5.10}{-1.75}{1}{0}{0}{1}
\rcstack{5.52}{-1.75}{0}{1}{1}{1}
\rcstack{5.94}{-1.75}{0}{2}{1}{1}
\rcstack{6.36}{-1.75}{0}{1}{0}{1}
\node[lbl,anchor=west] at (4.00,-2.45)
  {$\sigma_{q_1,2}{=}\tfrac{2+2+2}{3}{=}2$,\ $s_2{=}1$};

\node[ttl] at (8.60,1.45) {$t{=}4$ ($\tau_3{=}\tau_{j_q}$)};
\rcstack{7.55}{0}{1}{0}{0}{1}
\rcstack{7.97}{0}{1}{1}{1}{1}
\rcstack{8.39}{0}{1}{0}{0}{1}
\rcstack{8.81}{0}{0}{2}{1}{1}
\rcstack{9.23}{0}{0}{1}{0}{1}
\rcstack{9.65}{0}{0}{1}{0}{1}
\node[liv] at (7.73,-0.32) {$q_2$};
\node[liv] at (8.15,-0.32) {$q_4$};
\node[liv] at (8.57,-0.32) {$q_5$};
\node[liv] at (8.99,-0.32) {$q_7$};
\node[liv] at (9.41,-0.32) {$q_8$};
\node[liv] at (9.83,-0.32) {$q_9$};
\rcstack{7.55}{-1.75}{1}{0}{0}{1}
\rcstack{7.97}{-1.75}{0}{2}{0}{1}
\rcstack{8.39}{-1.75}{1}{0}{0}{1}
\rcstack{8.81}{-1.75}{0}{2}{1}{1}
\rcstack{9.23}{-1.75}{0}{1}{0}{1}
\rcstack{9.65}{-1.75}{0}{1}{0}{1}
\node[lbl] at (8.60,-2.45) {$\sigma_{q_1,3}{=}\tfrac43$ (all six live)};

\node[ttl] at (11.05,1.45) {$t{=}5$};
\rcstack{10.84}{0}{1}{0}{0}{1}
\rcstack{11.26}{0}{1}{0}{0}{1}
\node[qn] at (11.02,-0.32) {$q_4$};
\node[qn] at (11.44,-0.32) {$q_7$};
\rcstack{10.84}{-1.75}{1}{0}{0}{1}
\rcstack{11.26}{-1.75}{1}{0}{0}{1}

\node[lbl,align=center] at (5.55,-3.35)
  {$\mathrm{FACTS}
   =\bigl[\sum_{t=1,2}\tfrac{\hat e^t}{s_1}
     -\hat e_{q_1,2}(\tfrac1{s_2}-\tfrac1{s_1})\bigr]
   +\bigl[\tfrac{\hat e^3}{s_2}-\tfrac{g_{q_1,2}}{2s_2}\bigr]
   =\tfrac12+\tfrac12=1$\\[2pt]
   $1<\tfrac54=\eqref{eq:thirtytwo}$ at the flat envelope
   $b_j\equiv2$; here $g_{q_1,1}=0$ and $g_{q_1,2}=1$.};
\node[tag] at (5.55,-4.20)
  {AEMV's proof asserts $\eqref{eq:thirtytwo}\le\mathrm{FACTS}$; the
   run refutes this internal inequality, not LQD's competitive ratio.};
\end{tikzpicture}
\caption{The realized counterexample ($n=2$), full run: LQD (top) and
the optimal offline state (bottom) at each step, drawn as in
Figure~\ref{fig:modelstep} after arrivals and evictions.  Gray cells:
stored packets; blue frames: admitted this step; crossed cells:
evicted or declined; dotted outline: empty queue; the dark bottom
cell departs in transmission.  Circled names are live w.r.t.\ $q_1$,
and dashed verticals mark the phase borders.}
\label{fig:realizedcex}
\end{figure}
\FloatBarrier

\paragraph{The gap grows linearly with the instance.} The failure cannot hide in
the fixed $O(B)$ additive slack that the competitive ratio tolerates
(Remark~\ref{rem:cradditive}).  Concatenate in time $r$ drained copies
of the realized $n=2$ instance of Figure~\ref{fig:realizedcex} under
fresh queue names (each copy runs to completion before the next
begins), keeping the buffer fixed at $B=8$ as $r$ grows.  Each
copy starts from an empty switch and replays identically, so each
copy's pending-window queue ($q_1$ in the figure) contributes the
same deficit of $\tfrac14$ between $\mathrm{FACTS}$ and
\eqref{eq:thirtytwo}.  (OPT also decomposes per copy: colour the
packets by copy; restricting any offline schedule to one colour only
removes buffer load, giving a feasible schedule of the single-copy
instance, so at most $20$ packets of each colour are transmitted:
$\OPT\le20r$, attained by replaying the single-copy optimum copy by
copy.)  The deficits add; the certificate
\certpath{certificates/erratum/realized_counterexample.py} re-runs its
entire certification at $r=2$ to confirm this.  The published chain's
claimed inequality therefore overshoots what its premises justify by
$\tfrac r4$, which grows linearly with $r$, so no $O(B)$ budget
absorbs it.  Unlike Remark~\ref{rem:cradditive}'s concatenation,
which preserves the \emph{true} ratio, this one grows only the
proof's bookkeeping gap.  Thus the gap cannot be dismissed as an
$O(B)$ term; our repair absorbs each window's deficit (at most
$\tfrac12$, Lemma~\ref{lem:amortizedagg}) locally, within that window's own
payment.  We repair the step before using it.

\subsection{The repair}\label{sec:therepair}

Charging the $\tfrac12$ against the rate is not an option: as noted at
the start of Section~\ref{sec:erratum}, it would cost a full $\tfrac12$ of $\varrho$
and overshoot $2$.  The reason is dimensional: $\varrho$ is a rate
(payment per unit debt), while the deficit is an absolute constant.
The corrected step pays each queue at least $\varrho\,\hat e_q-\tfrac12$,
so a queue with the smallest possible credit $\hat e_q=1$ alone forces
any uniformly salvageable rate down to $\varrho-\tfrac12$ (via the
standing $\CR\le1+1/\varrho$).  With AEMV's own rate
$\rhostar\approx1.4455$, this would salvage only the guarantee
$\CR\le1+1/(\rhostar-\tfrac12)\approx2.058$.  The $\tfrac12$ must instead be charged
\emph{additively}, out of the payment slack each queue already carries
above $\varrho\,\hat e_q$.

In Section~\ref{sec:thegap}, we have shown where the aggregation
double-spends; now we make the double spend explicit and pay for it.
Recall the mechanism: passing to the running-max envelope funds each
phase's drop penalty out of the \emph{next} phase's first dwell term
(the drop-penalty pairing), leaving that dwell already spent; when the
next phase is itself a drop whose own first step carries a correction,
the same term must cover both the drop penalty and the correction, and
it is the correction that goes uncovered.  One amortized lemma pays
every such double-spent step by telescoping a potential together with
the unused surplus that rise phases leave behind, and the resulting
$\tfrac12$ loss is absorbed later, by the exact finite-head reserve of
Lemma~\ref{lem:reserve}.

\begin{lemma}[amortized corrected aggregation]\label{lem:amortizedagg}
Let $\alpha\ge\tfrac12$.  Fix one pending window, relabel its phases
$0,1,\ldots$, put $s_i=\sigma_{q,i}-1$, and let
$a_0=b_0\ge s_0$ and $a_i=\max\{s_i,a_{i-1}\}$.  Write
$\mathrm{FACTS}$ for the sum of the applicable per-phase bounds
\eqref{eq:aemv29}--\eqref{eq:aemv30}, and let $\mathrm{ENV}$ be
\eqref{eq:envphase} (equivalently, \eqref{eq:aemv31}).  With
\[
\Lambda(s,a):=\frac12+\frac1{2a}-\frac1{2s},
\]
one has
\begin{equation}\label{eq:amortizedagg}
\mathrm{FACTS}\ge \mathrm{ENV}-\Lambda(s_0,b_0)
\ge \mathrm{ENV}-\frac12.
\end{equation}
The threshold $\alpha\ge\tfrac12$ is sharp for this potential
argument: the rise inequality used below fails for every
$\alpha<\tfrac12$.  The abstract family of Section~\ref{sec:thegap}
asymptotically attains the loss $\tfrac12$.
\end{lemma}

\begin{proof}
The proof has three parts.  First we match each claim in
$\mathrm{ENV}$ against the per-phase provision that covers it, and
isolate the only steps at which the coverage can fail; the possible
shortfall there is quantified by the charges $c_{i+1}$ of
\eqref{eq:badcharge} below.  Then we prove a transition inequality,
\eqref{eq:transitionamort}, trading these charges against the unused
rise surpluses \eqref{eq:risesurplus} and the change of the potential
$\Lambda$ across the transition.  Finally we telescope
\eqref{eq:transitionamort} over the window: the total shortfall is
bounded by the total fall of $\Lambda$, which never exceeds its
initial value $\Lambda(s_0,b_0)\le\tfrac12$.

The expression $\mathrm{ENV}$ contains a dwell claim $\hat e^t/a_i$
at each pending step, a rise claim
$\alpha\hat e^{\tau_{i+1}-1}(a_{i+1}-a_i)/a_i$ at each transition,
and a correction claim $-1/(2a_i)$ at each correction step.
The per-phase facts provide the corresponding dwell terms with
denominator $s_i$, the rise bonuses with the raw levels, the drop
penalties from \eqref{eq:aemv30}, and corrections $-1/(2s_i)$ in drop
phases.  In a rise phase, \eqref{eq:aemv29} is the raw-level
form retained here: it supplies these dwell terms and the
raw-level bonus.  The preceding stronger Lemma~15 display dominates
\eqref{eq:aemv29} because the debt is nonincreasing and the maximum
contains $s_{i+1}-s_i$.

Pair each drop penalty with the first dwell term of its target phase.
Every $\mathrm{ENV}$ term is then covered, except possibly when that
target step is a correction and the target phase itself drops.  At such
a target the predecessor drop, the target dwell, and the target
correction leave the shortfall
\begin{equation}\label{eq:badcharge}
c_{i+1}:=\left[
 \frac1{2s_{i+1}}-\frac1{2a_{i+1}}
 -\hat e_{q,i+1}\left(\frac1{s_i}-\frac1{a_{i+1}}\right)
 \right]^+;
\end{equation}
set $c_{i+1}=0$ at every other target.  Indeed, an unpaired dwell
covers its envelope dwell because $s_i\le a_i$; at an ordinary
correction its dwell surplus covers the denominator change because
$\hat e^t\ge1$; and a correction in a non-drop phase only lowers
$\mathrm{ENV}$, since \eqref{eq:aemv29} carries no correction charge.
These cases are disjoint and exhaust the pending steps.

At a rise transition $i\to i+1$, the raw bonus leaves the unused
surplus
\begin{equation}\label{eq:risesurplus}
r_i:=\alpha\hat e^{\tau_{i+1}-1}_q
\left(\frac{s_{i+1}-s_i}{s_i}
      -\frac{a_{i+1}-a_i}{a_i}\right)\ge0;
\end{equation}
set $r_i=0$ otherwise.  The bracket is nonnegative by the definition
of the running maximum.  If the target state is still pending, then
$\hat e_q^{\tau_{i+1}-1}\ge1$.  If the target boundary has zero debt,
there is no target charge $c_{i+1}$; any terminal rise surplus
$r_i\ge0$ can be discarded.  Thus, with both sums restricted to
transitions whose targets are still pending, the preceding coverage gives
\begin{equation}\label{eq:cover}
\mathrm{FACTS}\ge\mathrm{ENV}-\sum_i c_{i+1}+\sum_i r_i.
\end{equation}

It remains to amortize the right-hand side.  Set
$z_i:=1/s_i-1/a_i$.  Since $1\le s_i\le a_i$ throughout a pending
window, $0\le z_i\le1$ and
$\Lambda(s_i,a_i)=(1-z_i)/2$.  We claim that every transition between
pending states satisfies
\begin{equation}\label{eq:transitionamort}
r_i-c_{i+1}\ge\frac{z_i-z_{i+1}}2.
\end{equation}

For a drop, $a_{i+1}=a_i$, $z_{i+1}\ge z_i$, and $r_i=0$.  If its
target is chargeable, \eqref{eq:badcharge} becomes
\[
c_{i+1}=\left[\frac{z_{i+1}}2-\hat e_{q,i+1}z_i\right]^+
\le\left[\frac{z_{i+1}}2-z_i\right]^+
\le\frac{z_{i+1}-z_i}2;
\]
otherwise the same final bound is immediate from $c_{i+1}=0$.  This is
\eqref{eq:transitionamort}.

For a rise the target cannot be chargeable, so $c_{i+1}=0$.  If
$s_{i+1}\le a_i$, then $a_{i+1}=a_i$ and
\[
r_i-\frac{z_i-z_{i+1}}2
\ge
\frac{(s_{i+1}-s_i)(2\alpha s_{i+1}-1)}
     {2s_i s_{i+1}}\ge0.
\]
If instead $s_{i+1}>a_i$, then $a_{i+1}=s_{i+1}$, $z_{i+1}=0$, and
\[
r_i-\frac{z_i-z_{i+1}}2
\ge
\frac{(a_i-s_i)(2\alpha s_{i+1}-1)}{2s_i a_i}\ge0.
\]
Both inequalities use $\alpha\ge\tfrac12$ and $s_{i+1}\ge1$.
For the asserted sharpness, fix $\alpha<\tfrac12$, take last-step debt
one, $s_i=1$, choose
$1<s_{i+1}<1/(2\alpha)$, and take $a_i\ge s_{i+1}$.  The first case is
then an equality, and
\[
r_i-\frac{z_i-z_{i+1}}2
=\frac{(s_{i+1}-1)(2\alpha s_{i+1}-1)}{2s_{i+1}}<0.
\]
Thus this transition inequality really fails below $\alpha=1/2$.
A flat transition changes neither $z$ nor the charge and is immediate.

Finally sum \eqref{eq:transitionamort} over transitions between pending
states.  If $f$ is the last pending state, then
\[
\sum_i(c_{i+1}-r_i)
\le\frac{z_f-z_0}{2}
\le\frac{1-z_0}{2}=\Lambda(s_0,b_0)\le\frac12.
\]
Combining this telescope with \eqref{eq:cover} proves
\eqref{eq:amortizedagg}.
\end{proof}

\begin{corollary}[repaired aggregation]\label{cor:repaired}
Let $\alpha\ge\tfrac12$ and $\beta\ge1-\sqrt{1-\alpha}-\alpha/2$
(Lemma~16's condition).  The standing hypotheses $\alpha\le2/3$ and
$\alpha+\beta<1$ of Section~\ref{sec:aemvdefs} are also in force; the
latter ensures the S-increase $(1-\alpha-\beta)\,b_0$ is nonnegative.
For every queue $q$ with $\hat e_q>0$ and
$b_0=\lceil\sigma_{q,i_q}\rceil-1$, the splitting scheme assigns $q$ a
total payment of at least $(1-\alpha-\beta)\,b_0$ plus
\eqref{eq:thirtytwo} evaluated at the running-max envelope, minus
$\tfrac12$.  Since that envelope is only known to be a real-valued
non-decreasing sequence with initial value $b_0$, this bound may be
relaxed to the \emph{infimum} of \eqref{eq:thirtytwo} over all such
sequences: $q$'s payment is at least $(1-\alpha-\beta)\,b_0$ plus that
infimum, minus $\tfrac12$.
\end{corollary}

\begin{proof}
Lemma~\ref{lem:amortizedagg} gives
$\mathrm{FACTS}\ge\mathrm{ENV}-\tfrac12$; $\mathrm{ENV}$ dominates
\eqref{eq:thirtytwo} at the running-max envelope (this passage,
from their step-sum (31) to their (32), is sound as published), and
the S-increase bound supplies the $(1-\alpha-\beta)\,b_0$ term.
The queue's payment is its L-increase part (at least
$\mathrm{FACTS}$, the summed per-phase bounds) plus its S-increase part
(at least $(1-\alpha-\beta)\,b_0\ge0$), so
$\mathrm{payment}_q\ge\mathrm{FACTS}+(1-\alpha-\beta)b_0\ge
\eqref{eq:thirtytwo}-\tfrac12+(1-\alpha-\beta)b_0$, which is the claim.
\end{proof}

Corollary~\ref{cor:repaired} restores the aggregation inequality, at the
cost of an explicit $\tfrac12$-per-queue charge; Section~\ref{sec:upper}
absorbs that charge and improves the constant by replacing AEMV's
per-packet endgame relaxation with an exactly solved continuum
optimization.

\section{The improved upper bound}\label{sec:upper}

With the aggregation repaired (Corollary~\ref{cor:repaired}), in this
section we sharpen the final optimization, tightening AEMV's
published $\CR(\LQD)\le1.6917948$ to $\mathbf{1.683652}$
(Theorem~\ref{thm:improved}).  We import the repaired aggregation,
solve its envelope relaxation with the mandatory first dwell retained,
and use the resulting finite-head reserve to absorb the residual
$\tfrac12$ per queue.  Our analysis covers the whole integer parameter
domain directly; only the two shortest windows require elementary
unaggregated calculations.

In more detail, we prove Theorem~\ref{thm:improved} in
five steps.  Step~1 combines the repaired aggregation with the full
S-payment into one master inequality per queue
(Lemma~\ref{lem:masterpayment}): a queue with head level $b$, debt
$e$, and block length $L$ is paid at least
$hL+E^{\inf}_{b,e}-\tfrac12$, where $E^{\inf}_{b,e}$ is the infimum of
the aggregated display over all envelopes.  Step~2 bounds
$E^{\inf}_{b,e}$ from below by a continuum optimization that keeps the
discrete head step intact, and solves that optimization exactly
(Lemma~\ref{lem:headclosed}): the optimal envelope jumps once, right
after the head, and then stays constant, which gives a closed form
$K_b(x)$ in the single variable $x=e/b$ (Figure~\ref{fig:prefix}).
Step~3 passes to the $b\to\infty$ limit $J_{\rm cont}$
(Lemma~\ref{lem:jclosed}), whose minimum over $x$ is the paper's rate:
it exceeds $\hat\varrho=1.4627334$, and $1+1/\hat\varrho\le1.683652$.
Step~4 returns to finite $b$ and shows that at every integer pair
$(b,e)$ except $(1,1)$ and $(2,1)$ the finite-head bound retains a
reserve strictly above the repair's $\tfrac12$
(Lemma~\ref{lem:reserve}; Figure~\ref{fig:casemap} maps the case
analysis).  Step~5 settles the two exceptional pairs: a window with a
spare transmission ($L>b$) recovers through one extra S-payment, and
the two shortest windows ($L=b$) are checked directly from the
unaggregated per-phase facts (Lemma~\ref{lem:basewindows}).

\subsection{The endgame of \texorpdfstring{\cite{AEMV24}}{[AEMV24]} (Claim 19, Lemma 21) and its constants}

Recall the aggregated display~\eqref{eq:thirtytwo}, AEMV's~(32)
restated in Section~\ref{sec:aemvdefs} and cited by this one number
throughout.  AEMV's Lemma~18
bounds the L-increase part (Section~\ref{sec:aemvdefs}) of each queue's
payment from below by a sum
over the running-max envelope $b_0\le b_1\le\cdots$.
Section~\ref{sec:therepair} showed this bound holds only up to a
$\tfrac12$-per-queue slack, which the finite-head reserve below
absorbs.  AEMV's endgame then bounds this sum from below \emph{per
packet}, which is not tight; we instead lower-bound it by its
\emph{exact} continuum infimum over envelopes.

Claim~19 from \cite{AEMV24} lower-bounds each summand of
\eqref{eq:thirtytwo}: it prices the $f$-th OPT-extra packet, of prefix
length $m_f=b_0+f$, at the prefix's endpoint level $b_{m_f}$, and
minimizes over that endpoint per packet (inside its own proof) at
$b_{m_f}=(m_f-\tfrac32)/\alpha$; the paragraph closing this subsection
examines where this loses.  Their Lemma~21 then optimizes, over real
$x>0$, the resulting closed form
\[
g_\alpha(x)\;=\;\frac{1-\alpha-\beta}{x}
+\alpha\Bigl(1+\frac1x\Bigr)\ln(1+x)+\alpha\ln\frac1\alpha,
\qquad x=\hat e_q/b_0,
\]
(eq.~(41) in \cite{AEMV24}).  Here $\beta$ is fixed to its Lemma-16
minimum $1-\sqrt{1-\alpha}-\alpha/2$, so it too is a function of
$\alpha$ alone: that is why $g_\alpha$ carries only the subscript
$\alpha$.  For each scheme parameter
$\alpha$ this inner infimum is the extracted payment rate
$\rhostar(\alpha)=\inf_x g_\alpha(x)$; choosing $\alpha$ to
maximize it gives $\alpha\approx0.618906$ and
$\max_\alpha\rhostar(\alpha)\ge1.4455154$, hence, via
$\CR\le1+1/\varrho$, their journal bound $\CR(\LQD)\le1.6917948$.
(The earlier conference version \citep{AEMV21} gave $1.707$.)  Table~\ref{tab:constants}
records both parameter choices to full precision: theirs,
recomputed here as a 60-digit check on $g_\alpha$; and ours, from
Theorem~\ref{thm:improved}.  In both rows $\beta$ is the smallest
value permitted by AEMV's Lemma~16.  Their rate is
$\max_\alpha\rhostar(\alpha)$; ours is a certified lower bound on a
different, larger functional of the same payment expression, the
infimum of the payment rate over envelopes, which
Theorem~\ref{thm:improved} solves exactly.  The improvement comes
from that exact solution, not from the lower $\alpha$ itself: the
exact solution shifts the optimal $\alpha$ downward, and $613/1000$
is a nearby rational chosen so that the constants certify exactly.

\begin{table}[!t]
\centering\small
\begin{tabular}{lllll}
\toprule
& $\alpha$ & $\beta=1-\sqrt{1-\alpha}-\alpha/2$ & rate & $\CR(\LQD)\le$ \\
\midrule
AEMV, optimized & $0.61890600\ldots$ & $0.07321888\ldots$
& $1.44551540\ldots$ & $1.69179477\ldots$ \\
This paper, Theorem~\ref{thm:improved} & $613/1000$ (exact)
& $0.07140675\ldots$ & $\ge1.4627334$ & $1.683652$ \\
\bottomrule
\end{tabular}
\caption{The two parameter choices and their rates.}
\label{tab:constants}
\end{table}

Theorem~\ref{thm:improved} keeps everything up to and including the
repaired aggregation (Corollary~\ref{cor:repaired}), discards Claim~19,
and replaces it by the continuum envelope relaxation, solved exactly in
closed form, with the finite-$b_0$ refinements certified
separately.

\paragraph{Where AEMV's endgame is not tight.}
Recall Claim~19's pricing from the start of this subsection: the
$f$-th OPT-extra packet is priced at its prefix's \emph{endpoint}
level $b_{m_f}$ (specifically, every dwell term of the prefix after
the first is priced there), and the endpoint is then minimized per
packet, at $b_{m_f}=(m_f-\tfrac32)/\alpha$.  Each per-packet bound is tight only for
envelopes that sit at that packet's private endpoint level from the
first step on, and no single non-decreasing envelope does so for every
packet simultaneously, so slack is unavoidable.  The
remark after (32) (Section~\ref{sec:aemvdefs}) licenses optimizing
their payment expression over envelopes; the closed form below shows
quantitatively that this infimum \emph{strictly} exceeds the constant
their chain extracts: at AEMV's own
$\alpha^*$ the envelope-optimized rate is $\approx1.462555$, exceeding their
$\rhostar\approx1.44552$ by $\approx0.017$.  The
slack is recovered by optimizing the same payment expression
directly, with no new structural lemma about LQD.

\subsection{The theorem and its proof}

\begin{theorem}[improved upper bound]\label{thm:improved}
For every instance and every tie rule,
$\OPT\le(1+1/\hat\varrho)\LQD+(\alpha/\hat\varrho)B$, where
$\hat\varrho:=1.4627334$ and $\alpha=613/1000$.  Consequently
\[
\CR(\LQD)\le1+1/\hat\varrho
=1.683651579980\ldots<\mathbf{1.683652}.
\]
The remaining parameter is
$\beta=1-\sqrt{1-\alpha}-\alpha/2$.  The analytic chain is
Lemma~\ref{lem:windowstructure}, Lemma~\ref{lem:amortizedagg},
Corollary~\ref{cor:repaired}, and Lemmas~\ref{lem:masterpayment},
\ref{lem:headclosed}, \ref{lem:jclosed}, \ref{lem:reserve}, and
\ref{lem:basewindows};
its numerical leaves are certified by outward interval arithmetic.
\end{theorem}

Write $h:=1-\alpha-\beta$ for the S-increase that each LQD
transmission within the window earns, and, for a queue with debt
$e:=\hat e_q>0$, write
\[
b:=b_0=\lceil\sigma_{q,i_q}\rceil-1,
\qquad x:=e/b.
\]
Let $L$ be the number of steps of the pending window at which LQD
still stores $q$ (the \emph{constant-debt block}), so the window has
$L+e$ steps in total.  To pin down the boundary: since $\hat e^t$ counts
transmissions at step $t$ \emph{or later}, the debt holds at $e$ on
the $L$ block steps \emph{and} on the first step of the drained tail,
and only then falls by one per step, ending the window at value $1$.
Lemma~\ref{lem:windowstructure} gives $L\ge b$.

For integers $b,e\ge1$, let $m=b+e$ and define the
\emph{discrete envelope infimum}
\begin{equation}\label{eq:einf}
E^{\inf}_{b,e}:=
\inf_{b=b_0\le b_1\le\cdots\le b_m}
\left\{
\sum_{j=0}^{b-1}\frac{\alpha(b_{j+1}-b_j)+1}{b_j}\,e
+\sum_{j=b}^{m-1}\left(
 \frac{\alpha(b_{j+1}-b_j)+1}{b_j}(e+b-j)-\frac1{2b_j}
\right)\right\},
\end{equation}
where the infimum is over positive real envelope levels.  The endpoint
$b_m$ is included because the last rise term uses it.  The two sums
mirror the two blocks of the canonical prefix in
Figure~\ref{fig:window}: the first $b$ steps carry the full weight
$e$, and the decay steps carry the declining weights together with the
blanket corrections $-1/(2b_j)$.

Display~\eqref{eq:thirtytwo} prices the window's first $m=b+e$ steps:
the \emph{canonical prefix}, matching a window with $L=b$ exactly (its
weights put the value $e$ on $b+1$ steps, then decline to $1$).  For
this prefix, $b_m$ is the actual next envelope level: it is the level
at the next step when the cut is internal, and the boundary level just
after the last pending step when the cut is terminal.  For
completeness, here is why the imported step-sum passage covers every
$L\ge b$.  On the priced steps the canonical weights ($e$ for $j<b$,
then $e+b-j$) never exceed the true debt: the true debt holds at $e$
through all $L\ge b$ block steps, and on any priced tail step
$e-(j-L)\ge e-(j-b)$.  The display's blanket corrections $-1/(2b_j)$
on every $j\ge b$ only over-subtract, since true correction steps
occur only on the drained tail (LQD still stores $q$ on block steps).
And each unpriced step beyond the prefix contributes a nonnegative
net amount: its dwell $w/b_j$ with weight $w\ge1$ covers its (at most
one) correction $-1/(2b_j)$, and rise terms are nonnegative.  Hence
the full window's aggregated step-sum ($\mathrm{ENV}$) dominates \eqref{eq:thirtytwo} evaluated
at the envelope of its first $m$ steps.  That prefix is feasible in
\eqref{eq:einf}, for every $L\ge b$.

\begin{lemma}[master payment reduction]\label{lem:masterpayment}
For a queue with debt $e=\hat e_q>0$, head level $b=b_0$, and
constant-debt block length $L\ (\ge b)$, with $h=1-\alpha-\beta$ and
$E^{\inf}_{b,e}$ as in \eqref{eq:einf},
\[
\operatorname{payment}_q
\;\ge\;hL+E^{\inf}_{b,e}-\frac12.
\]
\end{lemma}

\begin{proof}
Corollary~\ref{cor:repaired} gives the L-increase envelope bound with
loss at most $1/2$.  Its displayed $hb$ S-increase is only the lower
bound obtained after retaining $b$ of the $L$ constant-debt steps.
In fact each of those $L$ steps is a step at which LQD still stores and
transmits from $q$ (every nonempty queue transmits once per step,
Definition~\ref{def:switch}); since $t_q$ is itself a phase boundary,
each such step lies in a phase with $\tau_i\ge t_q$ and
$\hat e_{q,i}\ge e>0$, so AEMV's S-increase pays $h$ per such transmission
[their eq.~(13)].  The S- and L-increase accounts are disjoint, so the
full S-payment is $hL$.  The canonical-prefix argument above and
the definition of $E^{\inf}_{b,e}$ lower-bound the repaired L-increase
by $E^{\inf}_{b,e}-1/2$.  Adding the two contributions proves the
claim.  In particular, the repair does not spend the restored
S-payment $h(L-b)$.
\end{proof}

The next lemma is the technical core of the improvement: it bounds the
discrete envelope infimum by a continuum optimization that retains the
discrete head step, and then solves that optimization exactly.
Figure~\ref{fig:prefix} shows the two objects side by side: the
discrete weights of the canonical prefix, and the continuum trapezoid
together with the optimal envelope, which jumps once and stays
constant.  Rescaling steps by $b$ turns the blanket corrections
$-1/(2b_j)$ into exact midpoint terms of the continuum dwell, so the
passage to the continuum loses nothing.

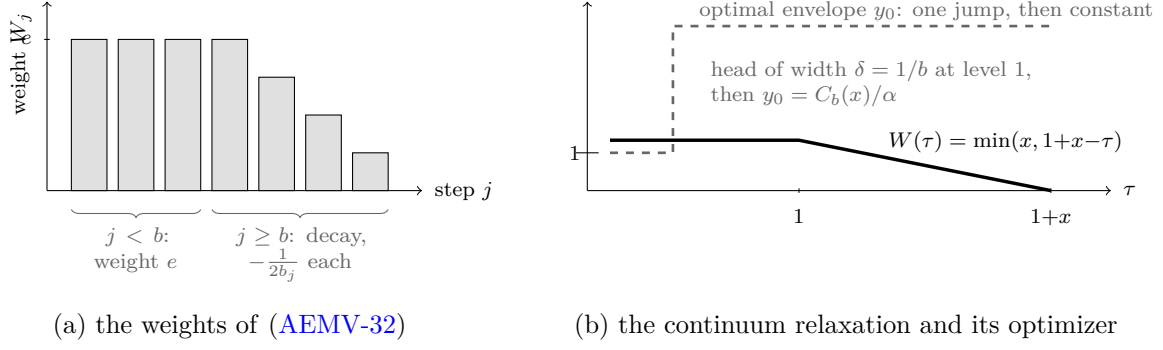
\begin{figure}[!t]
\centering
\begin{tikzpicture}
\begin{scope}[x=0.62cm,y=0.5cm]
\draw[->] (-0.4,0) -- (7.6,0) node[right=1pt,font=\scriptsize]{step $j$};
\draw[->] (-0.4,0) -- (-0.4,5.0) node[rotate=90,anchor=south east,font=\scriptsize,yshift=2pt]{weight $W_j$};
\draw (-0.46,4) -- (-0.34,4) node[left=3pt,font=\scriptsize]{$e$};
\foreach \j/\w in {0/4,1/4,2/4,3/4,4/3,5/2,6/1}{
  \draw[fill=black!12] (\j+0.12,0) rectangle (\j+0.88,\w);}
\draw[decorate,decoration={brace,mirror,amplitude=3pt},black!60]
  (0.12,-0.5) -- (2.88,-0.5) node[midway,below=3pt,font=\scriptsize,black!60,align=center,text width=1.5cm]{$j<b$:\\ weight $e$};
\draw[decorate,decoration={brace,mirror,amplitude=3pt},black!60]
  (3.12,-0.5) -- (6.88,-0.5) node[midway,below=3pt,font=\scriptsize,black!60,align=center,text width=1.9cm]{$j\ge b$: decay, $-\tfrac1{2b_j}$ each};
\node[font=\small] at (3.5,-3.6) {(a) the weights of \eqref{eq:thirtytwo}};
\end{scope}
\begin{scope}[xshift=7.2cm,x=2.5cm,y=0.5cm]
\draw[->] (-0.12,0) -- (2.65,0) node[right=1pt,font=\scriptsize]{$\tau$};
\draw[->] (-0.12,0) -- (-0.12,5.0);
\draw (1,-0.06) -- (1,0.06); \node[below=3pt,font=\scriptsize] at (1,0) {$1$};
\draw (2.333,-0.06) -- (2.333,0.06); \node[below=3pt,font=\scriptsize] at (2.333,0) {$1{+}x$};
\draw[line width=1.3pt] (0,1.333) -- (1,1.333) -- (2.333,0);
\node[font=\scriptsize,anchor=west] at (1.42,1.30) {$W(\tau)=\min(x,1{+}x{-}\tau)$};
\draw[dashed,line width=1pt,black!60] (0,1) -- (0.333,1) -- (0.333,4.353) -- (2.333,4.353);
\draw (-0.18,1) -- (-0.06,1) node[left=3pt,font=\scriptsize]{$1$};
\node[font=\scriptsize,black!60,anchor=west] at (0.42,4.75) {optimal envelope $y_0$: one jump, then constant};
\node[font=\scriptsize,black!60,align=left,anchor=west] at (0.48,2.85)
  {head of width $\delta=1/b$ at level $1$,\\ then $y_0=C_b(x)/\alpha$};
\node[font=\small] at (1.25,-3.6) {(b) the continuum relaxation and its optimizer};
\end{scope}
\end{tikzpicture}
\caption{The object minimized in Lemma~\ref{lem:headclosed}, for
$(b,e)=(3,4)$: the discrete weights of the canonical prefix (a), and
the continuum trapezoid with the optimal one-jump envelope (b).}
\label{fig:prefix}
\end{figure}
\FloatBarrier

\begin{lemma}[finite-head envelope]\label{lem:headclosed}
For integers $b,e\ge1$, put $x=e/b$ and
\[
C_b(x):=1-\frac1b+\frac x2,
\qquad
\psi_\alpha(C):=
\begin{cases}
C, & C\le\alpha,\\[2pt]
\alpha\bigl(1+\ln(C/\alpha)\bigr), & C\ge\alpha .
\end{cases}
\]
Then
\[
\frac{E^{\inf}_{b,e}}e\;\ge\;K_b(x),
\qquad
K_b(x):=\frac1b+\psi_\alpha\!\bigl(C_b(x)\bigr),
\qquad
J_b^\#(x):=\frac hx+K_b(x).
\]
Here $K_b$ is the exact optimum of the finite-head continuum
relaxation of the envelope charge; $J_b^\#$ includes the base
S-payment $hb/e=h/x$.  Only this one-sided bound enters the proof.
\end{lemma}

In words: after the mandatory head step at level $1$, the continuum
problem asks how to serve the remaining trapezoid of weight, of total
mass $C_b(x)$ per unit of debt, with an envelope that may only grow.
If the mass is light ($C_b(x)\le\alpha$), the best envelope never
rises at all and every unit is served at level $1$, at total price
$C_b(x)$; if it is heavy, the envelope jumps at once to the constant
level $C_b(x)/\alpha$ and stays there, and the price grows only
logarithmically.  The proof makes this optimality rigorous by a
supporting-line argument for the convex exponential.

\begin{proof}
Write $m=b+e$ and $W_j=\min(e,m-j)$ for $j=0,\ldots,m-1$.
Thus $W_j=e$ on the constant block $j<b$ and $W_j=m-j$ on the
decay block.  Map step $j$ to
$[j/b,(j+1)/b)$ and the discrete envelope to the step path
$y(\tau)=b_{\lfloor b\tau\rfloor}/b$.  On $[0,1+x]$ put
\[
W(\tau)=\min(x,1+x-\tau).
\]
Let $\delta=1/b$ and $T=1+x$.
The discrete envelope charge dominates the corresponding continuum
charge step by step.  Use the right-endpoint convention for Stieltjes
atoms: a jump from $b_j/b$ to $b_{j+1}/b$ occurs at
$(j+1)/b$.  Put
$u_j=(b_{j+1}-b_j)/b_j\ge0$ and
$W_j^{\min}:=W((j+1)/b)\le W_j/b$.  The rise charges satisfy
\[
\alpha bW_j^{\min}\ln(1+u_j)
\le \alpha W_j u_j,
\]
using $\ln(1+u_j)\le u_j$.  The scaled continuum dwell on step $j$ is
\[
b\int_{j/b}^{(j+1)/b}\frac{W(\tau)}{y(\tau)}\,d\tau
=\begin{cases}
W_j/b_j,&j<b,\\[2pt]
(W_j-\tfrac12)/b_j,&j\ge b.
\end{cases}
\]
Thus a constant-block dwell is represented exactly, while on a decay
step the midpoint loss is precisely the correction $-1/(2b_j)$ in
\eqref{eq:thirtytwo}.  At the last endpoint $W(T)=0$, so a final rise
only makes the discrete charge larger.

Let $\mathcal Y_{\rm ind}$ be the induced step paths and let
$\mathcal Y_\delta$ be all finite-valued, non-decreasing,
right-continuous paths
$y\ge1$ with $y=1$ on $[0,\delta)$.  For each discrete envelope, its
value dominates $e$ times the continuum functional $\mathcal F$ of its
induced path.  Write
\[
\mathcal F[y]:=\frac1x\left(
 \alpha\int_0^{1+x}W\,d(\ln y)
 +\int_0^{1+x}\frac{W}{y}\,d\tau\right).
\]
Taking infima, and then enlarging the path class, gives
\[
\frac{E^{\inf}_{b,e}}e
\ge\inf_{y\in\mathcal Y_{\rm ind}}\mathcal F[y]
\ge\inf_{y\in\mathcal Y_\delta}\mathcal F[y].
\]

Every path in $\mathcal Y_\delta$ is $1$ on
$[0,\delta)$, so that head step contributes exactly $1/b$.  On the
remaining interval,
\[
\frac1x\int_\delta^T W(\tau)\,d\tau
=1-\frac1b+\frac x2=C_b(x).
\]
It remains to minimize
\[
I[y]=\frac1x\left(
 \alpha\int_{[\delta,T]}W\,d(\ln y)
 +\int_\delta^T\frac{W}{y}\,d\tau\right).
\]
Put $u=\ln y$, take $u(\delta^-)=0$, and let $\nu=du$, a finite
nonnegative measure on $[\delta,T]$.  Thus
$u(t)=\nu([\delta,t])$.  Write $C=C_b(x)$ and
\[
k:=\max\{0,\ln(C/\alpha)\}.
\]
The candidate measure is $\nu_0=k\,\delta_\delta$: equivalently,
$u_0=k$ and $y_0=e^k$ throughout $[\delta,T]$, after one possible
jump at $\delta$.

We prove global optimality directly, without a tangent-cone argument.
The supporting-line inequality for the convex function $e^{-u}$ is
\[
e^{-u}\ge e^{-k}-e^{-k}(u-k).
\]
For $\mu:=\nu-\nu_0$, Tonelli's theorem (applied first to the positive
and negative parts of the finite signed measure $\mu$) gives
\begin{align*}
x\bigl(I[u]-I[u_0]\bigr)
&\ge \alpha\int_{[\delta,T]}W\,d\mu
 -e^{-k}\int_\delta^T W(t)\,\mu([\delta,t])\,dt\\
&=\int_{[\delta,T]}D(s)\,d\mu(s),
\end{align*}
where
\[
D(s):=\alpha W(s)-e^{-k}\int_s^T W(t)\,dt.
\]

If $C\ge\alpha$, then $e^{-k}=\alpha/C$ and
$D(\delta)=0$, because
$\int_\delta^T W=xC$ and $W(\delta)=x$.  On the constant part
$[\delta,1]$, $D$ is increasing.  On the tail, writing $r=T-s$,
\[
D(s)=r\left(\alpha-e^{-k}\frac r2\right)\ge0,
\]
since $r\le x\le2C$.  Now $\mu$ can be signed only at $\delta$,
where $D=0$; on $(\delta,T]$ it is nonnegative because the candidate
has no mass there.  Hence $\int D\,d\mu\ge0$.

If $C\le\alpha$, then $k=0$ and $\nu_0=0$, so $\mu=\nu\ge0$.
Here $D(\delta)=x(\alpha-C)\ge0$, $D$ again increases on the
constant part, and the tail formula is nonnegative because
$r\le x\le2C\le2\alpha$.  Thus $I[u]\ge I[u_0]$ in both cases.

The candidate value is $C$ when $C\le\alpha$.  When $C\ge\alpha$ it
is
$\alpha k+e^{-k}C=\alpha(1+\ln(C/\alpha))$.
Adding the already separated head dwell $1/b$ proves the formula for
$K_b$; the two branches agree at $C=\alpha$.
\end{proof}

Combining the repaired aggregation, the full S-payment, and the
finite-head envelope now gives the single inequality used by the rest
of the proof.  For a candidate rate $\varrho$, define
\[
\Xi_b(x):=e\bigl(J_b^\#(x)-\varrho\bigr)=bx\bigl(J_b^\#(x)-\varrho\bigr).
\]
Lemmas~\ref{lem:masterpayment} and \ref{lem:headclosed} yield
\begin{equation}\label{eq:masterineq}
\boxed{
\operatorname{payment}_q-\varrho e
\ge \Xi_b(x)-\frac12+h(L-b).}
\end{equation}
Thus $\Xi_b(x)$ measures the queue's payment surplus over the target
rate: whenever $\Xi_b(x)>\tfrac12$, the pair $(b,e)$ survives the
repair's charge with no further help, and every extra block step
beyond the mandatory $b$ adds another $h$ of slack.

Letting the head width tend to zero recovers the continuum endgame in
closed form.

\begin{lemma}[continuum closed form]\label{lem:jclosed}
For $x>0$,
\[
J_{\rm cont}(x)
:=\lim_{b\to\infty}J_b^\#(x)
=\frac hx+\alpha\left(1+\ln\frac{1+x/2}{\alpha}\right).
\]
It is the infimum of the continuum envelope functional over the
relaxed $\delta=0$ class represented by finite nonnegative Stieltjes
measures on $[0,1+x]$, where an atom at the origin is allowed.  The
optimizer in this closure jumps at the origin to
$y_0=(1+x/2)/\alpha$ and stays constant.
\end{lemma}

\begin{proof}
The formula is the $b\to\infty$ limit in
Lemma~\ref{lem:headclosed}.  Equivalently, repeat its supporting-line
and measure proof with $\delta=0$ and $C=1+x/2$.  The jump at the
origin lies in the closure of the class with $y(0)=1$; ramps over
$[0,\varepsilon]$ approach it without changing the limiting value.
\end{proof}

\begin{figure}[!t]
\centering
\includegraphics[width=0.55\linewidth]{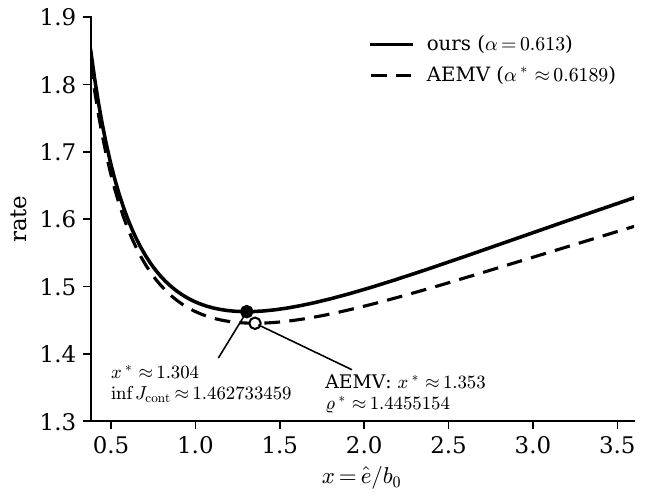}
\caption{$J_{\rm cont}(x)$ at our $\alpha=613/1000$ (solid) against
AEMV's $g_\alpha(x)$ at their $\alpha^*\approx0.6189$ (dashed,
Table~\ref{tab:constants}).}
\label{fig:jcont}
\end{figure}
\FloatBarrier

The derivative is
\[
J_{\rm cont}'(x)=
\frac{\alpha x^2-hx-2h}{x^2(2+x)}.
\]
Its numerator has one positive root,
\[
x^*=\frac{h+\sqrt{h^2+8\alpha h}}{2\alpha}
=1.30428\ldots,
\]
so $x^*$ is the unique global minimizer.  At our exact parameter
$\alpha=613/1000$,
\[
\inf_{x>0}J_{\rm cont}(x)
=J_{\rm cont}(x^*)
=1.4627334591698\ldots
>\hat\varrho=1.4627334.
\]
The strict $5.91\cdot10^{-8}$ separation is certified by outward
interval arithmetic in
\certpath{certificates/upper_bound/closed_form_certificate.py}.
Figure~\ref{fig:jcont} plots $J_{\rm cont}$ against AEMV's
$g_\alpha$: both curves are payment rates, so the higher minimum is
the better bound; ours stays above $\hat\varrho=1.4627334$, while
AEMV's extracted rate is $\rhostar\approx1.4455154$.

For completeness we also re-certify AEMV's published rate after the
repair at the exact decimal pair
\[
\check\alpha^*:=0.6189060016726951,
\qquad
\check\varrho^*:=1.4455154033861448.
\]
A separate, non-load-bearing 60-digit numerical recomputation places
the first within $4\cdot10^{-17}$ of their optimizer.  The second is a
conservative exact decimal below the recomputed rate.  The same
closed-form certificate proves
$\check\varrho^*<\inf_xJ_{\rm cont}(x;\check\alpha^*)$.

It remains only to absorb the repair's absolute $1/2$.  The next lemma
does so on almost the whole parameter grid; Figure~\ref{fig:casemap}
shows how its four-part case analysis, together with
Lemma~\ref{lem:basewindows}, covers every integer pair $(b,e)$.

\begin{figure}[!t]
\centering
\begin{tikzpicture}[x=0.9cm,y=0.68cm,
  key/.style={font=\scriptsize,align=left,anchor=west},
  swatch/.style={draw=black!40,minimum width=0.34cm,minimum height=0.24cm,inner sep=0pt}]
\fill[black!5]  (2.55,0.5) rectangle (6.55,5.5);   
\fill[black!13] (0.55,1.55) rectangle (1.45,5.5);  
\fill[black!13] (1.55,1.55) rectangle (2.45,5.5);  
\fill[black!30] (0.55,0.5) rectangle (1.45,1.45);  
\fill[black!30] (1.55,0.5) rectangle (2.45,1.45);  
\foreach \b in {1,...,6} \foreach \e in {1,...,5}
  \fill[black!60] (\b,\e) circle (1.1pt);
\foreach \e in {1,...,5} \node[black!45,font=\scriptsize] at (6.95,\e) {$\cdots$};
\foreach \b in {1,...,6} \node[black!45,font=\scriptsize] at (\b,5.85) {$\vdots$};
\draw[->] (0.2,0) -- (7.4,0) node[right=1pt,font=\scriptsize]{head level $b$};
\draw[->] (0.2,0) -- (0.2,6.4) node[rotate=90,anchor=south east,font=\scriptsize,yshift=2pt]{debt $e$};
\foreach \b in {1,...,6} \node[font=\scriptsize] at (\b,-0.42) {$\b$};
\foreach \e in {1,...,5} \node[font=\scriptsize] at (-0.12,\e) {$\e$};
\node[swatch,fill=black!5]  (sw1) at (8.35,5.0) {};
\node[key,text width=6.1cm] at (8.65,5.0)
  {$b\ge3$: $\Xi_b$ is monotone in $b$, so everything reduces to $b=3$;
   there, convexity and two tangents on $[0.62,0.72]$ give $\Xi_3>\tfrac12$};
\node[swatch,fill=black!13] (sw2) at (8.35,3.35) {};
\node[key,text width=6.1cm] at (8.65,3.35)
  {$b\in\{1,2\}$, $e\ge2$: convexity plus one endpoint value and
   derivative, at $x=2$ ($b=1$) or $x=1$ ($b=2$)};
\node[swatch,fill=black!30] (sw3) at (8.35,1.55) {};
\node[key,text width=6.1cm] at (8.65,1.55)
  {$(1,1)$ and $(2,1)$: here $\Xi_b\le\tfrac12$; one spare
   transmission's $h$ closes the gap when $L>b$, and the base windows
   ($L=b$) go to Lemma~\ref{lem:basewindows}};
\end{tikzpicture}
\caption{How Lemmas~\ref{lem:reserve} and~\ref{lem:basewindows} cover
the integer parameter grid $(b,e)=(b_0,\hat e_q)$.}
\label{fig:casemap}
\end{figure}
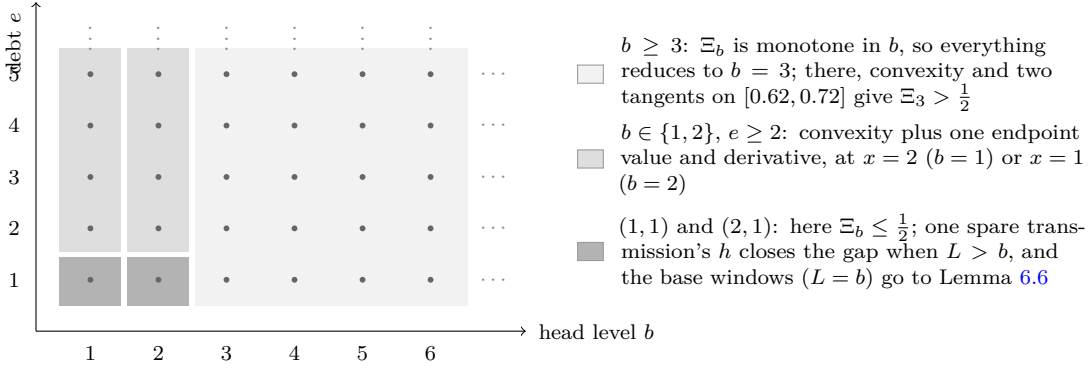
\FloatBarrier

\begin{lemma}[uniform finite-head reserve]\label{lem:reserve}
For either exact row
\[
(\alpha,\varrho)=(613/1000,\hat\varrho)
\quad\text{or}\quad
(\alpha,\varrho)=(\check\alpha^*,\check\varrho^*),
\]
take $\beta=1-\sqrt{1-\alpha}-\alpha/2$, $h=1-\alpha-\beta$, and for
integers $b,e\ge1$ with $x=e/b$, use $\Xi_b(x)$ as defined immediately
before Lemma~\ref{lem:jclosed}.  Lemma~\ref{lem:headclosed} shows that
$\Xi_b$ lower-bounds $hb+E^{\inf}_{b,e}-e\varrho$; the repair costs
$\tfrac12$ per queue, which is why $\tfrac12$ is the target.
For every pair $(b,e)\ne(1,1),(2,1)$,
$\Xi_b(x)>1/2$.  At each of the two exceptional pairs,
$\Xi_b(x)+h>1/2$; the extra $h$ is one further transmission's payment,
available only when the window has $L>b$, so the shortest windows
($L=b$, namely $(1,1,1)$ and $(2,1,2)$) are left to
Lemma~\ref{lem:basewindows}.
\end{lemma}

\begin{proof}
The proof has four parts: monotonicity in $b$ at fixed $x$ reduces
every $b\ge3$ to $b=3$; convexity of $\Xi_3$ with a two-tangent bound
closes $b=3$; the same convexity closes the $b=1$ and $b=2$ rays from
one endpoint value and derivative each; and the two exceptional pairs
are checked directly, the extra transmission's $h$ recovering
$\Xi_b+h>1/2$ there.

Put $A=1+x/2$, $C=A-1/b$, and
$\mu(x)=J_{\rm cont}(x)-\varrho\ge0$.  For $b\ge3$ we have
$C\ge2/3>\alpha$, so the logarithmic branch of
Lemma~\ref{lem:headclosed} applies.  Treating $b$ as a real variable at
fixed $x$, and putting $u=1/(bC)$ so that $A=C(1+u)$, differentiation
starts from the identity
\[
\Xi_b(x)=bx\mu(x)+x-\alpha bx\ln(1+u)
\]
and gives
\[
\frac{\partial\Xi_b(x)}{\partial b}
=x\mu(x)+\alpha x\bigl(u-\ln(1+u)\bigr)\ge0.
\]
Thus every $b\ge3$ reduces pointwise to $b=3$, even though $3x$ need
not be an integer.  Moreover
\[
\frac{\partial^2\Xi_b(x)}{\partial x^2}=
\frac{\alpha b\,[\,4(1-1/b)+x\,]}{4C_b(x)^2}>0,
\]
so $\Xi_3$ is strictly convex.

At both parameter rows, interval evaluation gives
\[
\Xi_3'(0.62)<0<\Xi_3'(0.72).
\]
The minimizer therefore lies in $[0.62,0.72]$, a bracket of width
$0.10$.  A convex function lies above its two endpoint tangents; each,
extended across the bracket to the opposite endpoint, gives
\[
\min_{x>0}\Xi_3(x)\ge
\max\left\{
\Xi_3(0.62)+0.10\,\Xi_3'(0.62),
\Xi_3(0.72)-0.10\,\Xi_3'(0.72)
\right\}.
\]
Outward interval arithmetic gives lower bounds $0.5097$ for our row
and $0.5269$ for AEMV's, both strictly above $1/2$.

On the rays used next, $C_b(x)>\alpha$, so the same logarithmic formula
and convexity apply.  For $b=1$, integrality and $e\ge2$ give $x\ge2$; convexity together
with $\Xi_1'(2)>1.06$ and $\Xi_1(2)>1.21$ handles the whole ray.
For $b=2$, $e\ge2$ gives $x\ge1$; here
$\Xi_2'(1)>0.513$ and $\Xi_2(1)>0.531$.  These common lower bounds hold
at both parameter rows.  Thus only $(1,1)$ and $(2,1)$ remain.
Direct interval evaluation gives, respectively,
\[
\begin{array}{c|cc}
& (1,1)&(2,1)\\ \hline
\text{our row}&0.3528598&0.4051002\\
\text{AEMV row}&0.3623597&0.4080449
\end{array}
\quad\text{for }\Xi_b(x),
\]
while $h=0.3155932\ldots$ and $0.3078751\ldots$ in the two rows.
Hence adding one further transmission makes $\Xi_b+h>1/2$ in all four
cases.  All numerical leaves of this proof are enclosed by
\certpath{certificates/upper_bound/fixed_b_certificate.py}; no integer-pair enumeration is used.
\end{proof}

Only the two shortest windows remain: $(b,e,L)=(1,1,1)$ and $(2,1,2)$,
the \emph{base windows}, where $L=b$ leaves no spare transmission to
pair with Lemma~\ref{lem:reserve}'s second conclusion.  They admit
elementary direct proofs from the unaggregated per-phase provisions.

\begin{lemma}[the two base windows]\label{lem:basewindows}
At either parameter row, a queue with $(b,e,L)=(1,1,1)$ or
$(2,1,2)$ satisfies $\operatorname{payment}_q>\varrho e$.
\end{lemma}

\begin{proof}
For levels $a,c\ge1$, write the unit-weight transition term as
\[
\tau(a,c)=
\begin{cases}
\alpha(c-a)/a,&c\ge a,\\
1/a-1/c,&c<a.
\end{cases}
\]
This is the retained rise charge or paired drop credit used in
the proof of Lemma~\ref{lem:amortizedagg}.  Make the step reduction
explicit.  When $e=1$, the pending window has steps
$t_q,t_q+1,\ldots,t_q+L$.  If $t_q+j$ lies in phase $i$, set
$s_j=\sigma_{q,i}-1$; thus a phase level is repeated over its steps,
and $s_j=s_{j+1}$ unless a phase boundary lies between them.  Every
pending debt weight is one.  At an internal boundary, \eqref{eq:aemv29}
supplies $\tau(s_j,s_{j+1})$ for a rise, while pairing
\eqref{eq:aemv30}'s drop penalty with the target dwell supplies the
same term for a drop.  There is no term after $j=L$: a terminal drop
penalty is multiplied by zero next-boundary debt, and omitting a
terminal rise can only lower the bound.

The constant-debt block $j=0,\ldots,L-1$ contains no correction step,
because LQD still stores $q$ there.  The only possible correction is
the decay step $j=L$; subtracting $1/(2s_L)$ even when its phase is not
a drop can only lower the floor.  Consequently the unaggregated facts
give
\[
V_L(s)=\sum_{j=0}^{L}\left(\frac1{s_j}
-\frac{\mathbf 1[j=L]}{2s_j}\right)
+\sum_{j=0}^{L-1}\tau(s_j,s_{j+1}),
\qquad
\operatorname{payment}_q\ge hL+V_L(s),
\]
where $s_0\in(b-1,b]$ and every $s_j\ge1$.  A rise phase dominates
the repeated-step dwell form because its start-of-phase debt is no
smaller than every debt inside the phase.  Thus $V_L$ is a direct lower
bound on the unaggregated L-increase, while $hL$ is the disjoint
S-payment from the $L$ block transmissions.

For $(b,e,L)=(1,1,1)$, the initial level is forced to be $s_0=1$.
The direct path value is
\[
1+\frac1{2s_1}+\alpha(s_1-1),\qquad s_1\ge1.
\]
It is increasing because $\alpha\ge1/2$, so its minimum is $3/2$.
Including the S-payment, payment is at least $h+3/2$, exceeding
$\varrho$ by $0.3528$ in our row and $0.3623$ in AEMV's row.

For $(b,e,L)=(2,1,2)$, $s_0\in(1,2]$ and the direct path value is
\[
V=\frac1{s_0}+\frac1{s_1}+\frac1{2s_2}
+\tau(s_0,s_1)+\tau(s_1,s_2).
\]
The final dwell and transition satisfy
\[
\frac1{2s_2}+\tau(s_1,s_2)\ge\frac1{s_1}-\frac12:
\]
for a drop this follows from $s_2\ge1$, and for a rise the left side
is minimized at $s_2=s_1$ because $\alpha\ge1/2$, and is even larger.
If $s_1<s_0$, this
gives $V\ge3/s_0-1/2\ge1$.  If $s_1\ge s_0$, then AM--GM gives
\[
V\ge\frac1{s_0}-\alpha-\frac12
+\frac2{s_1}+\frac{\alpha s_1}{s_0}
\ge\frac1{s_0}-\alpha-\frac12
+2\sqrt{\frac{2\alpha}{s_0}}
\ge2\sqrt\alpha-\alpha,
\]
because the preceding expression is decreasing in $s_0\in(1,2]$ and
is therefore minimized at $s_0=2$.
Since $2\sqrt\alpha-\alpha<1$ on our parameter range, both branches
give $V\ge2\sqrt\alpha-\alpha$.  Thus payment is at least
$2h+2\sqrt\alpha-\alpha$, exceeding $\varrho$ by $0.1213$ and
$0.1247$ in the two rows.  The radical inequalities are also checked
outwardly in \certpath{certificates/upper_bound/fixed_b_certificate.py}.
\end{proof}

\begin{proof}[Proof of Theorem~\ref{thm:improved}]
Apply the master inequality \eqref{eq:masterineq}.
For every pair other than $(1,1),(2,1)$, Lemma~\ref{lem:reserve}
gives $\Xi_b>1/2$.  At either exceptional pair, the same lemma handles
every $L>b$ using the extra $h(L-b)$; when $L=b$, apply
Lemma~\ref{lem:basewindows}.  Hence every queue receives payment at
least $\varrho\hat e_q$.

Both exact rows satisfy the imported hypotheses:
$\alpha\in[1/2,2/3]$, $\alpha+\beta<1$, and Lemma~16's lower bound on
$\beta$.  AEMV's unchanged Section~2 reduction therefore gives
\[
\CR(\LQD)\le1+1/\check\varrho^*\le1.6917948
\quad\text{and}\quad
\CR(\LQD)\le1+1/\hat\varrho\le1.683652,
\]
respectively.  Thus the repaired journal row also restores the weaker
conference guarantee $1.707$.

For the exact additive term in the theorem, sum the payments in our
row.  The mapping gives
$\sum_q\hat e_q\ge\OPT-\LQD$, while the telescoped feasibility bound
from Section~\ref{sec:aemvdefs} gives total payment at most
$\LQD+\alpha B$.  Hence
\[
\hat\varrho(\OPT-\LQD)
\le \hat\varrho\sum_q\hat e_q
\le \sum_q\operatorname{payment}_q
\le \LQD+\alpha B,
\]
which rearranges to the theorem's stated
$\OPT\le(1+1/\hat\varrho)\LQD+(\alpha/\hat\varrho)B$.
\end{proof}

\section{Conclusion and future work}\label{sec:conclusion}

In this work, we have narrowed the certified gap for LQD to
\[
\CR(\LQD)\;\in\;[\,1.46929591,\ 1.683652\,],
\]
down from the previously published $[1.44546086,\,1.6918]$.  We
compute the lower bound in exact integer arithmetic
(Section~\ref{sec:lower}), independently of the repaired aggregation
step.  The upper bound is a certified decimal that interval
arithmetic places just below a closed-form value at an algebraic
minimizer; a second small interval certificate verifies the uniform
finite-head reserve (Section~\ref{sec:upper}).  It rests, together with the restoration
of the previously published $1.6918$, on the repaired aggregation
step (Section~\ref{sec:erratum}), without which even $1.6918$ is
unproven as written.

The direction we find most promising for immediate future work is the \emph{universal} lower
bound, that is, the best ratio provable against \emph{every} deterministic
online policy, not only LQD.  \citet{BDGN19}'s staggered-birth
adversary certifies $\sqrt2$ there (Section~\ref{sec:related}), but its
linear-programming form reaches only $1.3274$ on the uniform family
$\Phi_k$, which we believe is a limitation of the family, not of the method.  The
front-loaded family introduced here is the natural candidate to drive
that program past $\sqrt2$: it is the ingredient that lifted
the LQD-specific bound, so feeding it into the universal-adversary
formulation is, we hope, the most likely route to progress there.

Closing the LQD-specific gap $[\,1.46929591,\,1.683652\,]$ looks harder
at both ends.  On the lower side, the front-loaded family appears close
to the best a construction of this kind can certify: within it the
value moves only in later decimals, so we expect further progress to
require a qualitatively new construction rather than a refinement of
this one.  On the upper side, for each fixed $\alpha$ we solve the
continuum envelope relaxation exactly; Theorem~\ref{thm:improved} uses
the nearby rational choice $\alpha=613/1000$, and further improvements
may require coupling the per-queue analysis across the shared
buffer.  The payment scheme charges each queue's total
(Section~\ref{sec:aemvdefs}) against its own outstanding OPT-extra
credit $\hat e_q$; currently each queue must clear this charge
\emph{on its own}, i.e.\ reach payment $\ge\varrho\,\hat e_q$.  One possible next
step is a \emph{funding layer} relaxation that instead lets queues pool
their payment surplus across the shared buffer, requiring only the
buffer-wide total to clear $\varrho\sum_q\hat e_q$; then a queue running
short could draw on another queue's surplus.  We leave the formulation
and certification of such a coupled relaxation to future work.

We also note our certificate-first approach of Section~\ref{sec:lower}: replacing an asymptotic construction with a single checkable instance. We believe that a similar approach may be useful for related buffer management problems and leave this research program for future work.

In one sentence: the certified gap for LQD now reads
$\CR(\LQD)\in[\,1.46929591,\ 1.683652\,]$, and every number in that
sentence can be recomputed from the certificates shipped with this
paper.


\appendix
\section{Certificates and reproducibility}\label{sec:appendix-certs}

We have made the code for every certificate referred to in this paper
openly available at
\url{https://doi.org/10.5281/zenodo.21260154}.  The directory
\certpath{certificates/} contains runnable producers and recorded outputs,
not only the numerical answers.  Install the pinned Python dependencies
with
\begin{center}
\texttt{python3 -m pip install -r certificates/requirements.txt}
\end{center}
and run the complete claim map with
\begin{center}
\texttt{python3 certificates/verify\_numbers.py}.
\end{center}
Every critical check uses an explicit failure branch and therefore
also runs under \texttt{python3 -O}.  The full run is dominated by the
$k=3\cdot10^5$ exact-integer row of Table~\ref{tab:frontladder} and
typically takes a few minutes on one core.  Rebuilding that
row additionally requires a POSIX shell and a C++17 compiler.  Paths
in the table below are relative to \certpath{certificates/}; paths in
the surrounding prose and commands are relative to the repository root.

\begin{table}[!t]
\centering\small
\begin{tabular}{
  >{\raggedright\arraybackslash}p{0.30\linewidth}
  >{\raggedright\arraybackslash}p{0.38\linewidth}
  >{\raggedright\arraybackslash}p{0.22\linewidth}}
\toprule
Claim & Artifact & Trust base \\
\midrule
Def.~\ref{def:frontfamily}: the 310-knot profile $F^{\front}$ &
\certpath{data/front_loaded_profile_cdf.json} & data \\
Def.~\ref{def:frontfamily}: instance generator ($\ell_j$, $B$) &
\certpath{lower_bound/instance_generator.py} & exact/float \\
Table~\ref{tab:frontladder}: the six certified rows, tie rule $T_0$ &
\certpath{cpp/} (self-contained) + \certpath{lower_bound/run_all_rows.py} &
exact integers \\
Section~\ref{sec:thegap}: the symbolic countermodel &
\certpath{erratum/aggregation_counterexample.py} & symbolic (sympy) \\
Section~\ref{sec:thegap}: the realized-instance witness ($n=2$ pattern
on an explicit run) &
\certpath{erratum/realized_counterexample.py} & exact rationals; brute-forced
offline optimum \\
Lemma~\ref{lem:jclosed}: $\inf J_{\rm cont}>\hat\varrho$,
$\CR\le1+1/\hat\varrho\le1.683652$;
$\varrho<\inf J_{\rm cont}$ at both exact parameter rows &
\certpath{upper_bound/closed_form_certificate.py} &
60-digit interval arithmetic \\
Lemma~\ref{lem:reserve} and Lemma~\ref{lem:basewindows}: the two
tangent bounds, the $b=1,2$ endpoint checks, the extra-S reserve, and
the two direct base-window inequalities &
\certpath{upper_bound/fixed_b_certificate.py} &
60-digit interval arithmetic \\
Table~\ref{tab:constants}: AEMV's recomputed constants &
\certpath{upper_bound/g_alpha_recompute.py} & 60-digit floating point \\
\bottomrule
\end{tabular}
\caption{Certificate-backed claims, their verifying scripts, and the
arithmetic each verification trusts; \texttt{certificates/README.md}
gives the full claim map.}
\label{tab:certificates}
\end{table}

The lower-bound simulations use checked unsigned 64-bit integer
arithmetic, with all certified values well below $2^{64}$; their exact
integer buffer sizes and reported totals are compared against Python's
unbounded integers.  The two erratum scripts use exact symbolic or
rational arithmetic.  For the
upper bound, \certpath{upper_bound/closed_form_certificate.py} encloses the algebraic
minimizer and the logarithm in Lemma~\ref{lem:jclosed}, at both exact
parameter rows.  The companion
\certpath{upper_bound/fixed_b_certificate.py} then checks only the scalar leaves of the
analytic proof of Lemma~\ref{lem:reserve}: the two endpoint tangent
values for the convex function $\Xi_3$, the $b=1,2$ endpoint values and
derivatives, $\Xi+h>1/2$ at the two exceptional pairs, and the radical
base-window inequalities of Lemma~\ref{lem:basewindows}.  It performs
no search over integer pairs, paths, grids, or block lengths.

The repository retains earlier region, grid, and extended-length
programs as archived cross-checks, but no theorem in this revision
depends on them.  The critical/legacy distinction and the precise
producer--consumer map are documented in
\certpath{certificates/README.md}.  An independent Python engine that
agrees with the C++ engine on the exact $\OPT$ and $\LQD$ totals at
$k=400$ (both under $T_0$) is also kept as a noncritical cross-check; run
under the opposite high-index-drop rule it returns the pinned, strictly
larger value $\LQD=1{,}806{,}170$, a
concrete witness of the tie-rule sensitivity that
Theorem~\ref{thm:tieinv} controls.

Taken together, every certificate-backed numerical claim in this
paper can be reproduced, and every certificate re-verified, by a
single script run in \certpath{certificates/}.

\section{Deferred proofs: tie-rule invariance}\label{sec:appendix-tie}

In this appendix, we prove Theorem~\ref{thm:tieinv}: the certified
lower bound holds for every non-clairvoyant tie rule.  The proof rests
on three tools: an offline normal form (Lemma~\ref{lem:normalform}, which
also justifies the ``live service plus hoards'' shape of the fluid
model of Section~\ref{sec:lower}); a localization of the tie-rule
freedom to the overflow margins (Proposition~\ref{prop:tiefreedom});
and an exchange move (Lemma~\ref{lem:exchange}, illustrated in
Figure~\ref{fig:exchange}).  With the three tools in place, the proof
of the theorem couples the online and offline runs.

\begin{lemma}[offline normal form]\label{lem:normalform}
On a saturated interval instance (Definition~\ref{def:vocab}), every
offline schedule can be replaced by one of equal value, with pointwise
no larger buffer occupancy, in \emph{admit-late} form:
\begin{enumerate}[label=(\alph*)]
\item each packet a queue transmits while live is admitted in the slot
of that transmission, and each packet it transmits after death is
admitted at its death slot $\delta_q$ (where the queue receives its
final $B$ arrivals, so the whole post-death hoard, at most $B$
packets, fits);
\item consequently, at the end of the arrival phase of every slot
\emph{strictly before} $\delta_q$, a queue holds at most one packet,
and from $\delta_q$ on it holds exactly its remaining transmissions: an
optimal offline policy never needs to accumulate packets in a queue
before its death slot;
\item the normal form \emph{never evicts} a stored packet (every
packet it admits, it transmits), so a corpse holding $H\ge1$ packets
at the end of a slot transmits one packet in each of the next $H$
slots; in particular, at any slot after $\delta_q$ it transmits from
$q$ if and only if $q$'s content there is nonzero.
\end{enumerate}
\end{lemma}

\begin{proof}
Fix an offline schedule $S$; discard (without loss of value) any packet
$S$ admits but never transmits.  Build $S'$ by re-timing admissions:
a packet that $S$ transmits from queue $q$ at slot $s$ while $q$ is
live is admitted by $S'$ at slot $s$ itself. This is legal, since $q$
receives $B\ge1$ fresh packets at every live slot (saturation) and the
arrival phase precedes transmission; a packet that $S$ transmits from
$q$ after $q$'s death is admitted by $S'$ at the death slot $\delta_q$,
where $q$ receives $B$ packets.  The hoard fits: no packet arrives at $q$
after $\delta_q$'s arrival phase, so \emph{every} schedule transmitting
from $q$ at $\delta_q$ or later ($S$ among them) must already store
all of them at $\delta_q$, whence the buffer cap bounds their total (the
death-slot transmission together with the post-death hoard) by $B$.  $S'$ transmits
the same number of packets from each queue in the same slots (the
post-death transmissions are automatic, hence occupy consecutive
slots under $S'$ just as under $S$), so its value equals $S$'s.
Occupancy: strictly before its death slot, a queue holds a packet under
$S'$ only in a slot where it transmits, and then holds $\ge1$ under $S$
too; from its death slot on, $S'$ holds exactly the queue's remaining
transmissions, and since no arrivals reach the queue after death every
schedule ($S$ included) must already store those packets.  So
$S'$'s occupancy is pointwise $\le S$'s.
Finally, $S'$ never evicts: it admits exactly the packets it
transmits, each admitted at or before its transmission slot, so its
queues drain by transmissions alone.  Since transmission is automatic
(every nonempty queue transmits each slot,
Definition~\ref{def:switch}), a corpse under $S'$, holding exactly
its remaining transmissions, transmits in precisely its next $H$
slots.
\end{proof}

\begin{proposition}[one-phase tie-rule freedom, localized]\label{prop:tiefreedom}
Fix the queue-level vector immediately before the evictions of one
arrival phase, and fix the number of packets that must be evicted.  The
multiset of levels after LQD finishes those evictions is independent of
the tie choices made \emph{within that phase}.  Equivalently, two runs
that enter the phase through a bijection matching equal levels and then
receive matched current arrivals finish the eviction phase with the same
level multiset.

The only remaining freedom is over queue \emph{identities}, and it
is confined to the terminal margin: for some cutoff $w$, the
last-touched tied group ends with $R\ge0$ queues at level $w+1$ and
the rest at $w$, and its $R$ marginal packets are the
\emph{residuals}.  Comparing two runs matched by the entry bijection,
the terminal level multisets agree, so the high--low disagreements at
the margin balance: every queue at $w+1$ in one run whose partner sits
at $w$ can be paired with a queue having the opposite mismatch.
Re-pairing each such pair restores a level-matching bijection.  If the
two re-paired queues also have identical remaining arrival schedules
(in particular, if both are corpses, or if their future live intervals
coincide), the re-pairing preserves the schedule matching as well.
\end{proposition}

\begin{proof}
On the level multiset, one eviction replaces one occurrence of the
current maximum $M$ by $M-1$.  This multiset operation is independent of
which queue realizes that occurrence.  Induction over the required
evictions proves the first claim.  Repeatedly decrementing the maximum
drains the top of the multiset down to a flat cutoff $w$ with $R$
queues left one above it (the shape of Remark~\ref{rem:tierule}).
Every eviction occurs at the running maximum of the multiset, which is
non-increasing and never falls below the final maximum $w+1$; hence
every evicted queue ends the phase at $w$ or $w+1$, and a queue whose
level ever sits at $w$ or below is never touched, so its terminal
level and identity are the same in both runs.  Only the $R$ residual
assignments at the margin remain free.
Because both runs have the same number of queues at every terminal
level, high--low mismatches balance and can be paired.  Crossing the
partners in each pair matches equal levels, and when their remaining
arrival schedules agree, crossing preserves that agreement.
\end{proof}

The exchange move is the tool for the moment when the coupling must
swap a live queue with a corpse.  The live queue $L$'s final arrival
batch rebuilds the corpse's hoard inside $L$, the stale corpse $C$ is
dropped down to the one packet it may still have to send, and the
offline bijection is crossed; transmissions and occupancies match the
reference schedule's exactly, so the tracking invariant survives
(Figure~\ref{fig:exchange}).

\begin{lemma}[the exchange move]\label{lem:exchange}
Let $S$ be an offline schedule in the normal form of
Lemma~\ref{lem:normalform} for an instance $I$, and let $\tilde S$ be a
schedule for an instance $\tilde I$ that has tracked $S$ through a
bijection $\kappa$ of queues up to slot $t-1$: at the end of slot $t-1$
every queue of $\tilde I$ stores exactly as many packets as its
$\kappa$-partner does under $S$.  Suppose that in $\tilde I$ queue $L$ is
live at slot $t$ (receiving $B$ packets there) with its arrivals
ending at $t$, queue $C$ is dead,
$\kappa(L)$ is live under $S$ at $t$ \emph{with arrivals continuing past
$t$} (so, in normal form, $\kappa(L)$ ends slot $t$ empty), and $\kappa(C)$
is dead holding $H$ packets.  

\emph{Notation.} Write $z_L,z_C\in\{0,1\}$ for the
transmissions $S$ makes from $\kappa(L),\kappa(C)$ at slot $t$; by
clause (c) of Lemma~\ref{lem:normalform},
$$z_C=\min\{H,1\}\quad\text{and}\quad H'=H-z_C.$$
Then $\tilde S$ can act
within slot $t$ so that: the pair $\{L,C\}$ transmits exactly
$z_L+z_C$ packets at $t$; its occupancy at the end of the arrival
phase is exactly $z_L+H$, matching that of
$\{\kappa(L),\kappa(C)\}$ under $S$; and at the end of slot $t$, queue $L$
holds $H'$ and $C$ holds $0$, exactly the slot-$t$-end contents of
$\kappa(C)$ and $\kappa(L)$ under $S$.  Hence $\tilde S$ has tracked $S$
through the \emph{crossed} bijection $\kappa'=\kappa\circ(L\,C)$, with the
same transmission count and the same end-of-phase buffer occupancy at
slot $t$.
\end{lemma}

\begin{proof}
Transmissions in the model are automatic (every nonempty queue
transmits one packet per slot, Definition~\ref{def:switch}), so
$\tilde S$ chooses only drops and admissions, and the move must leave
each queue holding, at the transmission phase, a content whose
\emph{forced} transmission produces the stated end state. There are three cases; in each, the
drop is legal (drops always are), the admission fits into $L$'s $B$
fresh slot-$t$ arrivals, and the pair's occupancy at the end of the
arrival phase equals $z_L+H$, which is exactly what $S$ stores in
$\kappa(L),\kappa(C)$ at that point: $\kappa(L)$ holds $z_L$ (in normal
form, exactly the packet it sends at $t$, if any) and $\kappa(C)$ holds
$H$. So, since the move touches no other queue (and simultaneous moves
within one slot act on disjoint pairs, each matching its own
$S$-counterpart's occupancy and transmissions), feasibility of the
buffer bound is inherited from $S$.

\emph{Case $z_L=1$.}  Drop $C$ from $H$ to $z_C$ and admit
$1+H'$ packets into $L$; this fits, since
$1+H'\le1+H=z_L+H\le B$ ($S$ stores $z_L+H$ packets in its pair
alone, within a buffer of size $B$).  At the transmission phase $L$
holds $1+H'\ge1$ and sends one packet ($=z_L$), and $C$ holds and
sends $z_C$ (if $z_C=1$ then $H\ge1$, so the packet is there).
End state: $L$ holds $H'$, $C$ holds $0$; the pair sent
$1+z_C=z_L+z_C$.

\emph{Case $z_L=0$, $H'=0$.}  Here $H=z_C$.  Admit nothing into
$L$ and drop nothing ($C$ already holds $z_C$).  $L$ is empty and
sends nothing ($=z_L$); $C$ sends its $z_C$ and ends empty.  End
state: $L$ holds $0=H'$, $C$ holds $0$; occupancy $z_C=z_L+H$.

\emph{Case $z_L=0$, $H'\ge1$.}  Then $H\ge H'\ge1$, so
$z_C=\min\{H,1\}=1$ and $H=H'+1\ge2$.  Drop $C$ from $H$ to $0$ and
admit all of $H$ into $L$ ($H\le B$, as $S$ stores $\kappa(C)$'s $H$
packets in its buffer).  At the transmission phase $L$ holds $H\ge1$
and \emph{must} send one packet (this forced transmission is why a
literal ``admit $H'$ into $L$ and transmit nothing from it'' is not
available), while $C$, now empty, sends nothing.  End state: $L$
holds $H-1=H'$, $C$ holds $0$; the pair sent $1=z_L+z_C$ packets,
and its arrival-phase occupancy was $H+0=z_L+H$.

In all three cases the end-of-slot contents are exactly those of the
crossed partners: with arrivals past $t$, $\kappa(L)$ ends slot $t$ empty
(Lemma~\ref{lem:normalform}), and $\kappa(C)$ ends it holding
$H-z_C=H'$.
\end{proof}

\begin{figure}[!t]
\centering
\begin{tikzpicture}[
  sq/.style={draw,fill=black!12,minimum size=0.26cm,inner sep=0pt},
  qlab/.style={font=\scriptsize},
  rowlab/.style={font=\scriptsize\itshape,anchor=east},
  kap/.style={dashed,black!55}]
\def\stack#1#2#3{
  \foreach \i in {1,...,#3}{\node[sq] at (#1,#2+0.28*\i-0.14) {};}}
\begin{scope}
\node[font=\scriptsize\bfseries] at (1.5,3.1) {end of slot $t{-}1$};
\node[rowlab] at (-0.4,2.3) {under $S$:};
\draw (0.7,2.0) -- (1.1,2.0); \node[qlab] at (0.9,1.72) {$\kappa(L)$};
\node[qlab,black!55] at (0.9,2.18) {$0$};
\draw (2.0,2.0) -- (2.4,2.0); \node[qlab] at (2.2,1.72) {$\kappa(C)$};
\stack{2.2}{2.0}{3}
\node[qlab,black!55,anchor=west] at (2.45,2.42) {$H$};
\node[rowlab] at (-0.4,0.3) {under $\tilde S$:};
\draw (0.7,0.0) -- (1.1,0.0); \node[qlab] at (0.9,-0.28) {$L$};
\node[qlab,black!55] at (0.9,0.18) {$0$};
\draw (2.0,0.0) -- (2.4,0.0); \node[qlab] at (2.2,-0.28) {$C$};
\stack{2.2}{0.0}{3}
\draw[kap] (0.9,1.62) -- (0.9,0.42);
\draw[kap] (2.2,1.62) -- (2.2,1.05);
\node[font=\scriptsize,black!55] at (1.55,1.05) {$\kappa$};
\end{scope}
\begin{scope}[xshift=3.85cm]
\draw[->,black!70] (0.15,1.15) -- (2.75,1.15);
\node[font=\scriptsize,align=center,black!70,text width=2.9cm] at (1.45,2.55)
  {slot $t$, under $S$: $\kappa(L)$ receives and sends $z_L{=}1$; $\kappa(C)$ drains one ($z_C{=}1$)};
\node[font=\scriptsize,align=center,black!70,text width=2.9cm] at (1.45,-0.35)
  {slot $t$, under $\tilde S$: $L$ admits $1{+}H'$ of its final $B$ arrivals, sends one; $C$ is dropped to $z_C$, sends it, ends empty};
\end{scope}
\begin{scope}[xshift=7.9cm]
\node[font=\scriptsize\bfseries] at (1.5,3.1) {end of slot $t$};
\draw (0.7,2.0) -- (1.1,2.0); \node[qlab] at (0.9,1.72) {$\kappa(L)$};
\node[qlab,black!55] at (0.9,2.18) {$0$};
\draw (2.0,2.0) -- (2.4,2.0); \node[qlab] at (2.2,1.72) {$\kappa(C)$};
\stack{2.2}{2.0}{2}
\node[qlab,black!55,anchor=west] at (2.45,2.28) {$H'$};
\draw (0.7,0.0) -- (1.1,0.0); \node[qlab] at (0.9,-0.28) {$L$};
\stack{0.9}{0.0}{2}
\draw (2.0,0.0) -- (2.4,0.0); \node[qlab] at (2.2,-0.28) {$C$};
\node[qlab,black!55] at (2.2,0.18) {$0$};
\draw[kap] (0.9,1.62) -- (2.2,0.42);
\draw[kap] (2.2,1.75) -- (0.9,0.75);
\node[font=\scriptsize,black!55,fill=white,inner sep=1pt] at (1.55,1.09) {$\kappa'=\kappa\circ(L\,C)$};
\end{scope}
\end{tikzpicture}
\caption{The exchange move of Lemma~\ref{lem:exchange} in its main
case ($z_L=1$, drawn with $H=3$): transmissions and occupancies match
on both sides, and the coupling continues with the crossed bijection
$\kappa'$.}
\label{fig:exchange}
\end{figure}
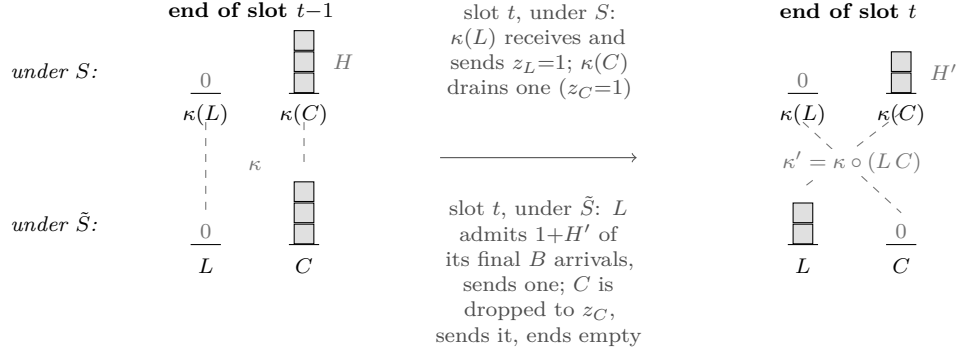
\FloatBarrier

\begin{proof}[Proof of Theorem~\ref{thm:tieinv}]
We construct $I_T$ slot by slot while running $\LQD_T$ on it, and
compare against the fixed reference run of $\LQD_{T_{\mathrm{ref}}}$
on $I$.  The
proof maintains \emph{two} bijections between the queues of $I_T$ and
those of $I$, updated at different events; keeping them separate is
needed for the corpse cases.  (We write $\omega$ for the
online coupling and $\kappa$ for the offline one; $\sigma$, $\varrho$,
and $\tau_i$ are reserved for the AEMV import of Section~\ref{sec:aemvdefs}.)

\emph{The online coupling ($\omega$).}  Invariant $(\star)$: at the end
of every slot, $\omega$ matches queues of equal $\LQD$-buffer level,
and $\omega$-partners have identical remaining arrival schedules ($B$
packets per slot over the same set of future slots; in particular,
equal liveness).
The arrivals of $I_T$ at each slot are \emph{defined} through
$\omega$: the queue matched by $\omega$ to each live reference queue
receives the corresponding arrival batch of $I$ (of size $B$), and
when a cohort of fresh queues
is born in $I$, $I_T$ receives fresh queues with identical schedules,
extending $\omega$ by the identity on them.  Since $T$ is non-clairvoyant, its slot-$t$
choices are made before any arrival after $t$ is fixed, so this
construction is well-founded.

Given $(\star)$ at the start of a slot, both runs receive
$\omega$-matched arrivals, so the eviction processes start from equal
level multisets; each eviction decrements a maximal element, so the
post-phase \emph{level multisets} agree
(Proposition~\ref{prop:tiefreedom}) whatever $T$ chooses.  The two
states can therefore disagree only on \emph{which} queues occupy the
margin levels $w,w{+}1$, and the disagreements split into balanced
mismatch pairs: since the multisets agree, $I_T$ and $I$ hold equally
many queues at level $w{+}1$ (and equally many at $w$), so each queue
that is high in $I_T$ where its $\omega$-partner is low is matched by
one that is low where its partner is high.  Repair each pair by its
type:
\begin{enumerate}\itemsep1pt
\item[(a)] \emph{both live at $t$, at least one receiving past $t$}:
swap the two queues' remaining arrival schedules (both are ours to
assign) and re-pair $\omega$ so levels align;
\item[(b)] \emph{neither receiving past $t$} (corpses, and queues at
their death slot: at the slot's end the two are indistinguishable):
re-pair $\omega$ only; remaining schedules are empty on both sides,
so $(\star)$'s schedule clause holds trivially;
\item[(c)] \emph{one receiving past $t$, the other dead}:
perform \emph{surgery}: truncate the live queue
$L$'s schedule at the current slot $t$ (it is dead from $t{+}1$),
redirect its remaining schedule to the corpse $C$ (it is live,
\emph{resurrected}, from $t{+}1$), and re-pair $\omega$ crosswise,
$C\mapsto\omega(L)$, $L\mapsto\omega(C)$.  Levels align in either
orientation of the mismatch, and the redirected schedule is $B$ per
slot up to $\omega(L)$'s death slot, as $(\star)$ requires.
\end{enumerate}
No mismatched queue is unborn: an unborn queue and its
$\omega$-partner both hold level $0$ and receive nothing through slot
$t$, so their levels cannot disagree; hence every mismatched queue
that receives past $t$ is live at $t$ and received its $B$-packet
batch there.  With this, the three types are exhaustive: if neither
queue receives past $t$ the pair is type~(b); if both do, type~(a);
and with exactly one, the pair is type~(a) when the other queue is
dying at $t$ (still live, so the swap applies) and type~(c) when it is
a corpse.
Repairs act on disjoint pairs and restore $(\star)$ at the slot's end.
By induction $(\star)$ always holds; in each slot both buffers have
equal numbers of nonempty queues, so the transmission counts agree slot by
slot, and $\LQD_T(I_T)=\LQD_{T_{\mathrm{ref}}}(I)$.  (For randomized $T$, run the
construction on each realization.)

\emph{The offline mimic ($\kappa$).}  Let $S$ be an optimal schedule for
$I$, normalized as in Lemma~\ref{lem:normalform}.  We build a feasible
schedule $S_T$ for $I_T$ with value$(S_T)=$ value$(S)$, tracking $S$
through a second bijection $\kappa$ (initially the identity, and
extended by the identity on each fresh cohort, in step with $\omega$)
with the invariant $(\dagger)$: at every slot end, each queue of $I_T$ stores
exactly the $S$-content of its $\kappa$-partner, and $\kappa$-partners have
identical remaining arrival schedules.  Away from repairs,
$S_T$ copies $S$'s decisions through $\kappa$, which $(\dagger)$ makes
feasible.  At the repairs:
\begin{enumerate}\itemsep1pt
\item[(a)] \emph{both queues live at $t$}: each starts the slot empty
(its former $\kappa$-partner, live at $t$, ended slot $t{-}1$ empty in
normal form) and receives its full $B$ arrivals.  $S_T$ re-pairs $\kappa$
alongside $\omega$ and admits into each queue everything its new
partner stores at the end of slot $t$'s arrival phase (its
slot-$t$-end content plus the packet it transmits at $t$, if any: at
most $B$, since the partner alone stores it), then transmits as the
partner does; whether the partner ends the slot empty (arrivals
continuing) or holds a death-slot hoard, its content is rebuilt from
fresh arrivals, so nothing scarce is spent;
\item[(b)] \emph{neither queue receives after $t$}: $\kappa$ is not
updated, since from $t{+}1$ both only drain, so their stored contents
cannot be exchanged, and need not be: the online coupling's
re-pairing of $\omega$ at such mismatches is only an accounting device for the
level multiset and imposes nothing on the offline player; this is the
one point where $\omega$ and $\kappa$ diverge, and nothing requires them
to agree;
\item[(c)] \emph{one queue continues past $t$, the other is a corpse}:
this is Lemma~\ref{lem:exchange} verbatim, with
$H$ the $S$-content of $\kappa(C)$: $L$ receives past $t$ by the type of
the pair, so by $(\dagger)$'s schedule clause $\kappa(L)$'s arrivals
likewise continue past $t$, as the lemma requires; the surgery's timing
gives $L$ its final batch of $B$ arrivals at exactly the slot where the
hoard must be installed, the corpse's stale content is dropped rather
than moved, and $\kappa$ is crossed.  Note that $\kappa(C)$ is dead under $S$
whenever needed: $\kappa$ pairs live queues with live partners by
$(\dagger)$'s schedule clause, and $C$ was dead in $I_T$ with a dead
$\kappa$-partner from the moment both died.
\end{enumerate}
In every case the per-slot transmission counts of $S_T$ equal $S$'s,
and the end-of-phase occupancies agree, so $S_T$ is feasible on $I_T$
and value$(S_T)=$ value$(S)=\OPT(I)$.  Hence
$\OPT(I_T)\ge\OPT(I)$.  Combining the two halves and taking the
$R$-copy limit described in the theorem statement completes the proof
under Definition~\ref{def:switch}.
\end{proof}

\begin{remark}[scope]\hypertarget{rem:tiescope-target}{}\label{rem:tiescope}
\begin{enumerate}[label=(\roman*)]
\item $I_T$ is generally not interval-structured (a resurrected queue
has two arrival intervals), which is immaterial: Definition~\ref{def:switch}
ranges over arbitrary instances, and
interval structure is needed only for the fast exact evaluation of
\LateQD{} on the certified instance itself (Remark~\ref{rem:nqueues});
\item non-clairvoyance is what the
construction uses, but it does not use determinism; for a randomized tie rule the
instance adapts to the realized choices, so the guarantee is the
adaptive-adversary one (the input may depend on the coins already
flipped, as opposed to the oblivious convention where the instance is
fixed before any coin); on a \emph{fixed} instance a different rule can
do strictly better: the opposite high-index-drop rule, which leaves
residuals on lower indices and thus favors corpses,
attains a higher LQD value on the certified instance ($1{,}806{,}170$
against $T_0$'s $1{,}802{,}645$ at $k=400$, both computed by the
cross-check engine); so the adaptive form is the right general statement;
\item the computational certificate itself remains stated for $T_0$: the theorem
transfers the certified value, not the certificate.
\end{enumerate}
\end{remark}

\bibliographystyle{plainnat}
\bibliography{bib}

\end{document}